\documentclass[%
 jmp,
]{revtex4-1}

\usepackage{extract}
\usepackage{fancyhdr}
\usepackage[utf8]{inputenc}
\usepackage{amsmath, amsfonts, amssymb, amstext, amsfonts}
\usepackage{graphicx}
\usepackage{multirow}
\usepackage{wrapfig}
\usepackage{epsfig}
\usepackage{framed}
\usepackage{textcomp}
\usepackage{dsfont}
\usepackage{tikz}
\usepackage[ruled,vlined]{algorithm2e}
\usepackage{booktabs}

\usepackage{graphicx}
\usepackage{dcolumn}
\usepackage{bm}
\usepackage{comment}
 \usepackage{amsthm}

 \usepackage{mathtools}
 \usepackage{color}
 \usepackage{enumitem}
 \usepackage{verbatim}
 \usepackage[
	colorlinks=true,
	urlcolor=blue,
	linkcolor=green
]{hyperref}

\newtheorem{innercustomgeneric}{\customgenericname}
\providecommand{\customgenericname}{}
\newcommand{\newcustomtheorem}[2]{%
  \newenvironment{#1}[1]
  {%
   \renewcommand\customgenericname{#2}%
   \renewcommand\theinnercustomgeneric{##1}%
   \innercustomgeneric
  }
  {\endinnercustomgeneric}
}

\newcustomtheorem{res}{Result}

 \newtheorem*{thm*}{Theorem}
 \newtheorem{thm}{Theorem}[section]
 \newtheorem{cor}[thm]{Corollary}
 \newtheorem{lem}[thm]{Lemma}

 \theoremstyle{definition}
 \newtheorem{defn}[thm]{Definition}
 \theoremstyle{remark}
 \newtheorem{rem}[thm]{Remark}
 \newtheorem*{ex}{Example}
 
 \newcommand{\sq}[1]{`#1'}

\usepackage[utf8]{inputenc}
\usepackage[T1]{fontenc}
\usepackage{mathptmx}

\usepackage[cal=cm]{mathalfa}

\newcommand{\idmat}{\mathds{1}}
\newcommand{\unitary}{\mathcal{U}}
\newcommand{\idvec}{\boldsymbol{1}}
\newcommand{\choi}{\mathfrak{C}}
\newcommand{\flip}{\mathbb{F}}
\newcommand{\idop}{\mathrm{id}}
\newcommand{\blt}{\mathcal{B}}

\newcommand{\trcl}{\mathcal{S}_1}

\newcommand{\norm}[1]{\left\lVert #1 \right\rVert}
\newcommand{\abs}[1]{\left\vert #1 \right\vert}
\newcommand{\tr}[1]{\mathrm{tr}\left[#1\right]} 
\newcommand{\ptr}[2]{\mathrm{tr}_{#1}\left[#2\right]}
\newcommand{\set}[2]{\left\{ #1 \,\middle|\, #2 \right\}}
\newcommand{\comu}[2]{\left[ #1\, ,\, #2 \right]}
\newcommand{\acomu}[2]{\left\{ #1 \, , \, #2 \right\}}
\newcommand{\R}{\mathbb{R}}

\newcommand{\C}{\mathbb{C}}
\newcommand{\N}{\mathbb{N}}

\DeclarePairedDelimiter\bra{\langle}{\rvert}
\DeclarePairedDelimiter\ket{\lvert}{\rangle}
\DeclarePairedDelimiterX\braket[2]{\langle}{\rangle}{#1 \delimsize\vert #2}

\renewcommand{\theequation}{\arabic{equation}}

\begin{document}

\preprint{AIP/123-QED}

\title[Sample title]{Quantum and classical dynamical semigroups of superchannels and semicausal channels}

\author{Markus Hasenöhrl}
 \email{m.hasenoehrl@tum.de}
\author{Matthias C.~Caro}%
 \email{caro@ma.tum.de}
\affiliation{ 
Department of Mathematics, Technical University of Munich, Garching, Germany\\ 
Munich Center for Quantum Science and Technology (MCQST), Munich, Germany
}%

\date{\today}

\begin{abstract}
\noindent 
Quantum devices are subject to natural decay. 
We propose to study these decay processes as the Markovian evolution of quantum channels, which leads us to dynamical semigroups of superchannels.
A superchannel is a linear map that maps quantum channels to quantum channels, while satisfying suitable consistency relations.
If the input and output quantum channels act on the same space, then
we can consider dynamical semigroups of superchannels.
No useful constructive characterization of the generators of such semigroups is known.
\noindent We characterize these generators in two ways: First, we give an efficiently checkable criterion for whether a given map generates a dynamical semigroup of superchannels.
Second, we identify a normal form for the generators of semigroups of quantum superchannels, analogous to the GKLS form in the case of quantum channels.
\noindent To derive the normal form, we exploit the relation between superchannels and semicausal completely positive maps, reducing the problem to finding a normal form for the generators of semigroups of semicausal completely positive maps. We derive a normal for these generators using a novel
technique, which applies also to infinite-dimensional systems.    
\noindent Our work paves the way to a thorough investigation of semigroups of superchannels: Numerical studies become feasible because admissible generators can now be explicitly generated and checked. And analytic properties of the corresponding evolution equations are now accessible via
our normal form. 
\end{abstract}

\maketitle

\section{Introduction and Motivation}\label{Sct:Introduction}

Anybody who has ever owned an electronic device knows: These devices have a finite lifespan after which they stop working properly. At least from a consumer perspective, a long lifespan is a desirable property for such devices. Thus, it is important for an engineer to know which kind of decay processes can affect a device, in order to suppress them by an appropriate design. 
Certainly, these considerations will also become important for the design of quantum devices. We therefore propose to study systematically the decay processes that quantum devices can be subject to. 

In this work, we take a first step in this direction by deriving the general form of linear time-homogeneous master equations that govern how quantum channels behave when inserted into a circuit board at different points in time. This leads to the study of dynamical semigroups of superchannels. Here, superchannels are linear transformations between quantum channels~\cite{Chiribella_2008}. 

\begin{figure}[!ht]
    \centering
    \includegraphics[scale = 0.45]{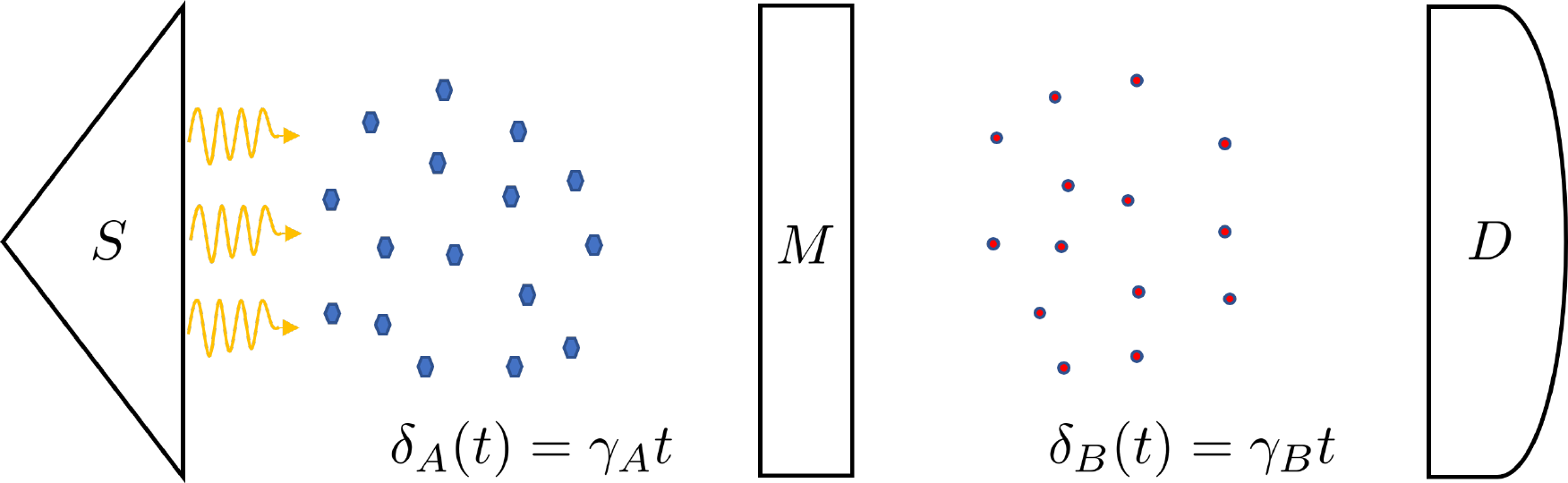}
    \caption{Estimating the transmissivity of a material under the influence of an influx of particles into the regions between the components.}
    \label{fig:initialExample}
\end{figure}

Let us consider a concrete example, see Fig.~\ref{fig:initialExample}. Suppose we are trying to estimate the optical transmissivity of some material ($M$), which we assume to depend on the polarization of the incident light. A simple approach is to send photons from a light source ($S$) through the material and to count how many photons arrive a the detector ($D$). We model the material by a quantum channel $T_M$, acting on the states of photons described as three-level systems, with the levels corresponding to vacuum, horizontal, and vertical polarization. In an idealized world, with a perfect vacuum in the regions between the source, the material, and the detector, we can infer the transmissivity from the measurement statistics of the state $T_M(\sigma)$, where $\sigma$ is the state of the photon emitted from the source. However, in a more realistic scenario, even though we might have created an (almost) perfect vacuum between the devices at construction time, some particles are leaked into that region over time. Then, interactions between the photons and these particles might occur, causing absorption or a change in polarization. Hence, the situation is no longer described accurately by $T_M$ alone, but also requires a description of the particle-filled regions. 

To find such a description, we us argue that the effect of particles in some region (here, either between $S$ and $M$; or $M$ and $D$) can be modeled by a quantum dynamical semigroup, parametrized by the particle density $\delta$. 
If the particle density is reasonably low and $Q_\delta$ is the quantum channel describing the effect of the particles on the incident light at a given $\delta$, then, as explained in Fig. \ref{fig:InitialExampleMarkovianity}, $Q_\delta$ satisfies the semigroup property $Q_{\delta_1 + \delta_2} = Q_{\delta_1} \circ Q_{\delta_2}$. 
\begin{figure}[!ht]
    \centering
    \includegraphics[scale = 0.35]{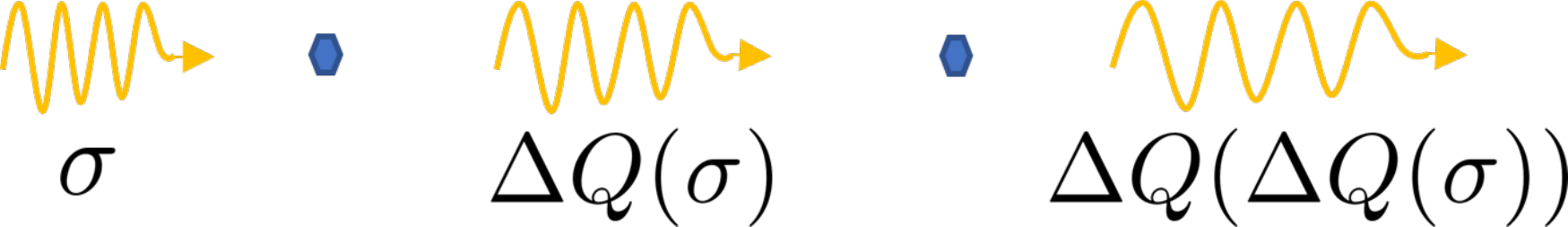}
    \caption{If the particle density is low, then the incident photon interacts with the particles in the region sequentially and independently. The effect of a single interaction can be described by a channel $\Delta Q$. Hence, the state after the first interaction is $\Delta Q(\sigma)$, the state after the second  interaction is $\Delta Q (\Delta Q(\sigma))$, and so forth. The number of interactions is given by the product of the particle density $\delta$ and the volume $V$. Hence, the effect of an region with fixed volume is described by the channel $Q_\delta = (\Delta Q)^{\delta V}$. It follows that if $\delta = \delta_1 + \delta_2$, then $Q_{\delta_1 + \delta_2} = (\Delta Q)^{\delta_1 V} (\Delta Q)^{\delta_2 V} = Q_{\delta_1} \circ Q_{\delta_2}$. The semigroup property for real $\delta$ can then be obtained in the continuum limit.}
    \label{fig:InitialExampleMarkovianity}
\end{figure}
Furthermore, if there are no particles then there should be no effect. Hence, $Q_0 = \idop$. After adding continuity in the parameter $\delta$ as a further natural assumption, the family $\{Q_\delta\}_{\delta \geq 0}$ forms a quantum dynamical semigroup. That is, we can write $Q_\delta = e^{L\delta}$, for some generator $L$ in GKLS-form. 

If we assume in our example that particles of type $A$ are leaked into the region between $S$ and $M$ at a rate $\gamma_A$ and that particles of type $B$ are leaked into the region between $M$ and $D$ at a rate $\gamma_B$, then the overall channel describing the transformation that emitted photons undergo at time $t$ is given by
\begin{align*}
    \hat{S}_t(T_M) = e^{\gamma_B L_B t} \circ T_M \circ e^{\gamma_A L_A t},
\end{align*}
where $L_A$ and $L_B$ are the generators of the dynamical semigroups describing the effect of the particles in the respective regions. 

We note that at any fixed time, $\hat{S}_t$ interpreted as a map on quantum channels is a superchannel written in `circuit'-form. This means, that $\hat{S}_t$ describes a transformation of quantum channels implemented via pre- and post-processing. Furthermore, $\hat{S}_t(T_M)$ can be determined by solving the time-homogenous master equation
\begin{align*}
    \frac{d}{dt} T(t) = \hat{L}(T(t)),
\end{align*}
where $\hat{L}(T) = \gamma_AL_A \circ T + \gamma_B T \circ L_B$, with initial condition $T(0) = T_M$. In other words, we have
\begin{align*}
    \hat{S}_t = e^{\hat{L}t}
\end{align*} 
and thus the family $\{\hat{S}_t\}_{t \geq 0}$ forms a dynamical semigroup of superchannels.

By inductive reasoning, we thus arrive at our central physical hypothesis: Decay-processes of quantum devices with some sort of influx are well described by dynamical semigroups of superchannels. It follows that such decay-processes can be understood by characterizing dynamical semigroups of superchannels. Such a characterization is the main goal of our work. 

In particular, we aim to understand dynamical semigroups of superchannels in terms of their generators. We characterize these generators fully by providing two results: First, we give an efficiently checkable criterion for whether a given map generates a dynamical semigroup of superchannels.
Second, we identify a normal form for the generators of semigroups of quantum superchannels, analogous to the GKLS form in the case of quantum channels. Interestingly, we find that the most general form of dynamical semigroups of superchannels goes beyond the simple introductory example above. 

We arrive at these results through a path (see Fig.~\ref{fig:FlowDiagramConcepts}) that also illuminates the connection to the classical case. We start by studying dynamical semigroups of classical superchannels, which (analogously to quantum superchannels being transformations between quantum channels) are transformations between stochastic matrices. We do so by establishing a one-to one correspondence between classical superchannels and certain classical semicausal channels, that is, stochastic matrices on a bipartite system ($AB$) that do not allow for communication from $B$ to $A$ (see Definition~\ref{Defn:ClassicalSemicausality}). We can then obtain a full characterization of the generators of semigroups of classical superchannels by characterizing generators of semigroups of classical semicausal maps first and then translating the results back to the level of superchannels. The study of (dynamical semigroups of) classical superchannels and classical semicausal channels is the content of Section~\ref{sct:classical}.

\begin{figure}[!ht]
    \centering
    \includegraphics[scale = 0.5]{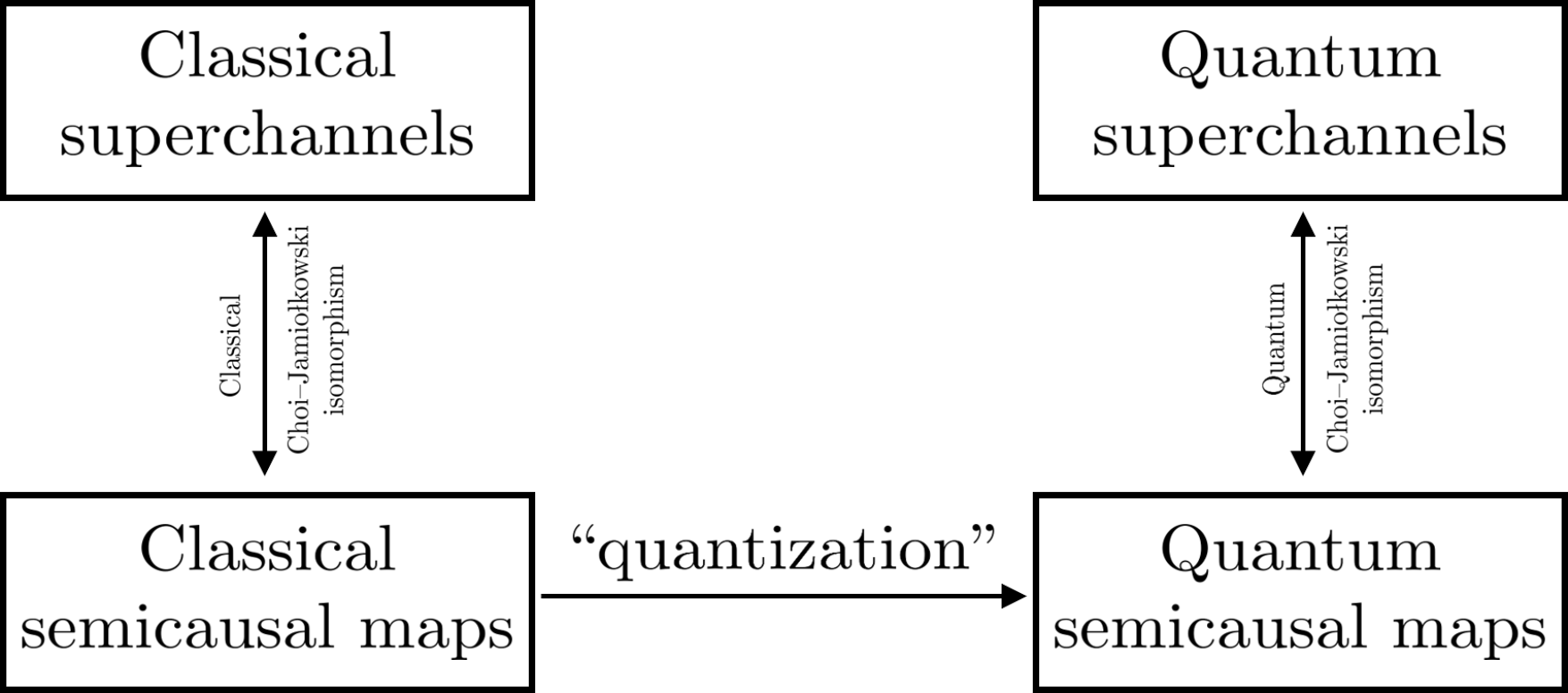}
    \caption{Schematic of the concepts studied in this work.}
    \label{fig:FlowDiagramConcepts}
\end{figure}

Armed with the intuition obtained from the classical case, we then go on to study the quantum case. We start by characterizing the generators of semigroups of semicausal~\cite{Beckman.2001} completely positive maps (CP-maps) -- our main technical result, and one of independent interest. This characterization can be obtained from the classical case by a `quantization'-procedure that allows us to see exactly which features of semigroups of semicausal CP-maps are ``fully quantum.'' Dynamical semigroups of semicausal CP-maps are discussed Section~\ref{subsec:SemicausalCPSemigroups}.  Finally, in~Section \ref{subsec:QuantumSuperchannels}, we use the one-to one correspondence (via the quantum Choi–Jamiołkowski isomorphism) between certain semicausal CP-maps and quantum superchannels to obtain a full characterization of the generators of semigroups of quantum superchannels. While the classical section (\ref{sct:classical}) and the quantum section (\ref{sct:quantum}) are heuristically related, they are logically independent and can be read independently.

This work is structured as follows. In the remainder of this section, we discuss results related to ours. Section~\ref{sct:results} contains an overview over our main results.
In Section~\ref{sct:prelims}, we recall relevant notions from functional analysis and quantum information, as well as some notation. The (logically) independent sections~\ref{sct:classical} and~\ref{sct:quantum} comprise the main body of our paper, containing complete statements and proofs of our results on dynamical semigroups of superchannels and semicausal channels. We study the classical case in Section~\ref{sct:classical} and the quantum case in Section~\ref{sct:quantum}. 
Finally, we conclude with a summary and an outlook to future research in Section~\ref{sct:conclusion}.

\subsection{Related work}

The study of quantum superchannels goes back to \cite{Chiribella_2008} and has since evolved to the study of higher-order quantum maps \cite{PhysRevLett.101.060401, PhysRevA.80.022339, doi:10.1098/rspa.2018.0706}. A peculiar feature of higher-order quantum theory is that it allows for indefinite causal order \cite{Oreshkov2012, PhysRevA.88.022318}. However, it was recently discovered that the causal order is preserved under (certain) continuous evolutions \cite{PhysRevX.8.011047, selby2020revisiting}. It therefore seems interesting to study continuous evolutions of higher-order quantum maps systematically. Our work can be seen as an initial step into his direction.

The study of (semi-)causal and (semi-)localizable quantum channels goes back to~\cite{Beckman.2001}. By proving the equivalence of semicausality and semilocalizability for quantum channels,~\cite{Eggeling.2002} resolved a conjecture raised in~\cite{Beckman.2001} (and attributed to DiVincenzo). Later,~\cite{Piani.2006} provided an alternative proof for this equivalence, and further investigated causal and local quantum operations.

\section{Results}\label{sct:results}

We give an overview over our answers to the questions identified in the previous section. In our first result, we identify a set of constraints that a linear map satisfies if and only if it generates a semigroup of quantum superchannels.

\begin{res}{1.1}[Lemma  \ref{thm:checkforGeneratorSuperchannel} - Informal] \label{Result1.1} Checking whether a linear map $\hat{L}:\blt(\blt(\mathcal{H}_A);\blt(\mathcal{H}_B))\to\blt(\blt(\mathcal{H}_A);\blt(\mathcal{H}_B))$ generates a semigroup of quantum superchannels can be phrased as a semidefinite constraint satisfaction problem. 
\end{res}

Therefore, we can efficiently check whether a given linear map is a valid generator of a semigroup of quantum superchannels. We can even solve optimization problems over such generators in terms of semidefinite programs. Thereby, this first characterization of generators of semigroups of quantum superchannels facilitates working with them computationally.

As our second result, we determine a normal form for generators of semigroups of quantum superchannels. Similar to the GKLS-form, we decompose the generator into a ``dissipative part'' and a ``Hamiltonian part,'' where the latter generates a semigroup of invertible superchannels such that the inverse is a superchannel as well.  

\begin{res}{1.2}[Theorem \ref{Thm:CharacterizationGeneratorsSuperchannels} - Informal] \label{Result1.2}
A linear map $\hat{L}:\blt(\blt(\mathcal{H}_A);\blt(\mathcal{H}_B))\to\blt(\blt(\mathcal{H}_A);\blt(\mathcal{H}_B))$ generates a semigroup of quantum superchannels if and only if it can be written as $\hat{L}(T)=\hat{D}(T) + \hat{H}(T)$, where the ``Hamiltonian part'' is of the form
\begin{align*}
    \hat{H}(T)(\rho)= -i[H_B,T(\rho)] - iT([H_A,\rho]),
\end{align*}
with local Hamiltonians $H_B$ and $H_A$, and where the ``dissipative part'' is of the form $\hat{D}(T)(\rho) = \ptr{E}{\hat{D}^\prime(T)(\rho)}$, where
\begin{subequations}
\begin{align}
    \hat{D}^\prime(T)(\rho) &= U (T \otimes \idop_E)(A(\rho \otimes \sigma)A^\dagger) U^\dagger& &- \frac{1}{2}  (T \otimes \idop_E)(\acomu{A^\dagger A} {\rho \otimes \sigma})& \label{eq:resSuCh1}\\
    &+ B (T \otimes \idop_E)(\rho \otimes \sigma) B^\dagger& &- \frac{1}{2} \acomu{B^\dagger B}{ (T \otimes \idop_E)(\rho \otimes \sigma)}& \label{eq:resSuCh2}\\
    &+ \comu{U (T \otimes \idop_E)(A(\rho \otimes \sigma))}{B^\dagger}& &+ \comu{B}{(T \otimes \idop_E)((\rho \otimes \sigma)A^\dagger)U^\dagger},& \label{eq:resSuCh3}
\end{align}
\end{subequations}
with unitary $U$ and arbitrary $A$ and $B$.
\end{res}
The ``dissipative part'' consists of three terms: Term~\eqref{eq:resSuCh1} itself generates a semigroup of superchannels (for $B = 0$), with the interpretation that the transformed channel ($\hat{S}_t(T)$) arises due to the stochastic application of $T \mapsto \ptr{E}{U (T \otimes \idop_E)(A(\rho \otimes \sigma)A^\dagger) U^\dagger}$ at different points in time (Dyson series expansion). Term~\eqref{eq:resSuCh2} itself generates a semigroup of superchannels (for $A = 0$) of the form $\hat{S}_t(T) = e^{L_B t} \circ T$, where $L_B$ is a generator of a quantum dynamical semigroup (and hence in GKLS-form). Term \eqref{eq:resSuCh3} is a ``superposition'' term, which is harder to interpret. It will become apparent from the path taken via the `quantization' of semicausal semigroups that this term is a pure quantum feature with no classical analogue. Therefore, the presence of \eqref{eq:resSuCh3} can be regarded as one of our main findings. It is also worth noting that the normal form in Result \ref{Result1.2} is more general than the form of of the generator we found in our introductory example. Hence, nature allows for more general decay-processes than the simple ones with an independent influx of particles before and after the target object. 
We also complement this structural result by an algorithm that determines the operators $U$, $A$, $B$, $H_A$ and $H_B$, if the conditions in Result \ref{Result1.1} are met. 

The proof of these results relies on the relation (via the Choi-Jamiołkowski isomorphism) between superchannels and semicausal CP-maps. Our next findings -- and from a technical standpoint our main contributions -- are the corresponding results for semigroups of semicausal CP-maps.

\begin{res}{2.1}[Lemma \ref{Lem:VerifySemicausality} - Informal] \label{Result2.1} 
Checking whether a linear map $L:\blt(\mathcal{H}_A\otimes\mathcal{H}_B)\to\blt(\mathcal{H}_A\otimes\mathcal{H}_B)$ generates a semigroup of $B \not\to A$ semicausal CP-maps can be phrased as a semidefinite constraint satisfaction problem for its Choi-matrix.  
\end{res}

Based on this insight, we can efficiently check whether a given linear map is a valid generator of a semigroup of semicausal CP-maps.

Since semigroups of semicausal CP-maps are in particular semigroups of CP-maps,  our normal form for generators giving rise to semigroups of semicausal CP-maps is a refining of the the GKLS-form. 

\begin{res}{2.2}[Theorem \ref{Thm:SemicausalInfiniteDimMainResult} - Informal] \label{Result2.2}
A linear map $L:\blt(\mathcal{H}_A\otimes\mathcal{H}_B)\to\blt(\mathcal{H}_A\otimes\mathcal{H}_B)$ generates a semigroup of $B \not\to A$ semicausal CP-maps (in the Heisenberg picture) if and only if it can be written as $L(X)=\Phi(X)-K^\dagger X - XK$, where the CP part $\Phi$ is of the form
\begin{align*}
    \Phi(X) &= V^\dagger \left(X \otimes \idmat_E \right) V, \text{ with } V = (\idmat_A \otimes U)(A \otimes \idmat_B) + (\idmat_A \otimes B), 
\end{align*}
with a unitary $U \in \blt(\mathcal{H}_E \otimes \mathcal{H}_B; \mathcal{H}_B \otimes \mathcal{H}_E)$ and arbitrary $A \in \blt(\mathcal{H}_A; \mathcal{H}_A \otimes \mathcal{H}_E)$ and $B \in \blt(\mathcal{H}_B; \mathcal{H}_B \otimes \mathcal{H}_E)$, and the $K$ in the non-CP part is of the form
\begin{align*}
    K &=  (\idmat_A \otimes B^\dagger U) (A \otimes \idmat_B) + \frac{1}{2} \idmat_A \otimes B^\dagger B + K_A \otimes \idmat_B + \idmat_A \otimes iH_B,
\end{align*}
with a self-adjoint $H_B$ and an arbitrary $K_A$.
\end{res}

This characterization has both computational and analytical implications: On the one hand, it provides a recipe for describing semicausal GKLS generators in numerical implementations. On the other hand, the constructive characterization of semicausal GKLS generators makes a more detailed analysis of their (e.g., spectral) properties tractable. It is also worth noting that in Result \ref{Result2.2} we can allow for (separable) infinite-dimensional spaces. In the finite-dimensional case, we also provide an algorithm to compute the operators $U$, $A$, $B$, $K_A$ and $H_B$, if the conditions of Result \ref{Result2.1} are met. 

Let us now turn to the corresponding results in the classical case. Here, instead of looking at (semigroups of) CP-maps and quantum channels, we look at (entry-wise) non-negative matrices and row-stochastic matrices (see Section \ref{sct:prelims} and Section \ref{sct:classical} for details) that we assume to act on $\R^\mathbb{X}$, for (finite) alphabets $\mathbb{X} \in \{\mathbb{A}, \mathbb{B}, \mathbb{E} \}$. 

The following result is the classical analogue of Result \ref{Result2.2}.

\begin{res}{3} [Corollary \ref{Corolary1ClassicalSemicausalGenerators} - Informal] \label{Result3}
A linear map $Q : \R^{\mathbb{A}} \otimes \R^{\mathbb{B}} \rightarrow \R^{\mathbb{A}} \otimes \R^{\mathbb{B}}$ generates a semigroup of (Heisenberg) $B\not\to A$ semicausal non-negative matrices if and only if it can be written as
\begin{align*}
    Q = (A \otimes \idmat_B)(\idmat_A \otimes U) - K_A\otimes \idmat_B + \sum\limits_{i=1}^{\lvert\mathbb{A}\rvert}\ket{a_i}\bra{a_i}\otimes B^{(i)},
\end{align*}
with a row-stochastic matrix $U \in \blt(\R^{\mathbb{B}}; \R^{\mathbb{E}} \otimes \R^{\mathbb{B}})$, a non-negative matrix $A \in \blt( \R^{\mathbb{A}}\otimes \R^{\mathbb{E}}; \R^{\mathbb{A}})$, a diagonal matrix $K_A$ and maps $B^{(i)} \in \blt(\R^{\mathbb{B}})$ that generate semigroups of row-stochastic matrices.  
\end{res}

We will discuss in detail how Result \ref{Result2.2} arises as the `quantization' of Result \ref{Result3} in the paragraph following the proof of Lemma \ref{Lem:VerifySemicausality}. Here, we highlight that in both the quantum and the classical case, the generators of semicausal semigroups are constructed from two basic building blocks. In the quantum case, these are a $B \not\to A$ semicausal CP-map $\Phi_{sc}$, with $\Phi_{sc}(X) = V_{sc}^\dagger (X \otimes \idmat_E) V_{sc}$ and $V_{sc} = (\idmat_A \otimes U) (A \otimes \idmat_B)$; and a GKLS generator of the form $\idop_A \otimes \hat{B}$. And in the classical case, they are a $B \not\to A$ semicausal non-negative map $\Phi_{sc} = (\idmat_A \otimes U) (A \otimes \idmat_B)$; and operators of the form $\ket{a_i}\bra{a_i} \otimes B^{(i)}$, where $B^{(i)}$ generates a semigroup of row-stochastic maps. The difference between the quantum case and the classical case then lies in the way the general form is constructed from the building blocks. While we simply take convex combinations of the building-blocks in the classical case, we have to take superpositions of the building-blocks, by which we mean that we need to combine the corresponding Strinespring operators, in the quantum case.  

As our last result, we present the normal form for generators of semigroups of classical superchannels. 

\begin{res}{4}
A linear map $\hat{Q} : \blt(\R^{\mathbb{A}}; \R^{\mathbb{B}}) \rightarrow \blt(\R^{\mathbb{A}}; \R^{\mathbb{B}})$ generates a semigroup of classical superchannels if and only if it can be written as
\begin{align*}
    \hat{Q}(M) = U (M \otimes \idmat_E) A - \sum_{i = 1}^{\abs{\mathbb{A}}} \braket{\idvec_{AE}}{A a_i} M \ket{a_i}\bra{a_i} + \sum_{i = 1}^{\abs{\mathbb{A}}} B^{(i)} M \ket{a_i}\bra{a_i},
\end{align*}
 with a column-stochastic matrix $U \in \blt(\R^\mathbb{E} \otimes \R^\mathbb{B}; \R^\mathbb{B})$, a non-negative matrix $A \in \blt(\R^\mathbb{A}; \R^\mathbb{A} \otimes \R^\mathbb{E})$, a diagonal matrix $K_A$, and a collection of generators of semigroups of column-stochastic matrices $B^{(i)} \in \blt(\R^\mathbb{B})$.
\end{res}
As in the quantum case, we have two kinds of evolutions: a stochastic application of $M \mapsto U (M \otimes \idmat_E) A$ at different points in time; and a conditioned post-processing evolution of the form $\sum_i e^{B^{(i)}t} M  \ket{a_i}\bra{a_i}$. Note that there are no ``superposition'' terms, like \eqref{eq:resSuCh3}.  

\newpage
\section{Notation and preliminaries}\label{sct:prelims}

In this section, we review basic notions from Functional Analysis, Quantum Information Theory, and the theory of dynamical semigroups. We also fix our notation for these settings as well as for a classical counterpart of the quantum setting.

\subsection{Functional analysis}

Throughout the paper, $\mathcal{H}$ (with some subscript) denotes a (in general infinite-dimensional) separable complex Hilbert space. Whenever $\mathcal{H}$ is assumed to be finite-dimensional, we explicitly state this assumption. We denote the Banach space of bounded linear operators with domain $\mathcal{H}_A$ and codomain $\mathcal{H}_B$, equipped with the operator norm, by $\blt(\mathcal{H}_A; \mathcal{H}_B)$ and write $\blt(\mathcal{H})$ for $\blt(\mathcal{H}; \mathcal{H})$. For $X \in \blt(\mathcal{H}_A; \mathcal{H}_B)$, the adjoint $X^\dagger \in \blt(\mathcal{H}_B; \mathcal{H}_A)$ of $X$ is the unique linear operator such that $\braket{\psi_B}{X \psi_A} = \braket{X^\dagger \psi_B}{\psi_A}$ for all $\ket{\psi_A} \in \mathcal{H}_A$ and all $\ket{\psi_B} \in \mathcal{H}_B$. Here, and throughout the paper, we use the standard Dirac notation.

An operator $Y \in \blt(\mathcal{H})$ is called self-adjoint if $Y^\dagger = Y$. A self-adjoint $Y\in \blt(\mathcal{H})$ is called positive semidefinite, denoted by $Y \geq 0$, if there exists an operator $Z \in \blt(\mathcal{H})$, such that $Y = Z^\dagger Z$. If $Y$ is positive semidefinite, then there exists a unique positive semidefinite operator $\sqrt{Y}$, such that $Y = \sqrt{Y}\sqrt{Y}$~\cite[p.~196]{Reed:104857}. The operator $\sqrt{Y}$ is called the square-root of $Y$. The absolute value $\abs{Y} \in \blt(\mathcal{H})$ of $Y$ is defined by $\abs{Y} = \sqrt{Y^\dagger Y}$. 

We define the set of trace-class operators $\trcl(\mathcal{H}_A; \mathcal{H}_B) = \set{\rho \in \blt(\mathcal{H}_A; \mathcal{H}_B)}{\tr{\abs{\rho}} < \infty}$, which becomes a Banach space when endowed with the norm $\norm{\rho}_1 := \tr{\abs{\rho}}$. We write $\trcl(\mathcal{H})$ for $\trcl(\mathcal{H}; \mathcal{H})$. The set $\trcl(\mathcal{H}_A; \mathcal{H}_B)$ satisfies the two-sided $^*$-ideal property: If $\rho \in \trcl(\mathcal{H}_A; \mathcal{H}_B)$ and $Y\in\blt(\mathcal{H}_A;\mathcal{H}_B)$, then $\rho^\dagger \in \trcl(\mathcal{H}_B; \mathcal{H}_A)$,  $\rho^\dagger Y \in \trcl(\mathcal{H}_A)$ and $Y \rho^\dagger \in \trcl(\mathcal{H}_B)$.

Besides the norm topology, we will use the strong operator topology and the ultraweak topology. The strong operator topology is the smallest topology on $\blt(\mathcal{H}_A; \mathcal{H}_B)$ such that for all $\ket{\psi_A} \in \mathcal{H}_A$ the map $\blt(\mathcal{H}_A; \mathcal{H}_B) \ni Y \mapsto Y\ket{\psi_A}\in\mathcal{H}_B$ is continuous, where $\mathcal{H}_B$ is equipped with the norm topology. The ultraweak topology on $\blt(\mathcal{H}_A; \mathcal{H}_B)$ is the smallest topology such that the map $\blt(\mathcal{H}_A; \mathcal{H}_B) \ni Y \mapsto \tr{\rho^\dagger Y}\in\mathbb{C}$ is continuous for all $\rho \in \trcl(\mathcal{H}_A; \mathcal{H}_B)$. Since $\mathcal{H}_A$ and $\mathcal{H}_B$ are separable, so is $\trcl(\mathcal{H}_B; \mathcal{H}_A)$. Hence, the sequential Banach Alaoglu theorem implies that every bounded sequence in $\blt(\mathcal{H}_A; \mathcal{H}_B)$ has an ultraweakly convergent subsequence. Here, we view $\blt(\mathcal{H}_A; \mathcal{H}_B)$ as the continuous dual of $\trcl(\mathcal{H}_B; \mathcal{H}_A)$. The aforementioned results can be found in many books, e.g,~\cite[ch.~VI.6]{Reed:104857}, however, usually only for the case $\mathcal{H}_A = \mathcal{H}_B$. The general results stated above can be obtained from this case by considering $\blt(\mathcal{H}_A; \mathcal{H}_B)$ and $\trcl(\mathcal{H}_A; \mathcal{H}_B)$ as subspaces of $\blt(\mathcal{H}_A \oplus \mathcal{H}_B)$ and $\trcl(\mathcal{H}_A \oplus \mathcal{H}_B)$, respectively. 

An operator $V \in \blt(\mathcal{H}_A; \mathcal{H}_B)$ is called an isometry if $\norm{V\ket{\psi_A}} = \norm{\ket{\psi_A}}$ for all $\ket{\psi_A} \in \mathcal{H}_A$. The (possibly empty) set of unitaries, the surjective isometries, is denoted by $\unitary(\mathcal{H}_A; \mathcal{H}_B)$ and we write $ \unitary(\mathcal{H})$ for $ \unitary(\mathcal{H}; \mathcal{H})$. As a special notation, if $\mathcal{H}_A^\prime$ and $\mathcal{H}_B^\prime$ are closed linear subspaces of $\mathcal{H}_A$ and $\mathcal{H}_B$, with (canonical) isometric embeddings $\idmat_{A^\prime \rightarrow A} \in \blt(\mathcal{H}_A^\prime; \mathcal{H}_A)$ and $\idmat_{B^\prime \rightarrow B} \in \blt(\mathcal{H}_B^\prime; \mathcal{H}_B)$, respectively, then we will write $\unitary_P(\mathcal{H}_A^\prime; \mathcal{H}_B^\prime) = \{\idmat_{B^\prime \rightarrow B} U \idmat_{A^\prime \rightarrow A}^\dagger \in \blt(\mathcal{H}_A; \mathcal{H}_B) \,|\, U \in \unitary(\mathcal{H}_A^\prime; \mathcal{H}_B^\prime)\}$ and, $\unitary_P(\mathcal{H})$ for $\unitary_P(\mathcal{H}; \mathcal{H})$. I.e., this is the set of partial isometries. 

\subsection{Flip operator, partial trace, complete positivity, and duality}
The flip operator $\flip_{A;B} \in \blt(\mathcal{H}_A \otimes \mathcal{H}_B; \mathcal{H}_B \otimes \mathcal{H}_A)$ is the unique operator satisfying $\flip_{A;B} (\ket{\psi_A} \otimes \ket{\psi_B}) = \ket{\psi_B} \otimes \ket{\psi_A}$, for all $\ket{\psi_A}\in\mathcal{H_A}$ and all $\ket{\psi_B}\in\mathcal{H}_B$.

The partial trace w.r.t.~the space $\mathcal{H}_A$ is the unique linear map $\mathrm{tr}_A : \trcl(\mathcal{H}_A \otimes \mathcal{H}_B; \mathcal{H}_A \otimes \mathcal{H}_C) \rightarrow \trcl(\mathcal{H}_B; \mathcal{H}_C)$ that satisfies $\tr{X \ptr{A}{\rho}} = \tr{(\idmat_A \otimes X)\rho}$, for all $\rho \in \trcl(\mathcal{H}_A \otimes \mathcal{H}_B)$ and all $X \in \blt(\mathcal{H}_C; \mathcal{H}_B)$. If the spaces involved have subscripts, the partial trace will always be denoted with the corresponding subscript. The partial trace with respect to $\rho \in \trcl(\mathcal{H}_A)$ is the unique linear map $\mathrm{tr}_\rho : \blt(\mathcal{H}_A \otimes \mathcal{H}_B; \mathcal{H}_A \otimes \mathcal{H}_C) \rightarrow \blt(\mathcal{H}_B; \mathcal{H}_C)$ that satisfies $\tr{\sigma \ptr{\rho}{X}} = \tr{(\rho \otimes \sigma)X}$, for all $\sigma \in \trcl(\mathcal{H}_C; \mathcal{H}_B)$ and all $X \in \blt(\mathcal{H}_A \otimes \mathcal{H}_B; \mathcal{H}_A \otimes \mathcal{H}_C)$. Proofs of existence and uniqueness can be found in~\cite[Thm. 2.28 and Thm. 2.30]{AttalPartialTrace}, where we used again the observation that the results above follow from the usual ones for $\mathcal{H}_B = \mathcal{H}_C$, by looking at operators on $\mathcal{H}_A \otimes (\mathcal{H}_B \oplus \mathcal{H}_C)$.

Let $T \in \blt(\blt(\mathcal{H}_B); \blt(\mathcal{H}_A))$. The map $T$ is called positive if $T(X_B)$ is positive semidefinite, whenever $X_B \in \blt(\mathcal{H}_B)$ is positive semidefinite. For $n \in \N_0$, the map $T_n : \blt(\C^n \otimes \mathcal{H}_B) \rightarrow \blt(\C^n \otimes \mathcal{H}_A)$ is uniquely defined by the requirement that $T_n(X_n \otimes X_B) = X_n \otimes T(X_B)$ for all $X_n \in \blt(\C^n)$ and all $X_B \in \blt(\mathcal{H}_B)$. The map $T$ is completely positive (CP) if the map $T_n$ is positive for all $n \in \N_0$. A CP-map $T$ is called normal if $T$ is continuous when $\blt(\mathcal{H}_A)$ and $\blt(\mathcal{H}_B)$ are both equipped with the ultraweak topology. We denote the set of normal CP-maps by $\mathrm{CP}_\sigma(\mathcal{H}_B; \mathcal{H}_A)$ and write $\mathrm{CP}_\sigma(\mathcal{H})$ for $\mathrm{CP}_\sigma(\mathcal{H}; \mathcal{H})$. By the Stinespring dilation theorem (in its form for normal CP-maps), $T$ is a normal CP-map if and only if there exists a (separable) Hilbert space $\mathcal{H}_E$ and an operator $V \in \blt(\mathcal{H}_A; \mathcal{H}_B \otimes \mathcal{H}_E)$ such that for all $X_B \in \blt(\mathcal{H}_B)$ we have $T(X_B) = V^\dagger (X_B \otimes \idmat_E) V$. Furthermore, the Stinespring dilation can be chosen to be minimal, that is, the pair $(V, \mathcal{H}_E)$ can be chosen such that $\mathrm{span}\{ (X_B \otimes \idmat_E) V \ket{\psi_A} \,|\, X_B \in \blt(\mathcal{H}_B), \ket{\psi_A} \in \mathcal{H}_A \}$ is norm-dense in $\mathcal{H}_B \otimes \mathcal{H}_E$. Furthermore, if $(V^\prime, \mathcal{H}_E^\prime)$ is another Stinespring dilation, then there exists an isometry $U \in \blt(\mathcal{H}_E; \mathcal{H}_E^\prime)$, such that $V^\prime = (\idmat_B \otimes U)V$. Another equivalent characterization is the so-called Kraus form: $T$ is a normal CP-map if and only if there exists a countable set of operators $\{L_i\}_i \subset \blt(\mathcal{H}_A; \mathcal{H}_B)$, the Kraus operators, such that for all $X_B \in \blt(\mathcal{H}_B)$, we have $T(X_B) = \sum_i L_i^\dagger X_B L_i$, where the series converges in the strong operator topology. One can obtain Kraus operators from a Stinespring dilation $(V, \mathcal{H}_E)$ by choosing an orthonormal basis $\{\ket{e_i}\}_i$ of $\mathcal{H}_E$ and defining $L_i = (\idmat_B \otimes \bra{e_i}) V$. A map $T$ is unital if $T(\idmat_B) = \idmat_A$ and a unital normal CP-map is called a Heisenberg (quantum) channel. 

Let $S \in \blt(\trcl(\mathcal{H}_A); \trcl(\mathcal{H}_B))$. The dual map $S^* \in \blt(\blt(\mathcal{H}_B); \blt(\mathcal{H}_A))$ is the unique linear map that satisfies $\mathrm{tr}[X_B^\dagger S(\rho)] = \mathrm{tr}[\left(S^*(X_B)\right)^\dagger \rho]$, for all $X_B \in \blt(\mathcal{H}_B)$ and all $\rho \in \trcl(\mathcal{H}_A)$. We call $S$ the Schrödinger picture map and $S^*$ the Heisenberg picture map. The map $S$ is called completely positive if $S^*$ is completely positive in the sense defined above. In that case, $S^*$ is automatically normal. In fact, $T$ is a normal CP-map if and only if there exists $S \in \blt(\blt(\mathcal{H}_A); \blt(\mathcal{H}_B))$, such that $S^* = T$. It follows that $S$ is completely positive if and only if there exists a separable Hilbert space $\mathcal{H}_E$ and an operator $V \in \blt(\mathcal{H}_A; \mathcal{H}_B \otimes \mathcal{H}_E)$, such that $S(\rho) = \ptr{E}{V\rho V^\dagger}$, for all $\rho \in \trcl(\mathcal{H}_A)$. Furthermore, $S$ is completely positive if and only if there exist a countable set of operators $\{L_i\}_i \subset \blt(\mathcal{H}_A; \mathcal{H}_B)$ such that $S(\rho) = \sum_i L_i \rho L_i^\dagger$ and the series converges in trace-norm. A map $S$ is trace-preserving if $\tr{S(\rho_A)} = \tr{\rho_A}$ for all $\rho_A \in \trcl(\mathcal{H}_A)$. A trace-preserving CP-map is called a (quantum) channel. 
The facts in this section are contained or follow directly from results in~\cite{davies1976quantum, AttalCPMaps}.

\subsection{Choi–Jamiołkowski isomorphism, partial transposition}

In this section, let $\mathcal{H}_A$, $\mathcal{H}_B$ and $\mathcal{H}_C$ be finite-dimensional Hilbert spaces with fixed orthonormal bases $\{\ket{a_i}\}_i$, $\{\ket{b_j}\}_j$ and $\{\ket{c_k}\}_k$, respectively. The transpose (w.r.t.~$\{\ket{a_i}\}_i$ and $\{\ket{b_j}\}_j$) of an operator $X \in \blt(\mathcal{H}_A; \mathcal{H}_B)$ is the unique linear operator $X^T \in \blt(\mathcal{H}_B; \mathcal{H}_A)$ such that $\braket{b_j}{X a_i} = \braket{a_i}{X^T b_j}$, for all elements of the orthonormal bases. The partial transposition (w.r.t.~$\{\ket{a_i}\}_i$) of an operator $X \in \blt(\mathcal{H}_A \otimes \mathcal{H}_B; \mathcal{H}_A \otimes \mathcal{H}_C)$ is the unique linear operator $X^{T_A} \in \blt(\mathcal{H}_A \otimes \mathcal{H}_B; \mathcal{H}_A \otimes \mathcal{H}_C)$ such that $(\bra{a_i} \otimes \idmat_C) X (\ket{a_j} \otimes \idmat_B) = (\bra{a_j} \otimes \idmat_C) X^{T_A} (\ket{a_i} \otimes \idmat_B)$, for all elements of the orthonormal basis.

The (quantum) Choi–Jamiołkowski isomorphism~\cite{CHOI1975285, JAMIOLKOWSKI1972275}, defined with respect to an orthonormal basis $\{\ket{a_i}\}_i$ of $\mathcal{H}_A$, is the bijective linear map $\mathfrak{C}_{A;B} : \blt(\blt(\mathcal{H}_A); \blt(\mathcal{H}_B)) \rightarrow \blt(\mathcal{H}_A \otimes \mathcal{H}_B)$, $\choi_{A;B}(T) = (\idop_A \otimes T)(\ket{\Omega}\bra{\Omega})$, and its inverse is given by $\choi_{A;B}^{-1}(\tau)(\rho) = \ptr{A}{(\rho^T \otimes \idmat) \tau}$, where $\ket{\Omega} := \sum_i \ket{a_i} \otimes \ket{a_i}$. A map $S \in \blt(\blt(\mathcal{H}_A); \blt(\mathcal{H}_B))$ is completely positive if and only if $\choi_{A;B}(S) \geq 0$; $S$ is trace-preserving if and only if $\ptr{B}{\choi_{A;B}(S)} = \idmat_A$ and we have the identity $\ptr{A}{\choi_{A;B}(S)} = S(\idmat_A)$. We will occasionally call elements of the image of $\choi_{A;B}$ Choi matrices. 

\subsection{Non-negative matrices and duality}\label{Sbsct:PreliminariesClassical}

As we provide characterizations for both the quantum and the classical case, we now also introduce the notation and definitions required for the latter. With a classical system $A$, we associate a finite alphabet $\mathbb{A} = \{a_1, a_2, \dots, a_{\abs{\mathbb{A}}}\}$ and a `state-space' $\R^\mathbb{A}$, with orthonormal basis $\{\ket{a_i}\}_{i = 1}^{\abs{\mathbb{A}}}$. We define by $\ket{\idvec_A} := \sum_i \ket{a_i}$ the all-one-vector. A vector $\ket{x} \in \R^\mathbb{A}$ is called non-negative if $\braket{a}{x} \geq 0$, for all $a \in \mathbb{A}$. A linear operator $M \in \blt(\R^\mathbb{A}; \R^\mathbb{B})$ is called non-negative if $M\ket{x}$ is non-negative, whenever $\ket{x}$ is non-negative (equivalently, all matrix elements are non-negative). A non-negative $M$ is called column-stochastic if $\bra{\idvec_B}M = \bra{\idvec_A}$; column-sub-stochastic if there exists a non-negative $P$, such that $M+P$ is column-stochastic; row-stochastic, if $M\ket{\idvec_A} = \ket{\idvec_B}$; and row-sub-stochastic if there exists a non-negative $P$, such that $M+P$ is row-stochastic. Given $\ket{x}$ or $\bra{x}$, we denote by $\mathrm{diag}(\ket{x}) = \mathrm{diag}(\bra{x})$ the diagonal matrix with the components of $x$ on the diagonal.  
Finally, we will use the `classical Choi–Jamiołkowski isomorphism' (also known as vectorization), which is a convenient notation to make the connection to the quantum case more transparent. The classical Choi–Jamiołkowski isomorphism, defined w.r.t.~$\{\ket{a_i}\}_i$, is the linear map $\choi^C_{A;B} : \blt(\R^\mathbb{A}; \R^\mathbb{B}) \rightarrow \blt(\R^\mathbb{A} \otimes \R^\mathbb{B})$ defined by $\choi^C_{A;B}(M) = (\idmat_A \otimes M)\ket{\Omega}$, where $\ket{\Omega} := \sum_i \ket{a_i}\otimes \ket{a_i}$. The inverse $(\choi^C_{A;B})^{-1}$ is then given by $(\choi^C_{A;B})^{-1}(\ket{x}) = (\bra{\Omega} \otimes \idmat_B)(\idmat_A \otimes \ket{x})$
We will sometimes refer to elements of the range of $\choi^C_{A;B}$ as Choi vectors.

\subsection{Dynamical semigroups}

Let $\mathcal{X}$ be a Banach space. A family of operators $\{T_t\}_{t \geq 0}$, with $T_t \in \blt(\mathcal{X})$ for all $t \geq 0$, is called a norm-continuous one-parameter semigroup on $\mathcal{X}$, or short, dynamical semigroup, if $T_0 = \idmat$, $T_{s+t} = T_s T_t$ for all $t, s \geq 0$ and the map $\R_{\geq 0} \ni t \mapsto T_t$ is norm-continuous. Norm-continuous dynamical semigroups are automatically differentiable and have bounded generators, that is, there exists $L \in \blt(\mathcal{X})$ such that $T_t = e^{tL}$ for all $t \geq 0$ and $L = \frac{d}{dt}\big\vert_{t = 0_+} T_t$~\cite[Thm. I.3.7]{engel2006one}.

Lindblad~\cite{Lindblad.1976} proved that $T_t \in \mathrm{CP}_\sigma(\mathcal{H})$ for all $t \geq 0$ if and only if there exist $\Phi \in \mathrm{CP}_\sigma(\mathcal{H})$ and $K \in \blt(\mathcal{H})$ such that $T_t = e^{tL}$, with $L(X) = \Phi(X) - K^\dagger X - X K$. In this case, we refer to $\{T_t\}_{t \geq 0}$ as a CP semigroup. We call the corresponding form of the generator $L$ the GKLS form~\cite{Gorini.1976, Lindblad.1976} and $\Phi$ its CP part.
If $\mathcal{H}$ is finite-dimensional, then $T_t = e^{tL} \in \mathrm{CP}_\sigma(\mathcal{H})$ for all $t \geq 0$ if and only if the operator $\mathfrak{L} := \choi_{A;B} = (\idop \otimes L)(\ket{\Omega}\bra{\Omega})$ is self-adjoint and $P^\bot \mathfrak{L}P^\bot \geq 0$, where $\ket{\Omega} = \sum_i \ket{a_i} \otimes \ket{a_i}$, for some orthonormal basis $\{a_i\}$ of $\mathcal{H}$ and $P^\bot \in \blt(\mathcal{H} \otimes \mathcal{H})$ is the orthogonal projection onto the orthogonal complement of $\{\ket{\Omega}\}$~\cite{wolf2008assessing, evans1977dilations}. 
The corresponding classical result is as follows: $\{T_t\}_{t \geq 0} \subseteq \blt(\R^\mathbb{A})$ is a dynamical semigroup of non-negative linear maps if and only if there exists a non-negative linear map $\Phi \in \blt(\R^\mathbb{A})$ and a diagonal map $K \in \blt(\R^\mathbb{A})$ (w.r.t.~the basis orthogonal basis $\{\ket{a_i}\}_i$) such that the generator $L$ has the form $\Phi - K$~\cite{batkai2017positive}.

\section{The Classical Case}\label{sct:classical}

Before studying the quantum scenario, we consider the classical version of our main question. I.e., we study continuous semigroups of classical superchannels and their generators. On the one hand, this allows us to develop an intuition that we can build upon for the quantum case. On the other hand, a comparison between the classical and the quantum case elucidates which features of the latter are actually quantum.
For the purpose of this section, $\mathbb{A}$, $\mathbb{B}$ and $\mathbb{E}$ denote finite alphabets as in Subsection \ref{Sbsct:PreliminariesClassical}.

A classical superchannel is a map that maps classical channels, i.e., stochastic matrices, to classical channels while preserving the probabilistic structure of the classical theory. To achieve the latter requirement, we require that a classical superchannel is a linear map and that probabilistic transformations, i.e., sub-stochastic matrices, are mapped to probabilistic transformations. Expressed more formally, we have

\begin{defn}[Classical Superchannels]\label{Dff:ClassicalSuperchannels}
A linear map $\hat{S} : \blt(\R^{\mathbb{A}}; \R^{\mathbb{B}})\to \blt(\R^{\mathbb{A}}; \R^{\mathbb{B}})$ is called a \emph{classical superchannel} if $\hat{S}(M)\in\blt(\R^{\mathbb{A}}; \R^{\mathbb{B}})$ is column sub-stochastic whenever $M \in \blt(\R^{\mathbb{A}}; \R^{\mathbb{B}})$ is column sub-stochastic and $\hat{S}(M) \in \blt(\R^{\mathbb{A}}; \R^{\mathbb{B}})$ is column stochastic whenever $M \in \blt(\R^{\mathbb{A}}; \R^{\mathbb{B}})$ is column stochastic.
\end{defn}

A related concept is that of a classical semicausal channel, which is a stochastic matrix on a bipartite space $\mathbb{A} \times \mathbb{B}$ such that no communication from $B$ to $A$ is allowed. We formalize this as follows:

\begin{defn}[Classical Semicausality] \label{Defn:ClassicalSemicausality} An operator $M \in \blt(\R^\mathbb{A} \otimes \R^\mathbb{B})$ is called column $B \not\to A$ semicausal if there exists $M^A \in \blt(\R^\mathbb{A})$, such that $(\idmat_A \otimes \bra{\idvec_B})M = M^A(\idmat_A \otimes \bra{\idvec_B})$. 

Similarly, $N \in \blt(\R^\mathbb{A} \otimes \R^\mathbb{B})$ is called row $B \not\to A$ semicausal if there exists $N^A \in \blt(\R^\mathbb{A})$, such that $N (\idmat_A \otimes \ket{\idvec_B}) = N^A \otimes \ket{\idvec_B}$.
\end{defn}

Clearly, $M$ is column $B \not\to A$ semicausal if and only if $M^T$ is row $B \not\to A$ semicausal. To emphasize the analogy to the quantum case, we will often refer to a column $B \not\to A$ semicausal map as a Schrödinger $B \not\to A$ semicausal map and to a row $B \not\to A$ semicausal map as a Heisenberg $B \not\to A$ semicausal map. In both cases, the maps $M^A$ and $N^A$ will be called the reduced maps. 

The structure of this section is as follows: We start by establishing the connection between classical superchannels and classical non-negative semicausal maps, followed by a characterization of classical non-negative semicausal maps as a composition of known objects; such a characterization is known in the quantum case as the equivalence between semicausality and semilocalizability. We then turn to the study of the generators of semigroups of semicausal and non-negative maps and finally use the correspondence between superchannels and semicausal channels to obtain the corresponding results for the generators of semigroups of superchannels.   

\subsection{Correspondence between classical superchannels and semicausal nonnegative linear maps}\label{SbSctClassicalSuperchannelsSemicausalChannels}

We first show, with a proof inspired by the one given in~\cite{Chiribella_2008} for the analogous correspondence in the quantum case, that we can understand classical superchannels in terms of classical semicausal channels. To concisely state this correspondence, we use the classical version of the Choi–Jamiołkowski isomorphism. Let us mention here one again that we assume all alphabets ($\mathbb{A}, \mathbb{B}, \dots$) to be finite for our treatment of the classical case.

\begin{thm}\label{Thm:ClassicalCorrespondenceSuperchannelsSemicausalMaps}
Let $\hat{S} : \blt(\mathbb{R}^\mathbb{A};\mathbb{R}^\mathbb{B})\to \blt(\mathbb{R}^\mathbb{A};\mathbb{R}^\mathbb{B})$ be a linear map and define $S\in\blt(\mathbb{R}^\mathbb{A}\otimes\mathbb{R}^\mathbb{B})$ via $S = \choi^C_{A;B} \circ \hat{S} \circ (\choi^C_{A;B})^{-1}$. Then, $\hat{S}$ is a classical superchannel if and only if $S$ is non-negative and (Schrödinger $B\not\to A$) semicausal such that the reduced map $S^{A}$ satisfies $S^{A}\ket{\idvec_A} = \ket{\idvec_A}$. In this case, $S^{A}$ is automatically non-negative.    
\end{thm}
\begin{proof}
We first show the ``if''-direction, i.e., that if $S$ is non-negative and (Schrödinger $B \not\to A$) semicausal, then $\hat{S} = (\choi^C_{A;B})^{-1} \circ S \circ \choi^C_{A;B}$ is a superchannel. Suppose $M$ is a non-negative matrix. Then $\hat{S}(M) $ is non-negative, since $\choi^C_{A;B}$ maps non-negative matrices to non-negative vectors, $S$ maps non-negative vectors to non-negative vectors and $(\choi^C_{A;B})^{-1}$ maps non-negative vectors to non-negative matrices.\\
Furthermore, if $M$ is column stochastic, then
\begin{align*}
    \bra{\idvec_B} \hat{S}(M)  &= \bra{\idvec_B} \left(\choi^C_{A;B}\right)^{-1} \circ S \circ \choi^C_{A;B}(M) 
    \\&= (\bra{\Omega} \otimes \bra{\idvec_B}) \left(\idmat_{A} \otimes S \left( \choi^C_{A;B}(M) \right)\right)\\
    &= \bra{\Omega} \left( \idmat_A \otimes S^A\left((\idmat_A \otimes \bra{\idvec_B}) \choi^C_{A;B}(M) \right) \right) \\
    &= \bra{\Omega} \left(\idmat_A \otimes S^A((\idmat_A \otimes (\bra{\idvec_B}M)) \ket{\Omega})\right) \\
    &= \bra{\Omega} \left(\idmat_A \otimes S^A\ket{\idvec_A}\right) \\
    &= \bra{\Omega} \left(\idmat_A \otimes \ket{\idvec_A}\right) \\
    &= \bra{\idvec_A},
\end{align*}
so $\hat{S}(M)$ is stochastic. In the preceding calculation, we used that $S$ is semicausal in the third line, that $M$ is stochastic in the fifth line, and that $S^A\ket{\idvec_A} = \ket{\idvec_A}$ in the sixth line.

\noindent Now suppose that $M$ is sub-stochastic, such that $M + Q$ is stochastic, with $Q$ non-negative. Then $\hat{S}(M+Q) = \hat{S}(M) + \hat{S}(Q)$ is stochastic and since $\hat{S}(Q)$ is non-negative, $\hat{S}(M)$ is sub-stochastic. This proves that $\hat{S}$ is a superchannel. The claim about the non-negativity of $S^A$ now follows directly from the semicausality condition.\\
For the converse, suppose $\hat{S}$ is a superchannel. Since for all $a \in \mathbb{A}$ and all $b \in \mathbb{B}$, the matrix $\ket{b}\bra{a}$ is sub-stochastic, it follows by linearity of $\hat{S}$ that $\hat{S}(M)$ is non-negative whenever $M$ is non-negative. Thus, since $(\choi^C_{A;B})^{-1}$ maps non-negative vectors to non-negative matrices, $\hat{S}$ maps non-negative matrices to non-negative matrices and $\choi^C_{A;B}$ maps non-negative matrices to non-negative vectors, it follows that $S$ is non-negative.

\noindent Next, we want to show that $S$ is Schrödinger $B \not\to A$ semicausal. Since $\hat{S}$ is a superchannel, $S$ maps Choi vectors of stochastic matrices to Choi vectors of stochastic matrices, that is, $(\idmat_A \otimes \bra{\idvec_B}) S \ket{x} = \ket{\idvec_A}$, for all non-negative vectors $\ket{x} \in \mathbb{R}^\mathbb{A}\otimes\mathbb{R}^\mathbb{B}$ that satisfy $(\idmat_A \otimes \bra{\idvec_B})\ket{x} = \ket{\idvec_A}$. 
As a tool, we define the set of scaled differences of Choi vectors of stochastic matrices by 
\begin{align} \label{Eq:DefiningC0}
    C_0 &:= \set{\lambda (\ket{p} - \ket{n})}{\lambda \in \R; \;\ket{p}, \ket{n} \in \mathbb{R}^\mathbb{A}\otimes\mathbb{R}^\mathbb{B} \text{ non-negative, with } (\idmat_A\otimes \bra{\idvec_B})\ket{p} = (\idmat_A\otimes \bra{\idvec_B})\ket{n} = \ket{\idvec_A}}.
\end{align}
We claim that 
\begin{align*}
    C_0 = C_0^\prime &:= \set{\ket{x^\prime} \in \mathbb{R}^\mathbb{A}\otimes\mathbb{R}^\mathbb{B}}{(\idmat_A \otimes \bra{\idvec_B})\ket{x^\prime} = 0}.
\end{align*}
To see this, first note that $C_0 \subseteq C_0^\prime$ follows directly from the definition. For the other inclusion, $C_0 \supseteq C_0^\prime$, we decompose $\ket{x^\prime} \in C_0^\prime$ as $\ket{x^\prime} = \ket{p^\prime} - \ket{n^\prime}$, for two non-negative vectors $\ket{p^\prime}, \ket{n^\prime} \in \mathbb{R}^\mathbb{A}\otimes\mathbb{R}^\mathbb{B}$. It follows that $(\idmat_A\otimes \bra{\idvec_B})\ket{p^\prime} = (\idmat_A\otimes \bra{\idvec_B})\ket{n^\prime}$. Furthermore, for $\epsilon > 0$ small enough, we have that $\ket{y^\prime} := \ket{\idvec_A} - \epsilon (\idmat_A\otimes \bra{\idvec_B})\ket{p^\prime}$ is non-negative. But then, for any non-negative unit $\ket{v} \in \mathbb{R}^\mathbb{B}$, with $\braket{\idvec_B}{v} = 1$, the vectors $\ket{p} := \epsilon \ket{p^\prime} +  \ket{y^\prime} \otimes \ket{v}$ and $\ket{n} := \epsilon \ket{n^\prime} + \ket{y^\prime} \otimes \ket{v}$ are Choi vectors of stochastic matrices. So $\ket{x^\prime} = \frac{1}{\epsilon}(\ket{p} - \ket{n}) \in C_0$.\\
\noindent We define $P^\bot \in \blt(\mathbb{R}^\mathbb{A}\otimes\mathbb{R}^\mathbb{B})$ by $P^\bot \ket{x} = \frac{1}{\abs{\mathbb{B}}} \left[(\idmat_A \otimes \bra{\idvec_B}) \ket{x}\right] \otimes \ket{\idvec_B}$ and $P := \idmat_{AB} - P^\bot$. Then, since $(\idmat_A \otimes \bra{\idvec_B})P \ket{x} = (\idmat_A \otimes \bra{\idvec_B})\ket{x} -  (\idmat_A \otimes \bra{\idvec_B})\ket{x} = 0$, we have that $P\ket{x} \in C_0$, for all $\ket{x} \in \mathbb{R}^\mathbb{A}\otimes\mathbb{R}^\mathbb{B}$. We define $S^A \in \blt(\mathbb{R}^\mathbb{A})$ by $S^A\ket{x_A} = \frac{1}{\abs{\mathbb{B}}}(\idmat_A \otimes \bra{\idvec_B})P^\perp S(\ket{x_A} \otimes \ket{\idvec_B}) = \frac{1}{\abs{\mathbb{B}}}(\idmat_A \otimes \bra{\idvec_B}) S(\ket{x_A} \otimes \ket{\idvec_B})$ and calculate
\begin{align*}
    (\idmat_A \otimes \bra{\idvec_B}) S\ket{x} &= (\idmat_A \otimes \bra{\idvec_B}) S(P\ket{x}) + (\idmat_A \otimes \bra{\idvec_B}) S(P^\bot \ket{x}) \\
    &= (\idmat_A \otimes \bra{\idvec_B}) S(P^\bot \ket{x}) \\&= (\idmat_A \otimes \bra{\idvec_B}) S\left(\frac{1}{\abs{\mathbb{B}}} \left[(\idmat_A \otimes \bra{\idvec_B}) \ket{x}\right] \otimes \ket{\idvec_B}\right) \\&= S^A( (\idmat_A \otimes \bra{\idvec_B}) \ket{x}),
\end{align*}
where we used in the second line that $C_0$ is invariant under $S$, a fact that follows directly from \eqref{Eq:DefiningC0}. This calculation exactly shows that $S$ is Schödinger $A \not\to B$ semicausal.\\
It remains to show that $S^A \ket{\idvec_A} = \ket{\idvec_A}$. This follows easily, since
\begin{align*}
    S^A\ket{\idvec_A} &= \frac{1}{\abs{\mathbb{B}}}(\idmat_A \otimes \bra{\idvec_B}) S(\ket{\idvec_A} \otimes \ket{\idvec_B}) \\
    &= \frac{1}{\abs{\mathbb{B}}}(\idmat_A \otimes \bra{\idvec_B}) \choi^C_{A;B} \circ \hat{S} \circ (\choi^C_{A;B})^{-1} (\ket{\idvec_A} \otimes \ket{\idvec_B}) \\
    &= \idmat_A \otimes \left[\bra{\idvec_B} \hat{S}\left(\frac{1}{\abs{\mathbb{B}}} \ket{\idvec_B}\bra{\idvec_A}\right)\right] \ket{\Omega}\\
    &= (\idmat_A \otimes \bra{\idvec_A}) \ket{\Omega}\\
    &= \ket{\idvec_A},
\end{align*}
where we used that $\frac{1}{\abs{\mathbb{B}}} \ket{\idvec_B}\bra{\idvec_A}$ is stochastic and that thus $\hat{S}(\frac{1}{\abs{\mathbb{B}}} \ket{\idvec_B}\bra{\idvec_A})$ is stochastic.
\end{proof}

In summary, Theorem \ref{Thm:ClassicalCorrespondenceSuperchannelsSemicausalMaps} tells us that, via the classical Choi–Jamiołkowski isomorphism, we can view classical superchannels equivalently also as suitably normalized semicausal non-negative maps.

\subsection{Relation between classical semicausality and semilocalizability}

The goal of this section is to get a better understanding of the structure of semicausal maps. For non-negative semicausal maps, we have the following structure theorem:

\begin{thm} \label{ClassicalSemicausalIsSemilocalizable}
A non-negative map $N \in \blt(\R^\mathbb{A} \otimes \R^\mathbb{B})$ is row $B \not\to A$ semicausal, if and only if there exists a (finite) alphabet $\mathbb{E}$, a (non-negative) row-stochastic matrix $U \in \blt(\R^\mathbb{B}; \R^\mathbb{E} \otimes \R^\mathbb{B})$ and a non-negative matrix $A \in \blt(\R^\mathbb{A} \otimes \R^\mathbb{E}; \R^\mathbb{A})$ such that 
\begin{align} \label{SemilocalizableEq}
    N = (A \otimes \idmat_B)(\idmat_A \otimes U).
\end{align}
In that case, we can choose $\abs{\mathbb{E}} = \abs{\mathbb{A}}^2$. 
\end{thm}

Borrowing the terminology from the quantum case~\cite{Beckman.2001, Eggeling.2002}, the preceding theorem tells us that non-negative semicausal maps are semilocalizable. We formally define the latter notion for the classical case as follows:
\begin{defn}
A non-negative map $N \in \blt(\R^\mathbb{A} \otimes \R^\mathbb{B})$ is called Heisenberg $B \not\to A$ semilocalizable if it can be written in the form of Eq.~\eqref{SemilocalizableEq}.

Similarly, a non-negative map $M \in \blt(\R^\mathbb{A} \otimes \R^\mathbb{B})$ is called Schrödinger $B \not\to A$ semilocalizable if it can be written as $M = (\idmat_A \otimes U)(A \otimes \idmat_B)$, for a (non-negative) column-stochastic matrix $U \in \blt(\R^\mathbb{E} \otimes \R^\mathbb{B}; \R^\mathbb{B})$ and a non-negative matrix $A \in \blt(\R^\mathbb{A}; \R^\mathbb{A} \otimes \R^\mathbb{E})$.
\end{defn}

The requirement that $U$ is stochastic and $A$ is non-negative in the decomposition above is essential. In fact, if one drops these requirements, then a decomposition $M = (\idmat_A \otimes U)(A \otimes \idmat_B)$ can be found for any matrix $M \in \blt(\R^\mathbb{A} \otimes \R^\mathbb{B})$.

Due to Theorem \ref{ClassicalSemicausalIsSemilocalizable}, a non-negative Schrödinger $B \not\to A$ semicausal and column-stochastic map $M$ admits an operational interpretation. First, note that if $M$ is not only semicausal, but also stochastic, then also the matrix $A$ in Eq.~\eqref{SemilocalizableEq} is stochastic. Thus, the interpretation of the decomposition is: First, Alice applies some probabilistic operation ($A$) to the composite system $\mathbb{A}\times\mathbb{E}$. Then she transmits the $E$-part to Bob, who now applies a stochastic operation ($U$) to his part of the system.

Given this interpretation, the idea behind the construction in the proof of Theorem \ref{ClassicalSemicausalIsSemilocalizable} is that Alice first looks the input of system $A$ and generates the output of system $A$ according to the distribution given by the matrix $N^A$. Then she copies the input as well as her generated output and sends this information to Bob, who is then able to complete the operation by generating an output conditional on his input and the information he got from Alice. Given that this construction requires copying, it might be considered surprising that a quantum analogue is true nevertheless~\cite{Eggeling.2002}.

\begin{proof}(Theorem \ref{ClassicalSemicausalIsSemilocalizable})
If $N$ is Schrödinger $B \not\to A$ semilocalizable, then
\begin{align*}
    N(\idmat_A \otimes \ket{\idvec_B}) = (A \otimes \idmat_B)(\idmat_A \otimes U\ket{\idvec_B}) = (A \otimes \idmat_B)(\idmat_A \otimes \ket{\idvec_{EB}}) = (A (\idmat_A \otimes \ket{\idvec_E})) \otimes \ket{\idvec_B}.
\end{align*}
So, $N$ is row $B \not\to A$ semicausal. \\
\noindent Conversely, if $N$ is row $B \not\to A$ semicausal, we choose $\mathbb{E} := \mathbb{A} \times \mathbb{A}$ and define 
\begin{align} \label{Eq:DefnAU}
    A &:= \sum_{i, j, k} \braket{a_j}{N^A a_k} \;\ket{a_j}\bra{a_k} \otimes \bra{a_k \otimes a_j}, \\
    U &:= \sum_{\substack{m, n, r, s\\ \braket{a_n}{N^A a_m} \neq 0}} \frac{\braket{a_n \otimes b_r}{N\, a_m \otimes b_s}}{\braket{a_n}{N^A a_m}} \ket{a_m \otimes a_n \otimes b_r}\bra{b_s} + \left[ \sum_{\substack{m, n\\ \braket{a_n}{N^A a_m} = 0}} \ket{a_m \otimes a_n} \right] \otimes \idmat_B. \nonumber
\end{align}
To show that $N = (A \otimes \idmat_B)(\idmat_A \otimes U)$, we calculate
\begin{align*}
    (A \otimes \idmat_B)(\idmat_A \otimes U) &= \sum_{\substack{i, j, k\\m, n, r, s\\ \braket{a_n}{N^A a_m} \neq 0}} \frac{\braket{a_j}{N^A a_k}\braket{a_n\otimes b_r}{N\,a_m \otimes b_s}}{\braket{a_n}{N^A a_m}} \left[ (\ket{a_j}\bra{a_k} \otimes \bra{a_k \otimes a_j} \otimes \idmat_B)(\idmat_A \otimes \ket{a_m \otimes a_n \otimes b_r}\bra{b_s}) \right] 
    \\&+ \sum_{\substack{i, j, k\\m, n\\ \braket{a_n}{N^A a_m} = 0}} \braket{a_j}{N^A a_k} (\ket{a_j}\bra{a_k} \otimes \bra{a_k \otimes a_j} \otimes \idmat_B)(\idmat_A \otimes \ket{a_m \otimes a_n} \otimes \idmat_B) \\
    &= \sum_{\substack{i, j, k, r, s\\ \braket{a_j}{N^A a_k} \neq 0}} \frac{\braket{a_j}{N^A a_k}\braket{a_j\otimes b_r}{N\,a_k \otimes b_s}}{\braket{a_j}{N^A a_k}}  \;\ket{a_j}\bra{a_k} \otimes \ket{b_r}\bra{b_s} \\ &+ \sum_{\substack{i, j, k\\ \braket{a_j}{N^A a_k} = 0}} \braket{a_j}{N^A a_k} \ket{a_j}\bra{a_k} \otimes \idmat_B\\
    &= N.
\end{align*}
For the last step, observe that the second sum vanishes and that one can drop the constraint that $\braket{a_j}{N^A a_k} \neq 0$ in the first sum (after cancellation), because $\braket{a_j\otimes b_r}{N\,a_k \otimes b_s} = 0$, if $\braket{a_j}{N^A a_k} = 0$. To see this last claim, note that, since $N$ is non-negative and semicausal, we have 
\begin{align*}
    0 \leq \braket{a_j\otimes b_r}{N\,a_k \otimes b_s} \leq \braket{a_j\otimes b_r}{N\,a_k \otimes \idvec_B} = \braket{a_j}{N^A a_k} \braket{b_r}{\idvec_B} = 0.
\end{align*}
It is clear, that $A$ and $U$ are non-negative, since $N$ and thus also $N^A$ are non-negative by assumption. It remains to show that $U$ is row-stochastic. We have
\begin{align*}
    U \ket{\idvec_{B}} &= \sum_{\substack{m, n, r, s\\ \braket{a_n}{N^A a_m} \neq 0}} \frac{\braket{a_n \otimes b_r}{N\, a_m \otimes b_s}}{\braket{a_n}{N^A a_m}} \ket{a_m \otimes a_n \otimes b_r} +  \sum_{\substack{m, n, s\\ \braket{a_n}{N^A a_m} = 0}} \ket{a_m \otimes a_n \otimes b_s} \\
    &= \sum_{\substack{m, n, r\\ \braket{a_n}{N^A a_m} \neq 0}} \frac{\braket{a_n \otimes b_r}{N\, a_m \otimes \idvec_B}}{\braket{a_n}{N^A a_m}} \ket{a_m \otimes a_n \otimes b_r} +  \sum_{\substack{m, n, s\\ \braket{a_n}{N^A a_m} = 0}} \ket{a_m \otimes a_n \otimes b_s} \\
    &= \sum_{\substack{m, n, r\\ \braket{a_n}{N^A a_m} \neq 0}} \ket{a_m \otimes a_n \otimes b_r} +  \sum_{\substack{m, n, s\\ \braket{a_n}{N^A a_m} = 0}} \ket{a_m \otimes a_n \otimes b_s} \\
    &= \ket{\idvec_{EB}},
\end{align*}
where we used the condition that $N$ is semicausal to obtain the third line. This finishes the proof.
\end{proof}

\begin{rem}
Theorem \ref{ClassicalSemicausalIsSemilocalizable} can be extended to weak-$^*$ continuous non-negative maps on the Banach space of bounded real sequences, but this requires extra care and does not yield additional insight beyond the previous proof.
\end{rem}

\subsection{Generators of semigroups of classical semicausal non-negative maps}

The main goal of this section is to establish a structure theorem for the generators of semigroups of non-negative semicausal maps. 
First, recall that a (norm)-continuous semigroup $\left\{N_t\right\}_{t \geq 0} \subseteq \blt(\R^\mathbb{A}\otimes \R^\mathbb{B})$ has a generator $Q \in \blt(\R^\mathbb{A}\otimes \R^\mathbb{B})$ such that $N_t = e^{tQ}$. A classical result states that $N_t$ is non-negative for all $t \geq 0$ if and only if the generator $Q$ can be written in the form $Q = \Phi - K$, where $\Phi$ is non-negative and $K$ is a diagonal matrix w.r.t.~the canonical basis~\cite{Liggett.2010}. A second, crucial observation is that $N_t$ is Heisenberg $B \not\to A$ semicausal for all $t \geq 0$ if and only if $Q$ is Heisenberg $B \not\to A$ semicausal. To see this, let us first show that the reduced maps $\left\{N^A_t\right\}_{t \geq 0}$ also form a norm-continuous semigroup of non-negative maps. Since non-negativity is clear, we derive the semigroup properties ($N^A_0 = \idmat_A$, $N^A_{t+s} = N^A_t N^A_s$ and continuity) from the corresponding ones of $\{N_t\}_{t \geq 0}$:
\begin{align*}
    N^A_0 &= (\idmat_A \otimes \bra{b_1}) (N^A_0 \otimes \ket{\idvec_B} ) = (\idmat_A \otimes \bra{b_1}) N_0 (\idmat_A \otimes \ket{\idvec_B}) = (\idmat_A \otimes \bra{b_1})(\idmat_A \otimes \ket{\idvec_B}) = \idmat_A,\\
    N^A_{t + s} &= (\idmat_A \otimes \bra{b_1}) (N^A_{t+s} \otimes \ket{\idvec_B} ) = (\idmat_A \otimes \bra{b_1}) N_{t + s} (\idmat_A \otimes \ket{\idvec_B}) = (\idmat_A \otimes \bra{b_1}) N_{t}N_{s} (\idmat_A \otimes \ket{\idvec_B}) \\&= (\idmat_A \otimes \bra{b_1}) N_{t} (\idmat_A \otimes \ket{\idvec_B}) N^A_{s} = (\idmat_A \otimes \bra{b_1}) (\idmat_A \otimes \ket{\idvec_B}) N^A_{t} N^A_{s}  = N^A_t N^A_s,\\
    \norm{N_t^A - N^A_s} &= \sup_{\norm{x}_\infty = 1} \norm{(N^A_t - N^A_s)\ket{x}}_\infty = \sup_{\norm{x}_\infty = 1} \norm{((N^A_t - N^A_s)\ket{x}) \otimes \ket{\idvec_B}}_\infty = \sup_{\norm{x}_\infty = 1} \norm{(N_t - N_s) (\ket{x} \otimes \ket{\idvec_B})}_\infty \\&\leq \sup_{\norm{y}_\infty = 1} \norm{(N_t - N_s)\ket{y}} = \norm{N_t - N_s}.
\end{align*}
Thus, we conclude that $N^A_t = e^{tQ^A}$ for some generator $Q^A \in \blt(\R^\mathbb{A})$. 
We further have 
\begin{align*}
    Q(\idmat_A \otimes \ket{\idvec_B}) &= \frac{d}{dt}\bigg\vert_{t = 0} N_t(\idmat_A \otimes \ket{\idvec_B})
    \\&= \frac{d}{dt}\bigg\vert_{t = 0} (\idmat_A \otimes \ket{\idvec_B})N^A_t  \\
    &= (\idmat_A \otimes \ket{\idvec_B})Q^A.
\end{align*}
Thus, $Q$ is semicausal if $N_t$ is semicausal for all $t \geq 0$. Conversely, if $Q$ is semicausal, then $N_t$ is semicausal, since
\begin{align*}
    N_t (\idmat_A \otimes \ket{\idvec_B}) &= e^{tQ} (\idmat_A \otimes \ket{\idvec_B}) \\&= \sum_{k = 0}^\infty \frac{t^k}{k!} Q^k (\idmat_A \otimes \ket{\idvec_B}) \\&= \sum_{k = 0}^\infty \frac{t^k}{k!} (\idmat_A \otimes \ket{\idvec_B}) \left(Q^A\right)^k \\&= (\idmat_A \otimes \ket{\idvec_B}) e^{tQ^A}.
\end{align*}
Therefore, our task reduces to characterizing semicausal maps of the form $Q = \Phi - K$. Let us first remark that it is straight-forward to check (numerically) whether a given map satisfies these two conditions: We just need to check for non-negativity of the off-diagonal elements and whether $(\idmat_A \otimes \bra{b}) Q\ket{a_i \otimes \idvec_B} = 0$, for all $a_i \in \mathbb{A}$ and all $b \in \{ \ket{\idvec_B} \}^\bot$. I.e., semicausality can be checked in terms of $\abs{\mathbb{A}}(\abs{\mathbb{B}} - 1)$ linear equations and $\abs{\mathbb{A}}\abs{\mathbb{B}}(\abs{\mathbb{A}}\abs{\mathbb{B}} - 1)$ linear inequalities. Thus, a desirable result would be a normal form for all Heisenberg $B \not\to A$ semicausal generators $Q$, which allows for generating such maps, rather than checking whether a given maps is of the desired form. The main result of this section is exactly such a normal form. \\
To understand our normal form below, note that there are two natural ways of constructing a generator (remember that the matrix elements are interpreted as transition rates) that does not transmit information from system $B$ to system $A$. First, we can leave system $A$ unchanged and have transitions only on system $B$. The most basic form of such a map is $\ket{a_i}\bra{a_i}\otimes B^{(i)}$, for some $1\leq i\leq \lvert\mathbb{A}\rvert$ and for some $B^{(i)}\in\blt(\R^{\mathbb{B}})$ that is itself a valid generator of a semigroup of row-stochastic maps. That means that $B^{(i)} = \Phi^{(i)} - \mathrm{diag}(\Phi^{(i)}\ket{\idvec_B})$, for some non-negative matrix $\Phi^{(i)} \in \blt(\R^\mathbb{B})$. Second, if we want to act non-trivially on system $A$, we can make both of the two parts of a generator $Q = \Phi - K$, the non-negative part $\Phi \in \blt(\R^{\mathbb{A}} \otimes \R^{\mathbb{B}})$ and the diagonal part $K \in \blt(\R^{\mathbb{A}} \otimes \R^{\mathbb{B}})$, semicausal separately. Such a map has the form $\Phi_{sc} - K_A \otimes \idmat_B$, where $\Phi_{sc}$ is semicausal non-negative and $K_A \in \blt(\R^{\mathbb{A}})$ is diagonal. The fact that (convex) combinations of these basic building blocks already give rise to the most general form of semicausal generators for semigroups of non-negative bounded linear maps is the content of our next theorem, which establishes the desired normal form. 

\begin{thm}[Generators of classical semigroups of semicausal non-negative maps]\label{Thm:ClassicalSemicausalGeneratorDecompositionCharacterization}
A map $Q\in \blt(\R^{\mathbb{A}} \otimes \R^{\mathbb{B}})$ is the generator of a (norm-continuous) semigroup of Heisenberg $B\not\to A$ semicausal non-negative linear maps if and only if there exist a non-negative Heisenberg $B\not\to A$ semicausal map $\Phi_{sc} \in \blt(\R^{\mathbb{A}} \otimes \R^{\mathbb{B}})$, a diagonal map $K_A \in \blt(\R^{\mathbb{A}} \otimes \R^{\mathbb{B}})$, and linear maps $B^{(i)} \in \blt(\R^{\mathbb{B}})$ that generate (norm-continuous) semigroups of row-stochastic maps, for $1\leq i\leq \lvert\mathbb{A}\rvert$, such that
\begin{align*}
    Q = \Phi_{sc} - K_A\otimes \idmat_B + \sum\limits_{i=1}^{\lvert\mathbb{A}\rvert}\ket{a_i}\bra{a_i}\otimes B^{(i)}.
\end{align*}
In that case, $\Phi_{sc}$ can be chosen 'block-off-diagonal', i.e., $\Phi_{sc} = \sum_{i \neq j} \ket{a_i}\bra{a_j} \otimes \Phi_{sc}^{(i j)}$, for some collection of (non-negative) maps $\Phi^{(i j)}_{sc} \in \blt(\R^{\mathbb{B}})$. 
\end{thm}
\begin{proof}
It is straight-forward to check that a generator $Q$ of the given form has non-negative off-diagonal entries w.r.t.~the standard basis and is Heisenberg $B\not\to A$ semicausal. By the above discussion, this means that such a generator indeed gives rise to a semigroup of semicausal non-negative maps.

\noindent We prove the converse. Suppose $Q$ is the generator of a semigroup of non-negative linear maps. Then we can expand it as $Q = \sum_{i,j=1}^{\lvert\mathbb{A}\rvert} \ket{a_i}\bra{a_j}\otimes Q^{(ij)}$, where the operators $Q^{(ij)}\in\blt(\R^{\mathbb{B}})$ are non-negative for $i \neq j$ and of the form of a generator of a non-negative semigroup (i.e., non-negative minus diagonal) for $i = j$. This decomposition, together with semicausality, implies that for all $1\leq i,j\leq \lvert\mathbb{A}\rvert$,
\begin{align*}
Q^{(ij)}\ket{\idvec_B} 
= \left(\bra{a_i}\otimes \idvec_B\right) Q \left(\ket{a_j}\otimes \ket{\idvec_B} \right)
= \bra{a_i}Q^A\ket{a_j}\cdot \ket{\idvec_B}.
\end{align*}
In other words, $\ket{\idvec_B}$ is an eigenvector of every $Q^{(ij)}$, with corresponding eigenvalue $\lambda^{(ij)} = \bra{a_i}Q^A\ket{a_j}$. Hence, if we define $B^{(i)}\in\blt(\R^{\mathbb{B}})$ as $B^{(i)} := Q^{(ii)} - \lambda^{(ii)}\idmat_B$, then $B^{i}$ generates a semigroup of non-negative maps (since $Q^{(ij)}$ does and $\lambda^{(ii)}\idmat_B$ is diagonal) and satisfies (by construction), $B^{(i)} \ket{\idvec_B} = 0$. So $B^{(i)}$ generates a semigroup of row-stochastic maps.\\
With this notation, we can rewrite $Q$ as 
$$
Q 
= \underbrace{\sum\limits_{i\neq j} \ket{a_i}\bra{a_j}\otimes Q^{(ij)}}_{=:\Phi_{sc}} - \underbrace{\sum\limits_{i=1}^{\lvert\mathbb{A}\rvert}-\lambda^{(ii)}\ket{a_i}\bra{a_i}}_{=:K_A}\otimes \idmat_B + \sum\limits_{i=1}^{\lvert\mathbb{A}\rvert}\ket{a_i}\bra{a_i}\otimes B^{(i)}. 
$$
Note that $\Phi_{sc}$ is semicausal, since it can be written as the linear combination of the three semicausal maps $Q$, $K_A \otimes \idmat_B$ and $\sum_i \ket{a_i}\bra{a_i} \otimes B^{(i)}$. Thus, we have reached the claimed form.   
\end{proof}

By applying Theorem \ref{ClassicalSemicausalIsSemilocalizable}, we can further expand the $\Phi$ part:

\begin{cor} \label{Corolary1ClassicalSemicausalGenerators}
A map $Q\in \blt(\R^{\mathbb{A}} \otimes \R^{\mathbb{B}})$ is the generator of a (norm-continuous) semigroup of Heisenberg $B\not\to A$ semicausal non-negative linear maps if and only if there exist a (finite) alphabet $\mathbb{E}$, a (non-negative) row-stochastic matrix $U \in \blt(\R^{\mathbb{B}}; \R^{\mathbb{E}} \otimes \R^{\mathbb{B}})$, a non-negative matrix $A \in \blt( \R^{\mathbb{A}}\otimes \R^{\mathbb{E}}; \R^{\mathbb{A}})$, a diagonal matrix $K_A \in \blt(\R^{\mathbb{A}} \otimes \R^{\mathbb{B}})$, and maps $B^{(i)} \in \blt(\R^{\mathbb{B}})$ that generate (norm-continuous) semigroups of (row-)stochastic maps, for $1\leq i\leq \lvert\mathbb{A}\rvert$, such that
\begin{align*}
    Q = (A \otimes \idmat_B)(\idmat_A \otimes U) - K_A\otimes \idmat_B + \sum\limits_{i=1}^{\lvert\mathbb{A}\rvert}\ket{a_i}\bra{a_i}\otimes B^{(i)}.
\end{align*}
In that case, we can choose $\abs{\mathbb{E}} = \abs{\mathbb{A}}^2$.
\end{cor}

One should also note that with the notation of Corollary \ref{Corolary1ClassicalSemicausalGenerators}, the reduced map is given by $Q^A = (A (\idmat_A \otimes \ket{\idvec_B})) - K_A$. So, the reduced dynamics only depends on the operators $A$ and $K_A$. Further note that if we require the semigroup to consist of non-negative semicausal maps that are also row-stochastic, then we obtain the additional requirement that $K_A \ket{\idvec_A} = A \ket{\idvec_{AE}}$, which completely determines $K_A$. For completeness and later use, we write down the form of the generators non-negative semigroups that are Schrödinger $B \not\to A$ semicausal.
\begin{cor} \label{CorolarySchroedingerClassicalSemicausalGenerators}
A map $Q\in \blt(\R^{\mathbb{A}} \otimes \R^{\mathbb{B}})$ is the generator of a (norm-continuous) semigroup of Schrödinger $B\not\to A$ semicausal non-negative linear maps if and only if there exist a (finite) alphabet $\mathbb{E}$, a (non-negative) column-stochastic matrix $U \in \blt(\R^{\mathbb{E}} \otimes \R^{\mathbb{B}}; \R^{\mathbb{B}})$, a non-negative matrix $A \in \blt(\R^{\mathbb{A}}; \R^{\mathbb{A}}\otimes \R^{\mathbb{E}})$, a diagonal matrix $K_A \in \blt(\R^{\mathbb{A}} \otimes \R^{\mathbb{B}})$, and maps $B^{(i)} \in \blt(\R^{\mathbb{B}})$ that generate (norm-continuous) semigroups of column-stochastic maps, for $1\leq i\leq \lvert\mathbb{A}\rvert$, such that
\begin{align*}
    Q = (\idmat_A \otimes U)(A \otimes \idmat_B) - K_A\otimes \idmat_B + \sum\limits_{i=1}^{\lvert\mathbb{A}\rvert}\ket{a_i}\bra{a_i}\otimes B^{(i)}.
\end{align*}
In that case, we can choose $\abs{\mathbb{E}} = \abs{\mathbb{A}}^2$.
\end{cor}
\noindent Similar to the row-stochastic case, $B^{(i)}$ generates a semigroup of column-stochastic maps if and only if $B^{(i)} = \Phi^{(i)} - \mathrm{diag}(\bra{\idvec_B}\Phi^{(i)})$, for some non-negative matrix $\Phi^{(i)} \in \blt(\R^\mathbb{B})$

\subsection{Generators of semigroups of classical superchannels}

We finally turn to semigroups of classical superchannels, that is, a collection of classical superchannels $\left\{\hat{S_t} \right\}_{t \geq 0}$, such that $\hat{S}_0 = \idop$, $\hat{S}_{t+s} = \hat{S}_t \hat{S}_s$ and the map $t \mapsto \hat{S}_t$ is continuous (w.r.t.~any and thus all of the equivalent norms in finite dimensions). To formulate a technically slightly stronger result, we call a linear map $\hat{S}$ a preselecting supermap, if  $\choi^C_{A;B} \circ \hat{S} \circ (\choi^C_{A;B})^{-1}$ is a non-negative Schrödinger $B \not\to A$ semicausal map. Theorem \ref{Thm:ClassicalCorrespondenceSuperchannelsSemicausalMaps} then tells us that a superchannel is a special preselecting supermap. The result of this section is the following: 

\begin{thm} \label{thm:normalFormClassicalSuperchannels}
A linear map $\hat{Q} : \blt(\R^{\mathbb{A}}; \R^{\mathbb{B}}) \rightarrow \blt(\R^{\mathbb{A}}; \R^{\mathbb{B}})$ generates a semigroup of classical preselecting supermaps if and only if there exists a (finite) alphabet $\mathbb{E}$, a column-stochastic matrix $U \in \blt(\R^\mathbb{E} \otimes \R^\mathbb{B}; \R^\mathbb{B})$, a non-negative matrix $A \in \blt(\R^\mathbb{A}; \R^\mathbb{A} \otimes \R^\mathbb{E})$, a diagonal matrix $K_A \in \blt(\R^\mathbb{A})$ and a collection of generators of semigroups of column-stochastic matrices $B^{(i)} \in \blt(\R^\mathbb{B})$, such that 
\begin{align} \label{Eq:PreselectingNormalForm}
    \hat{Q}(M) = U (M \otimes \idmat_E) A - M K_A + \sum_{i = 1}^{\abs{\mathbb{A}}} B^{(i)} M \ket{a_i}\bra{a_i}.
\end{align}
Furthermore, $\hat{Q}$ generates a semigroup of classical superchannels if and only if $\hat{Q}$ generates a semigroup of preselecting supermaps and $\braket{a_i}{K_A a_i} = \braket{\idvec_{AE}}{A a_i}$, for all $1 \leq i \leq \abs{\mathbb{A}}$. In this case, $\hat{Q}$ is given by
\begin{align} \label{Eq:SuperchannelNormalForm}
    \hat{Q}(M) = U (M \otimes \idmat_E) A - \sum_{i = 1}^{\abs{\mathbb{A}}} \braket{\idvec_{AE}}{A a_i} M \ket{a_i}\bra{a_i} + \sum_{i = 1}^{\abs{\mathbb{A}}} B^{(i)} M \ket{a_i}\bra{a_i}.
\end{align}
\end{thm}
\begin{proof}
The main idea is to relate the generators of superchannels to those of semicausal maps. This relation is given by definition for preselecting supermaps and by Theorem \ref{Thm:ClassicalCorrespondenceSuperchannelsSemicausalMaps} for superchannels. For a generator $\hat{Q}$ of a semigroup of preselecting supermaps $\{\hat{S}_t\}_{t \geq 0}$, we have
\begin{align*}
    \hat{Q} = \frac{d}{dt} \bigg\vert_{t = 0} \hat{S}_t = (\choi^C_{A;B})^{-1} \frac{d}{dt} \bigg\vert_{t = 0} \left[ \choi^C_{A;B} \circ \hat{S}_t \circ (\choi^C_{A;B})^{-1}\right] \choi^C_{A;B}
\end{align*}
Thus $\hat{Q}$ generates a semigroup of preselecting supermaps if and only if $\hat{Q}$ can be written as $\hat{Q} = (\choi^C_{A;B})^{-1} \circ Q \circ \choi^C_{A;B}$, for some generator $Q$ of a semigroup of non-negative Schrödinger $B \not\to A$ semicausal maps. Thus to prove the first part of our Theorem, we simply take the normal form in Corollary \ref{CorolarySchroedingerClassicalSemicausalGenerators} and compute the similarity transformation above.\\ 
For $\ket{\Omega} = \sum_i \ket{a_i} \otimes \ket{a_i} \in \R^\mathbb{A} \otimes \R^\mathbb{A}$ and an operator $X_A \in \blt(\R^\mathbb{A})$, the well known identity $(X_A \otimes \idmat_A)\ket{\Omega} = (\idmat_A \otimes X_A^T) \ket{\Omega}$ can be proven by a direct calculation. Simmilarly, it is easy to show that for $X_A \in \blt(\R^\mathbb{A}; \R^\mathbb{A} \otimes \R^\mathbb{E})$, the slightly more general identity $(X_A \otimes \idmat_A) \ket{\Omega} = (\idmat_A \otimes \mathbb{F}_{A;E} X_A^{T_A}) \ket{\Omega}$ holds, where $\mathbb{F}_{A; E}$ is the flip operator that exchanges systems $A$ and $E$. We use these two identities in the following calculations. \\
For $\tilde{A} \in \blt(\R^\mathbb{A}; \R^\mathbb{A} \otimes \R^\mathbb{E})$ and $\tilde{U} \in \blt(\R^\mathbb{E} \otimes \R^\mathbb{B}; \R^\mathbb{B})$, we have, for any $M\in\blt(\R^{\mathbb{A}}; \R^{\mathbb{B}})$,
\begin{align*}
    (\choi^C_{A;B})^{-1} (\idmat_A \otimes \tilde{U}) (\tilde{A} \otimes \idmat_B) \choi^C_{A;B}(M) &= (\choi^C_{A;B})^{-1} (\idmat_A \otimes \tilde{U}) (\tilde{A} \otimes \idmat_B) (\idmat_A \otimes M) \ket{\Omega} \\
    &= (\choi^C_{A;B})^{-1} (\idmat_A \otimes (\tilde{U} (\idmat_E \otimes M))) (\tilde{A} \otimes \idmat_A) \ket{\Omega} \\
    &= (\choi^C_{A;B})^{-1} (\idmat_A \otimes (\tilde{U} (\idmat_E \otimes M))) (\idmat_A \otimes \mathbb{F}_{A ; E}\tilde{A}^{T_A}) \ket{\Omega} \\
    &= (\choi^C_{A;B})^{-1} (\idmat_A \otimes (\tilde{U} (\idmat_E \otimes M) \mathbb{F}_{A ; E}\tilde{A}^{T_A})) \ket{\Omega} \\
    &=  (\choi^C_{A;B})^{-1} \choi^C_{A;B}(\tilde{U} (\idmat_E \otimes M) \mathbb{F}_{A ; E}\tilde{A}^{T_A}) \\
    &= \tilde{U} (\idmat_E \otimes M) \mathbb{F}_{A ; E}\tilde{A}^{T_A} \\
    &= (\tilde{U}\mathbb{F}_{B ; E}) (M \otimes \idmat_E) \tilde{A}^{T_A}.
\end{align*}
For $\tilde{K}_A \in \blt(\R^\mathbb{A})$, we get, for any $M\in\blt(\R^{\mathbb{A}}; \R^{\mathbb{B}})$,
\begin{align*}
    (\choi^C_{A;B})^{-1} (K_A \otimes \idmat_B) \choi^C_{A;B}(M) &= (\choi^C_{A;B})^{-1} (\tilde{K}_A \otimes \idmat_B) (\idmat_A \otimes M) \ket{\Omega} \\
    &= (\choi^C_{A;B})^{-1} (\idmat_A \otimes M) (\tilde{K}_A \otimes \idmat_A) \ket{\Omega} \\
    &= (\choi^C_{A;B})^{-1} (\idmat_A \otimes M) (\idmat_A \otimes \tilde{K}_A^T) \ket{\Omega} \\
    &= (\choi^C_{A;B})^{-1} (\idmat_A \otimes M \tilde{K}_A^T) \ket{\Omega} \\
    &= (\choi^C_{A;B})^{-1} \choi^C_{A;B}(M \tilde{K}_A^T) \\
    &= M \tilde{K}_A^T.
\end{align*}
And finally, for an operator $\tilde{B}^{(i)}\in\blt(\mathbb{R}^{\mathbb{B}})$ and for any $1 \leq i \leq \abs{\mathbb{A}}$, we have, for any $M\in\blt(\R^{\mathbb{A}}; \R^{\mathbb{B}})$,
\begin{align*}
    (\choi^C_{A;B})^{-1} (\ket{a_i}\bra{a_i} \otimes B^{(i)}) \choi^C_{A;B}(M) &=  (\choi^C_{A;B})^{-1} (\ket{a_i}\bra{a_i} \otimes B^{(i)}) (\idmat_A \otimes M) \ket{\Omega}\\
    &= (\choi^C_{A;B})^{-1} (\idmat_A \otimes B^{(i)} M) (\ket{a_i}\bra{a_i} \otimes \idmat_A) \ket{\Omega} \\
    &= (\choi^C_{A;B})^{-1} (\idmat_A \otimes B^{(i)} M) (\idmat_A \otimes \ket{a_i}\bra{a_i}) \ket{\Omega} \\
    &= (\choi^C_{A;B})^{-1} ( \idmat_A \otimes B^{(i)} M \ket{a_i}\bra{a_i} ) \ket{\Omega} \\
    &= (\choi^C_{A;B})^{-1} \choi^C_{A;B}(B^{(i)} M \ket{a_i}\bra{a_i}) \\
    &= B^{(i)} M \ket{a_i}\bra{a_i}.
\end{align*}
Applying the results of these calculations term by term to the normal form in Corollary \ref{CorolarySchroedingerClassicalSemicausalGenerators} yields the first claim, where we defined $A = \tilde{A}^{T_A}$, $U = \tilde{U}\mathbb{F}_{B ; E}$, $K_A = \tilde{K}_A^T$ and $B^{(i)} = \tilde{B}^{(i)}$.\\
If the semigroup $\{\hat{S}_t\}_{t \geq 0}$ consists of superchannels, that is, preselecting maps s.t.~(by Theorem \ref{Thm:ClassicalCorrespondenceSuperchannelsSemicausalMaps}) the reduced maps $S^A_t$ of the semigroup of semicausal maps $S_t := \choi^C_{A;B} \circ \hat{S}_t \circ (\choi^C_{A;B})^{-1}$ (which are defined by the requirement that $(\idmat_A \otimes \ket{\idvec_B})S_t = S^A_t (\idmat_A \otimes \ket{\idvec_B})$) satisfy $S^A_t \ket{\idvec_A} = \ket{\idvec_A}$, then differentiating this relation yields
\begin{align*}
    Q^A\ket{\idvec_A} = \frac{d}{dt} \bigg\vert_{t = 0} S^A_t \ket{\idvec_A} = \frac{d}{dt} \bigg\vert_{t = 0} \ket{\idvec_A} = 0. 
\end{align*}
We conclude that $\hat{Q}$ generates a semigroup of superchannels if and only if $Q$ generates a semigroup of semicausal maps and $Q^A \ket{\idvec_A} = 0$. We obtain directly from Corollary \ref{CorolarySchroedingerClassicalSemicausalGenerators} that $Q^A = (\idmat_A \otimes \ket{\idvec_E})\tilde{A} - \tilde{K}_A$. It follows that
\begin{align}
    \braket{a_i \idvec_E}{\tilde{A}\idvec_A} =  \braket{a_i \idvec_E}{A^{T_A} \idvec_A} = \braket{\idvec_{AE}}{A a_i} = \braket{a_i}{\tilde{K}_A \idvec_A} = \braket{a_i}{K_A a_i},
\end{align}
where we used that $\tilde{K}_A = K_A$ is diagonal in the last step. This is the condition claimed in the theorem. Finally, \eqref{Eq:SuperchannelNormalForm} is obtained by combining this condition with \eqref{Eq:PreselectingNormalForm}.   
\end{proof}

\newpage
\section{The Quantum Case}\label{sct:quantum}

We now turn to the quantum case. As introduced and described in more detail in~\cite{Chiribella_2008}, a quantum superchannel is a map that maps quantum channels to quantum channels while preserving the probabilistic structure of the theory. To achieve the latter, it is usually required that a quantum superchannel is a linear map and that probabilistic transformations, i.e., trace non-increasing CP-maps, should be mapped to probabilistic transformations, even if we add an innocent bystander. When dealing with superchannels, we will restrict ourselves to the finite-dimensional case, and leave the infinite-dimensional case~\cite{InfiniteSuperchannels} for future work. We follow~\cite{Chiribella_2008} and define superchannels as follows:

\begin{defn}[Superchannels]
A linear map $\hat{S} : \blt(\trcl(\mathcal{H}_A); \trcl(\mathcal{H}_B)) \rightarrow \blt(\trcl(\mathcal{H}_A); \trcl(\mathcal{H}_B))$ is called a superchannel if for all $n \in \N$ the map $\hat{S}_n = \idop_{\blt(\trcl(\C^n))} \otimes \hat{S}$ satisfies that $\hat{S}_n(T)$ is a probabilistic transformation whenever $T \in \blt(\trcl(\C^n \otimes \mathcal{H}_A); \trcl(\C^n \otimes \mathcal{H}_B))$ is a probabilistic transformation and that $\hat{S}_n(T)$ is a quantum channel whenever $T \in \blt(\trcl(\C^n \otimes \mathcal{H}_A); \trcl(\C^n \otimes \mathcal{H}_B))$ is a quantum channel. 
\end{defn}

A related concept is that of a semicausal quantum channel, which is a quantum channel on a bipartite space $\mathcal{H}_A \otimes \mathcal{H}_B$ such that no communication from $B$ to $A$ is allowed. Following~\cite{Beckman.2001, Eggeling.2002}, we formalize this as follows:
\begin{defn}[Semicausality]
A bounded linear map $L_* : \trcl(\mathcal{H}_A \otimes \mathcal{H}_B) \rightarrow \trcl(\mathcal{H}_A \otimes \mathcal{H}_B)$ is called Schrödinger $B \not\to A$ semicausal if there exists $L^A_* : \trcl(\mathcal{H}_A) \rightarrow \trcl(\mathcal{H}_A)$ such that $\ptr{B}{L_*(\rho)} = L^A_*(\ptr{B}{\rho})$, for all $\rho \in \trcl(\mathcal{H}_A \otimes \mathcal{H}_B)$. Similarly, $L : \blt(\mathcal{H}_A \otimes \mathcal{H}_B) \rightarrow \blt(\mathcal{H}_A \otimes \mathcal{H}_B)$ is called Heisenberg $B \not\to A$ semicausal, if there exists $L^A : \blt(\mathcal{H}_A) \rightarrow \blt(\mathcal{H}_A)$, such that $L(X_A \otimes \idmat_B) = L^A(X_A) \otimes \idmat_B$, for all $X_A \in \blt(\mathcal{H}_A)$.     
\end{defn}
The map $L_*$ is Schrödinger $B \not \to A$ semicausal if and only if the dual map $L := L_*^*$ is normal and Heisenberg $B \not\to A$ semicausal. We will often omit the Schrödinger or Heisenberg attribute if it is clear from the context. 
This section is structured analogously to the section about the classical case. Namely, we will start by reminding the reader of the connection between semicausal maps and superchannels as well as the characterization of semicausal CP-maps in terms of semilocalizable maps, as schematically shown in Fig. \ref{fig:connecting_the_notions}. We then turn to the study of the generators of semigroups of semicausal CP-maps and finally use the correspondence between superchannels and semicausal channels to obtain the corresponding results of the generators of semigroups of superchannels. 

\begin{figure}[!ht]
    \centering
    \includegraphics[scale=1]{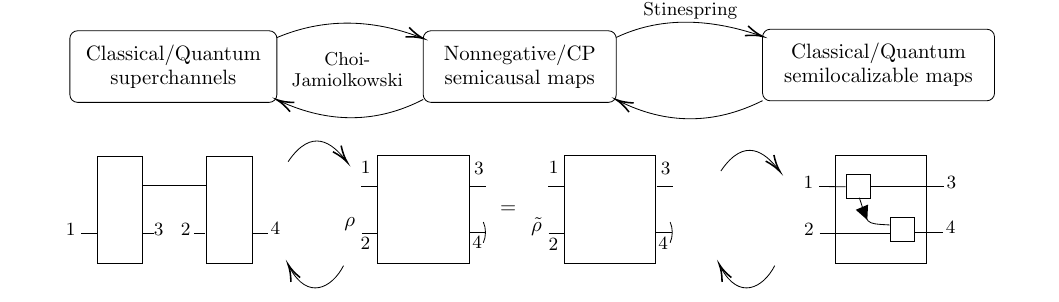}
    \caption{Visualization of the relation between the notions of superchannels, semicausal maps and semilocalizable maps. Superchannels and semicausal maps are related via a similarity transform with the Choi-Jamiołkowski isomorphism. Schrödinger $B \not\to A$ semicausal maps are those maps whose output, after tracing out system 4, does not depend on input 2 ($\rho$ or $\tilde{\rho}$). Semicausal maps are precisely those maps that allow for one-way communication only. This is called semilocalizability.  }
    \label{fig:connecting_the_notions}
\end{figure}

\subsection{Superchannels, semicausal channels, and semilocalizable channels}

We first state the characterization of superchannels in terms of semicausal maps, obtained in \cite{Chiribella_2008}: 
\begin{thm} \label{Thm:RelationQunatumSemicausalSuperchannels}
For finite-dimensional spaces $\mathcal{H}_A$ and $\mathcal{H}_B$, let $\hat{S} : \blt(\trcl(\mathcal{H}_A); \trcl(\mathcal{H}_B)) \rightarrow \blt(\trcl(\mathcal{H}_A); \trcl(\mathcal{H}_B))$
be a linear map and define $S = \choi_{A; B} \circ \hat{S} \circ \choi_{A;B}^{-1}$. Then $\hat{S}$ is a superchannel if and only if $S$ is CP and Schrödinger $B \not\to A$ semicausal such that the reduced map $S^A$ satisfies $S^A(\idmat_A) = \idmat_A$. 
\end{thm} 

The next result is due to Eggeling, Schlingemann, and Werner~\cite{Eggeling.2002}, who proved it in the finite-dimensional setting. The following form, which is a generalization of~\cite{Eggeling.2002} to the infinite-dimensional case, and which has previously been shown in~\cite[Theorem $4$]{kretschmann2005quantum}, can be obtained from our main result (Theorem \ref{Thm:SemicausalInfiniteDimMainResult}) by setting $K = 0$: 

\begin{thm} \label{Thm:QuantumSemicausalisSemilocalizable}
A map $\Phi \in CP_\sigma(\mathcal{H}_A \otimes \mathcal{H}_B)$ is Heisenberg $B \not\to A$ semicausal if and only if there exists a (separable) Hilbert space $\mathcal{H}_E$, a unitary $U \in \unitary(\mathcal{H}_E \otimes \mathcal{H}_B; \mathcal{H}_B \otimes \mathcal{H}_E)$ and arbitrary operator $A \in \blt(\mathcal{H}_A; \mathcal{H}_A \otimes \mathcal{H}_E)$, such that  
\begin{align} \label{Eq:QuantumSemilocalozability}
    \Phi(X) &= V^\dagger \left(X \otimes \idmat_E \right) V, \text{ with } V = (\idmat_A \otimes U)(A \otimes \idmat_B). 
\end{align}
If $\mathcal{H}_A$ and $\mathcal{H}_B$ are finite-dimensional, with dimensions $d_A$ and $d_B$, then $\mathcal{H}_E$ can be chosen such that $\mathrm{dim}(\mathcal{H}_E) \leq \left(d_A d_B\right)^2$. 
\end{thm}

We call a normal CP-map $\Phi \in \mathrm{CP}_\sigma(\mathcal{H}_A \otimes \mathcal{H}_B)$ \emph{semilocalizable} if its Stinespring dilation can be written in the form of Eq.~\eqref{Eq:QuantumSemilocalozability}. With that nomenclature, the above Theorem is exactly the quantum analogue of Theorem \ref{ClassicalSemicausalIsSemilocalizable}.

\subsection{Generators of semigroups of semicausal CP maps} \label{subsec:SemicausalCPSemigroups}

The main goal of this section is to establish a structure theorem for the generators of semigroups of semicausal CP-maps, the proof-structure of which is highlighted in Fig. \ref{fig:proof_structure}. This is our main technical contribution.
\begin{figure}[!ht]
    \centering
    \includegraphics[scale=1]{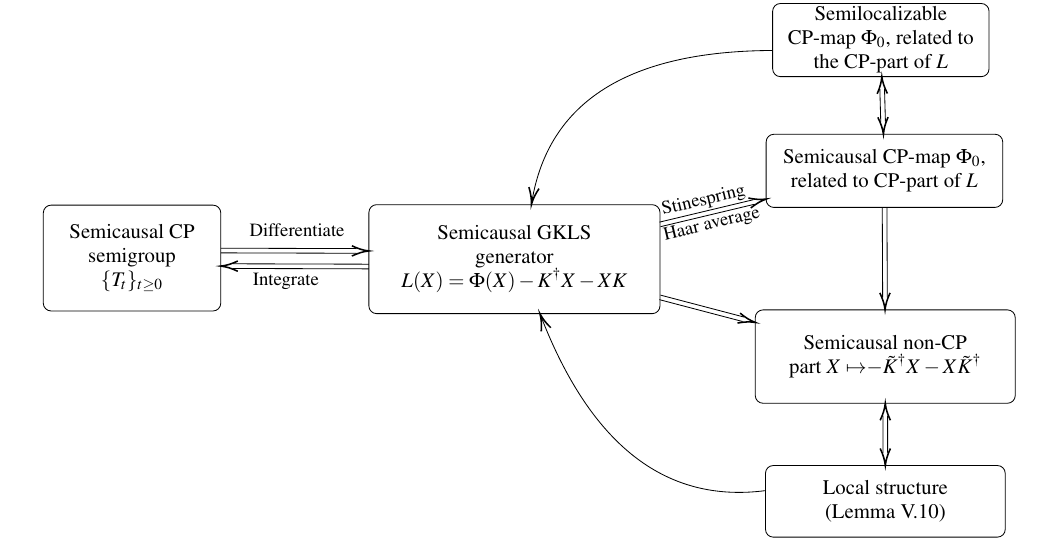}
    \caption{Overview of the proof structure leading to the normal form for semicausal Lindblad generators (Theorem \ref{Thm:SemicausalInfiniteDimMainResult}). We first observe that semicausality of the CP semigroup is equivalent to semicausality of the corresponding GKLS generator $L$. The insight is then that we can construct a CP-map $\Phi_0$, that is closely related to the CP-part of $L$ and that is semicausal (Lemma \ref{KeyLemmaInf}). From the semilocalizable form of $\Phi_0$, we then obtain an explicit form for the CP-part of $L$. This, together with the observation that a semicausal non-CP part has to have a local form, yields the desired normal form.  }
    \label{fig:proof_structure}
\end{figure} 
To get started, recall that a generator $L \in \blt(\blt(\mathcal{H}_A \otimes \mathcal{H}_B))$ generates a norm-continuous semigroup $\{ T_t \}_{t \geq 0} \subseteq \mathrm{CP}_\sigma(\mathcal{H}_A \otimes \mathcal{H}_B)$ of CP-maps (i.e., $T_t = e^{tL}$) if and only if $L$ can be written in GKLS-form, i.e., if and only if there exists $\Phi \in \mathrm{CP}_\sigma(\mathcal{H}_A \otimes \mathcal{H}_B)$ and $K \in \blt(\mathcal{H}_A \otimes \mathcal{H}_B)$ such that 
\begin{align} \label{Eq:GKLSFormInsemigroupChapter}
    L(X) = \Phi(X) - K^\dagger X - X K, \quad X \in \blt(\mathcal{H}_A \otimes \mathcal{H}_B).
\end{align}

As in the classical case, we continue by showing that $T_t$ is Heisenberg $B\not\to A$ semicausal, for all $t \geq 0$, if and only if $L$ is Heisenberg $B \not\to A$ semicausal. We start by showing that the family of reduced maps $\{ T^A_t \}_{t \geq 0}$ also forms a norm-continuous semigroup of normal CP-maps. 
That $T^A_t$ is normal and CP follows, since for any density operator $\rho_B \in \trcl(\mathcal{H}_B)$, we have
\begin{align*}
    T^A_t = \mathrm{tr}_{\rho_B} \circ T_t \circ D, 
\end{align*}
where $D \in \mathrm{CP}_\sigma(\mathcal{H}_A; \mathcal{H}_A \otimes \mathcal{H}_B)$ is defined by $D(X_A) = X_A \otimes \idmat_B$. So $T^A_t$ is a normal CP-map as composition of normal CP-maps. It remains to check the semigroup properties ($T_0^A = \idop_A$, $T^A_{t+s} = T^A_{t} T^A_{s}$ and norm-continuity). We have
\begin{align*}
    T_0^A(X_A) &= \ptr{\rho_B}{T_0(X_A \otimes \idmat_B)} = \ptr{\rho_B}{X_A \otimes \idmat_B} = X_A,\\
    T^A_{t+s}(X_A) &= \ptr{\rho_B}{T_{t+s}(X_A \otimes \idmat_B)} = \ptr{\rho_B}{T_{t}(T_{s}(X_A \otimes \idmat_B))} = \ptr{\rho_B}{T_t(T_s^A(X_A) \otimes \idmat_B)} =  \ptr{\rho_B}{(T_t^A T_s^A(X_A)) \otimes \idmat_B} = T_t^A T_s^A(X_A),\\
    \norm{T_t^A - T_s^A} &= \sup_{ \norm{X_A}_{\blt(\mathcal{H}_A)} = 1 } \norm{T_t^A(X_A) - T^A_s(X_A)}_{\blt(\mathcal{H}_A)} = \sup_{ \norm{X_A}_{\blt(\mathcal{H}_A)} = 1 } \norm{(T_t^A(X_A) - T^A_s(X_A)) \otimes \idmat_B}_{\blt(\mathcal{H}_A \otimes \mathcal{H}_B)} \\&= \sup_{ \norm{X_A}_{\blt(\mathcal{H}_A)} = 1 } \norm{T_t(X_A \otimes \idmat_B) - T_s(X_A \otimes \idmat_B)}_{\blt(\mathcal{H}_A \otimes \mathcal{H}_B)} \leq \sup_{\norm{X}_{\blt(\mathcal{H}_A\otimes \mathcal{H}_B)} = 1 } \norm{T_t(X) - T_s(X)}_{\blt(\mathcal{H}_A \otimes \mathcal{H}_B)} = \norm{T_t - T_s}.
\end{align*}
Thus, we conclude that $T_t^A = e^{t L^A}$, for some generator $L^A \in \blt(\blt(\mathcal{H}_A))$ of normal CP-maps. We further have
\begin{align*}
    L(X_A \otimes \idmat_B) = \frac{d}{dt} \bigg\vert_{t = 0} T_t(X_A \otimes \idmat_B) = \frac{d}{dt}\bigg\vert_{t = 0} T_t^A(X_A) \otimes \idmat_B = L^A(X_A) \otimes \idmat_B.
\end{align*}
Thus, $L$ is semicausal if $T_t$ is semicausal for all $t \geq 0$. Conversely, if $L$ is semicausal, then $T_t$ is semicausal for all $t \geq 0$, since 
\begin{align*}
    T_t(X_A \otimes \idmat_B) &= e^{tL}(X_A \otimes \idmat_B) \\&= \sum_{k = 0}^\infty \frac{t^k}{k!} L^k(X_A \otimes \idmat_B) \\ &= \sum_{k = 0}^\infty \frac{t^k}{k!} \left(L^A\right)^k(X_A) \otimes \idmat_B \\
    &= e^{tL^A}(X_A) \otimes \idmat_B.
\end{align*}
Therefore, our task reduces to characterizing semicausal maps in GKLS-form, i.e., we want to determine the corresponding $\Phi$ and $K$. Our main result (Theorem \ref{Thm:SemicausalInfiniteDimMainResult}) is a normal form which allows us to list all semicausal generators $L$. 

Before we delve into this, we treat the inverse question: Given some $L \in \blt(\blt(\mathcal{H}_A \otimes \mathcal{H}_B))$, is it a semicausal generator? A computationally efficiently chackable criterion can be constructed via the Choi-Jamiołkowski isomorphism. If $\mathcal{H}_A$ and $\mathcal{H}_B$ are finite-dimensional and $L \in \blt(\blt(\mathcal{H}_A \otimes \mathcal{H}_B))$ is given, then we define $\mathfrak{L} = \choi_{AB; AB}(L) \in \blt(\mathcal{H}_{A_1} \otimes \mathcal{H}_{B_1} \otimes \mathcal{H}_{A_2} \otimes \mathcal{H}_{B_2})$, where the Choi-Jamiołkowski isomorphism is defined w.r.t. the orthogonal bases $\{ \ket{a_i} \}_{i = 1}^{\mathrm{dim}(\mathcal{H}_A)} $ and $\{ \ket{b_j} \}_{j = 1}^{\mathrm{dim}(\mathcal{H}_B)}$ of $\mathcal{H}_A$ and $\mathcal{H}_B$, respectively, and where the spaces $\mathcal{H}_{A_1} = \mathcal{H}_{A_2} = \mathcal{H}_A$ and $\mathcal{H}_{B_1} = \mathcal{H}_{B_2} = \mathcal{H}_B$ are introduced for notational convenience. Furthermore, define $P^\bot \in \blt(\mathcal{H}_{A_1} \otimes \mathcal{H}_{B_1} \otimes \mathcal{H}_{A_2} \otimes \mathcal{H}_{B_2})$ to be the orthogonal projection onto the orthogonal complement of $\{\ket{\Omega}\}$, where $\ket{\Omega} = \sum_{i,j} \ket{a_i} \otimes \ket{b_j} \otimes \ket{a_i} \otimes \ket{b_j}$.

\begin{lem} \label{Lem:VerifySemicausality}
A linear map $L : \blt(\mathcal{H}_A \otimes \mathcal{H}_B) \rightarrow \blt(\mathcal{H}_A \otimes \mathcal{H}_B)$ is the generator of a semigroup of Heisenberg $B\not\to A$ semicausal CP-maps if and only if 
\begin{itemize}
    \item $\mathfrak{L}$ is self-adjoint and $P^\bot \mathfrak{L} P^\bot \geq 0$, and
    \item $\ptr{B_1}{\mathfrak{L}} = \mathfrak{L}^A \otimes \idmat_{B_2}$, for some (then necessarily self-adjoint) $\mathfrak{L}^A \in \blt(\mathcal{H}_{A_1} \otimes \mathcal{H}_{A_2})$.
\end{itemize}
The generated semigroup is unital (i.e., $T_t(\idmat_{AB}) = \idmat_{AB}$, for $t \geq 0$) if and only if $\ptr{A_1}{\mathfrak{L}^A} = 0$.

Furthermore, a linear map $L : \blt(\mathcal{H}_A \otimes \mathcal{H}_B) \rightarrow \blt(\mathcal{H}_A \otimes \mathcal{H}_B)$ is the generator of a semigroup of Schrödinger $B \not\to A$ semicausal CP-maps if and only if 
\begin{itemize}
    \item $\mathfrak{L}$ is self-adjoint and $P^\bot \mathfrak{L} P^\bot \geq 0$, and
    \item $(\mathbb{F}_{A_1;B_1} \otimes \idmat_{A_2}) \ptr{B_2}{\mathfrak{L}} (\mathbb{F}_{A_1;B_1} \otimes \idmat_{A_2}) = \idmat_{B_1} \otimes \mathfrak{L}^A$, for some (then necessarily self-adjoint) $\mathfrak{L}^A \in \blt(\mathcal{H}_{A_1} \otimes \mathcal{H}_{A_2})$.
\end{itemize}
The generated semigroup is trace-preserving (i.e., $\tr{T_t(\rho)} = \tr{\rho}$, for $\rho \in \blt(\mathcal{H}_A \otimes \mathcal{H}_B)$ and $t \geq 0$) if and only if $\ptr{A_2}{\mathfrak{L}^A} = 0$. 
\end{lem}

Thus, checking whether a map $L$ is the generator of a semigroup of semicausal CP-maps reduces to checking several semidefinite constraints. In particular, the problem to optimize over all semicausal generators is a semidefinite program.

\begin{proof}
It is known (see, e.g., the appendix in \cite{wolf2008assessing}) that $L$ generates a semigroup of CP-maps if and only if $\mathfrak{L}$ is self-adjoint and $P^\bot \mathfrak{L} P^\bot \geq 0$. This criterion goes by the name of conditional complete positivity~\cite{evans1977dilations}.
Thus, it remains to translate the other criteria to the level of Choi-Jamiołkowski operators. If $L$ is Heisenberg $B \not\to A$ semicausal, then
\begin{align*}
    \ptr{B_1}{\mathfrak{L}} &= \ptr{B_1}{ (\idop_{A_1B_1} \otimes L)(\ket{\Omega}\bra{\Omega}) } \\
    &= (\idop_{A_1} \otimes L)(\ket{\Omega_A}\bra{\Omega_A} \otimes \idmat_{B_2}) \\
    &= (\idop_{A_1} \otimes L^A)(\ket{\Omega_A}\bra{\Omega_A}) \otimes \idmat_{B_2} \\
    &= \mathfrak{L}^A \otimes \idmat_{B_2},
\end{align*}
where we defined $\ket{\Omega_A} = \sum_i \ket{a_i} \otimes \ket{a_i} \in \mathcal{H}_{A_1} \otimes \mathcal{H}_{A_2}$ and $\mathfrak{L}^A = (\idop_{A_1} \otimes L^A)(\ket{\Omega_A}\bra{\Omega_A})$. 
Conversely, if $\ptr{B_1}{\mathfrak{L}} = \mathfrak{L}^A \otimes \idmat_{B_2}$, define $L^A = \choi_{A;A}^{-1}(\mathfrak{L}^A)$. Then
\begin{align*}
    L(X_A \otimes \idmat_{B_1}) &= \ptr{A_1B_1}{\left((X_A^T \otimes \idmat_{B_1}) \otimes \idmat_{A_2B_2}\right) \mathfrak{L}} \\
    &= \ptr{A_1}{ \left(X_A^T \otimes \idmat_{A_2B_2} \right) \ptr{B_1}{\mathfrak{L}}} \\
    &= \ptr{A_1}{(X_A^T \otimes \idmat_{A_2B_2}) (\mathfrak{L}^A \otimes \idmat_{B_2})} \\
    &= \ptr{A_1}{(X_A^T \otimes \idmat_{A_2}) \mathfrak{L}^A} \otimes \idmat_{B_2}\\
    &= \choi_{A;A}^{-1}(\mathfrak{L}^A)(X_A) \otimes \idmat_{B_2} \\
    &= L^A(X_A) \otimes \idmat_{B}.
\end{align*}
Finally, it is known that a semigroup of CP-maps is unital if and only if $L(\idmat_{A_2B_2}) = 0$. But this is equivalent to our criterion, since a simple calculation shows that 
\begin{align*}
    \ptr{A_1B_1}{\mathfrak{L}} = L(\idmat_{A_2B_2}).
\end{align*}
This finishes the proof for the Heisenberg picture case. 
The Schrödinger case can be proven along similar lines, or be obtained directly from the Heisenberg case via the identity $\choi_{AB;AB}(L^*) = \mathbb{F}_{A_1B_1; A_2B_2} \left[\choi_{AB; AB}(L)\right]^T \mathbb{F}_{A_1B_1; A_2B_2}$.
\end{proof}
Let us now return to the main goal of this section: finding a normal form for semicausal generators in GKLS-form. 
We motivate (and interpret) our normal form as the `quantization' of the normal form for generators of classical semicausal semigroups (Theorem \ref{Thm:ClassicalSemicausalGeneratorDecompositionCharacterization}). 
In the classical case, the normal form had two building blocks: an operator of the form $Q_1 = \Phi_{sc} - K_A \otimes \idmat_B$, where $\Phi_{sc}$ is non-negative and semicausal and an operator of the form $Q_2 = \sum_{i = 1}^{\abs{\mathbb{A}}} \ket{a_i}\bra{a_i} \otimes B^{(i)}$, where the $B^{(i)}$'s are generators of row-stochastic maps, (i.e., $B^{(i)}$ generates a non-negative semigroup and $B^{(i)} \ket{\idvec_B} = 0$). 
It is straightforward to guess a quantum analogue for the first building block: a generator $L_1 \in \blt(\blt(\mathcal{H}_A \otimes \mathcal{H}_B))$ defined by
\begin{align} \label{Eq:QuantumFirstBuildingBlock}
    L_1(X) = \Phi_{sc}(X) - (K_A \otimes \idmat_B)^\dagger X - X (K_A \otimes \idmat_B),
\end{align}
where $\Phi_{sc} \in \mathrm{CP}_\sigma(\mathcal{H}_A \otimes \mathcal{H}_B)$, given in Stinespring form by $\Phi_{sc}(X) = V_{sc}^\dagger (X \otimes \idmat_E) V_{sc}$, is semicausal. One readily verifies that $L_1$ defines a semicausal generator. To `quantize' the second building block, note that $Q_2$ does not induce any change on system $A$. Indeed, since 
\begin{align} \label{Eq:OneSystemPartExponentialClassically}
    e^{tQ_2}(\idmat_A \otimes \ket{\idvec_B}) = \sum_{i = 1}^{\abs{\mathbb{A}}} \ket{a_i}\bra{a_i} \otimes (e^{tB^{(i)}}\ket{\idvec_B}) = \sum_{i = 1}^{\abs{\mathbb{A}}} \ket{a_i} \bra{a_i} \otimes \ket{\idvec_B} = \idmat_A \otimes \ket{\idvec_B},
\end{align}
the generated semigroup looks like the identity on system $A$.
In the quantum case, semigroups that do not induce any change on system $A$ are more restricted, since any information-gain about system $A$ inevitably disturbes system $A$ - so there can be no conditioning as in the classical case. Indeed, if one requires that $T_t \in \mathrm{CP}_\sigma(\mathcal{H}_A \otimes \mathcal{H}_B)$ satisfies the quantum analogue of Eq. \eqref{Eq:OneSystemPartExponentialClassically}, namely 
\begin{align} \label{Eq:IdentityReducedMapCondition}
    T_t(X_A \otimes \idmat_B) = X_A \otimes \idmat_B,
\end{align}
for all $X_A \in \blt(\mathcal{H}_A)$, then $T_t = \idop_A \otimes \Theta_t$ for some unital map $\Theta_t \in \mathrm{CP}_\sigma(\mathcal{H}_B)$, see Appendix \ref{Ap:InformationDisturbanceLemma} for a proof.
Differentiation of $T_t = \idop_A \otimes \Theta_t$ at $t = 0$ now implies that the generator of a semigroup of CP-maps that satisfy \eqref{Eq:IdentityReducedMapCondition} are of the form $\idop_A \otimes \hat{B}$, where $\hat{B}$ generates a semigroup of unital CP-maps (i.e., $\hat{B}(\idmat_B) = 0$). 
To conclude, the two building blocks are operators of the form of $L_1$ in Eq. \eqref{Eq:QuantumFirstBuildingBlock} and maps $L_2$ of the form
\begin{align*}
    L_2(X) = (\idmat_A \otimes B)^\dagger (X \otimes \idmat_E) (\idmat_A \otimes B) - \frac{1}{2} \acomu{\idmat_A \otimes B^\dagger B}{X} + i\comu{\idmat_A \otimes H_B}{X},
\end{align*}
with $B \in \blt(\mathcal{H}_B; \mathcal{H}_B \otimes \mathcal{H}_E)$ and a self-adjoint $H_B \in \blt(\mathcal{H}_B)$. 

In the classical case, we obtained the normal form (Theorem \ref{Thm:ClassicalSemicausalGeneratorDecompositionCharacterization}) by taking a convex combination of the basic building blocks. This corresponds to probabilistically choosing one or the other. In quantum theory, there is is a more general concept: superposition. To account for this, we construct our normal form not as a convex combination of the maps $L_1$ and $L_2$, but by taking a linear combination (superposition) of the Stinespring operators $V_{sc}$ and $\idmat_A \otimes B$ as the Stinespring operator of the CP-part of the GKLS-form (note here that the coefficients can be absorbed into $V_{sc}$ and $\idmat_A \otimes B$, respectively). This means that if $L$ is given by Eq. \eqref{Eq:GKLSFormInsemigroupChapter} with $\Phi(X) = V^\dagger (X \otimes \idmat_E) V$, then we take $V = V_{sc} + \idmat_A \otimes B$. It turns out that $K$ can then be chosen such that $L$ becomes semicausal. Also note that we can further decompose $V_{sc} = (\idmat_A \otimes U)(A \otimes \idmat_B)$, as in Theorem \ref{Thm:QuantumSemicausalisSemilocalizable}. 

Our main technical result is that the heuristics employed in the `quantization' procedure above is sound, i.e., that the generators constructed in that way are the only semicausal generators in GKLS-form.

\begin{thm} \label{Thm:SemicausalInfiniteDimMainResult}
Let $L : \blt(\mathcal{H}_A \otimes \mathcal{H}_B) \rightarrow \blt(\mathcal{H}_A \otimes \mathcal{H}_B)$ be defined by $L(X) = \Phi(X) - K^\dagger X - X K$, with $\Phi \in \mathrm{CP}_\sigma(\mathcal{H}_A \otimes \mathcal{H}_B)$ and $K \in \blt(\mathcal{H}_A \otimes \mathcal{H}_B)$. Then $L$ is Heisenberg $B \not\to A$ semicausal if and only if there exists a (separable) Hilbert space $\mathcal{H}_E$, a unitary $U \in \unitary(\mathcal{H}_E \otimes \mathcal{H}_B; \mathcal{H}_B \otimes \mathcal{H}_E)$, a self-adjoint operator $H_B \in \blt(\mathcal{H}_B)$, and arbitrary operators $A \in \blt(\mathcal{H}_A; \mathcal{H}_A \otimes \mathcal{H}_E)$, $B \in \blt(\mathcal{H}_B; \mathcal{H}_B \otimes \mathcal{H}_E)$ and $K_A \in \blt(\mathcal{H}_A)$, such that  
\begin{subequations}
\begin{align} 
    \Phi(X) &= V^\dagger \left(X \otimes \idmat_E \right) V, \text{ with } V = (\idmat_A \otimes U)(A \otimes \idmat_B) + (\idmat_A \otimes B),  \label{eq:GeneralFormInfDimA}\\ 
    K &=  (\idmat_A \otimes B^\dagger U) (A \otimes \idmat_B) + \frac{1}{2} \idmat_A \otimes B^\dagger B + K_A \otimes \idmat_B + \idmat_A \otimes iH_B.\label{eq:GeneralFormInfDimB}
\end{align}
\end{subequations}
If $\mathcal{H}_A$ and $\mathcal{H}_B$ are finite-dimensional, with dimensions $d_A$ and $d_B$, then $\mathcal{H}_E$ can be chosen such that $\mathrm{dim}(\mathcal{H}_E) \leq \left(d_A d_B\right)^2$.
\end{thm}

\begin{rem}
    Note that the characterization in Theorem~\ref{Thm:SemicausalInfiniteDimMainResult} is for generators of Heisenberg $B \not\to A$ semicausal dynamical semigroups. There are two special cases of interest: 
    First, if we want the dynamical semigroup to be unital, then we need to further impose $L(\idmat_{A}\otimes\idmat_B) = 0$ in the normal form above, which is equivalent to $A^\dagger A = K_A + K_A^\dagger$ -- a constraint that also appears in the usual Linblad form. Second, if the dynamical semigroup corresponds (in the sense of Theorem~\ref{Thm:RelationQunatumSemicausalSuperchannels}) to a semigroup of superchannels, then we additionally require that the reduced generator satisfies $L^A_* (\idmat_{A})=0$. We will use this in the ``translation step'' in Theorem~\ref{Thm:CharacterizationGeneratorsSuperchannels}.
\end{rem}

\begin{rem}
In the finite-dimensional case the proof of Theorem \ref{Thm:SemicausalInfiniteDimMainResult} is constructive. In Appendix \ref{Appe:ComputationalConstruction} we discuss in detail how to obtain the operators $A$, $U$, $K_A$, $B$ and $H_B$ starting from the conditions in Lemma \ref{Lem:VerifySemicausality}. 
\end{rem}

The remainder of this section is devoted to the proof of Theorem \ref{Thm:SemicausalInfiniteDimMainResult}, whose structure is highlighted in Fig.~\ref{fig:proof_structure}. 

We begin with a technical observation about certain Haar integrals. 

\begin{lem} \label{HaarMeasureIntegrationLemma}
Let $\mathcal{H}_n$ be an $n$-dimensional subspace of $\mathcal{H}_A$ with orthogonal projection $P_n \in \blt(\mathcal{H}_A)$ and let $V \in \blt(\mathcal{H}_A\otimes \mathcal{H}_B; \mathcal{H}_A\otimes \mathcal{H}_C)$. Then
\begin{align} \label{FiniteUnitaryHaarIntegral}
    \int_{\unitary_P(\mathcal{H}_n)} (U \otimes \idmat_C)\, V\, (U^\dagger\otimes \idmat_B) \,dU =  P_n \otimes \frac{1}{n} \ptr{P_n}{V},
\end{align}
where the integration is w.r.t. the Haar measure on $\unitary_P(\mathcal{H}_n)$. It follows that $\norm{P_n \otimes \frac{1}{n} \ptr{P_n}{V}} \leq \norm{V}$. \\
Furthermore if $\mathcal{H}$ is separable infinite-dimensional, with orthonormal basis $\{\ket{e_i}\}_{i \in \N}$ and $\mathcal{H}_n = \mathrm{span}\{\ket{e_1}, \ket{e_2}, \dots, \ket{e_n}\}$, then there exists $B \in \blt(\mathcal{H}_B;\mathcal{H}_C)$ and an ultraweakly convergent subsequence of $\left(P_n \otimes \frac{1}{n} \ptr{P_n}{V}\right)_{n \in \N}$ with limit $\idmat_A \otimes B$.
\end{lem}
\begin{proof}
To calculate the integral, we employ the Weingarten formula \cite{8178732, Collins2006, Fukuda_2019}, which for the relevant case reads:
\begin{align*}
    \int_{\unitary_P(\mathcal{H}_n)} U_{i\,j} U^\dagger_{j^\prime\, i^\prime} \,dU = \frac{1}{n} \delta_{i\, i^\prime} \delta_{j\, j^\prime},  
\end{align*}
where $U_{i\,j} = \braket{f_i}{U f_j}$ and $U^\dagger_{j^\prime\,i^\prime} = \braket{f_{j^\prime}}{U^\dagger f_{i^\prime}}$, for some orthonormal basis $\{\ket{f_1}, \ket{f_2}, \dots, \ket{f_n}\}$ of $\mathcal{H}_n$. A basis expansion then yields
\begin{align*}
    \int_{\unitary_P(\mathcal{H}_n)} (U \otimes \idmat_C)\, V\, (U^\dagger\otimes \idmat_B) \,dU = \sum_{i, j, i^\prime, j^\prime = 1}^n \left[ \ket{f_{i}}\bra{f_{i^\prime}} \otimes \left( \left(\bra{f_{j}} \otimes \idmat_C\right) \, V \, \left(\ket{f_{j^\prime}} \otimes \idmat_B\right)\right)  \int_{\unitary_P(\mathcal{H}_n)} U_{i\,j} U^\dagger_{j^\prime\, i^\prime} \,dU \right] = P_n \otimes \frac{1}{n} \ptr{P_n}{V}.
\end{align*}
For the second claim, we note that a standard estimate of the integral yields $\norm{\frac{1}{n} \ptr{P_n}{V}} = \norm{P_n \otimes \frac{1}{n} \ptr{P_n}{V}} \leq \norm{V}$. Thus the sequence $\left(\frac{1}{n} \ptr{P_n}{V}\right)_{n \in \N}$ is bounded and hence, by Banach-Alaoglu, has an ultraweakly convergent subsequence, whose limit we call $B$. The claim then follows by observing that, under the separability assumption, $\left(P_n\right)_{n \in \N}$ converges ultraweakly to $\idmat_A$ and that the tensor product of two ultraweakly convergent sequences converges ultraweakly.  
\end{proof}

As a first step towards our main result, we provide a characterization of those semicausal Lindblad generators that can be written with vanishing CP part.

\begin{lem} \label{Lem:CharactierizationWithoutCPPart}
Let $L : \blt(\mathcal{H}_A \otimes \mathcal H_B) \rightarrow \blt(\mathcal{H}_A \otimes \mathcal{H}_B)$, $L(X) := -K^\dagger X - XK$, with $K \in \blt(\mathcal{H}_A \otimes \mathcal{H}_B)$. Then $L$ is Heisenberg $B \not\to A$ semicausal if and only if there exist $K_A \in \blt(\mathcal{H}_A)$ and a self-adjoint $H_B \in \blt(\mathcal{H}_B)$, with $K = K_A \otimes \idmat_B + \idmat_A \otimes iH_B$.  
\end{lem}
\begin{proof}
If $K = K_A \otimes \idmat_B + \idmat_A \otimes iH_B$, then $L(X_A \otimes \idmat_B) = (- K_A^\dagger X_A - K_A X_A) \otimes \idmat_B + X_A \otimes (iH_B-iH_B) = (- K_A^\dagger X_A - X_A K_A) \otimes \idmat_B$. Hence, $L$ is semicausal. Conversely, suppose $L$ is semicausal with $L(X_A \otimes \idmat_B) = L^A(X_A)\otimes\idmat_B$. Let $\mathcal{H}_n$ be an $n$-dimensional subspace of $\mathcal{H}_A$ and $U \in \unitary_P(\mathcal{H}_n)$. Then
\begin{align*}
    \left(L(U\otimes \idmat_B)\right)(U^\dagger\otimes \idmat_B) = - K^\dagger (P_n\otimes\idmat_B) - (U\otimes \idmat_B) K (U^\dagger \otimes \idmat_B) = (L^A(U)\,U^\dagger)\otimes \idmat_B,
\end{align*} 
where $P_n \in \blt(\mathcal{H}_A)$ is the orthogonal projection onto $\mathcal{H}_n$. 
We integrate both sides w.r.t.~the Haar measure on $\unitary_P(\mathcal{H}_n)$. Lemma \ref{HaarMeasureIntegrationLemma} and some rearrangement and taking the conjugate yields
\begin{align} \label{eq:OnlyKPartProofApproxEquation}
    (P_n\otimes\idmat_B) K  = - P_n \otimes \frac{1}{n} \ptr{P_n}{K^\dagger} - L^A_n \otimes \idmat_B,
\end{align}
for some operator $L^A_n \in \blt(\mathcal{H}_A)$. If $\mathcal{H}_A$ is finite-dimensional, we can take $\mathcal{H}_n = \mathcal{H}_A$, so that $P_n = \idmat_A$. Hence $K = - \tilde{K}_A \otimes \idmat_B - \idmat_A \otimes B$, with $B = \frac{1}{n} \ptr{A}{K^\dagger}$ and $\tilde{K}_A = L^A_n$. If $\mathcal{H}_A$ is separable infinite-dimensional, we obtain the same result via a limiting procedure $n \rightarrow \infty$ as follows: Let $\{\ket{e_i}\}_{i \in \N}$ be an orthonormal basis of $\mathcal{H}_A$ and set $\mathcal{H}_n = \mathrm{span}\{\ket{e_1}, \ket{e_2}, \dots, \ket{e_n}\}$. Then, the second part of Lemma \ref{HaarMeasureIntegrationLemma} allows us to pass to a subsequence of $\left(P_n \otimes \frac{1}{n} \ptr{P_n}{K^\dagger}\right)_{n \in \N}$ that converges ultraweakly to a limit $\idmat_A \otimes B$. The corresponding subsequence of $\left((P_n \otimes \idmat_B)K\right)_{n \in N}$ converges ultraweakly to $K$, and hence that subsequence of $\left( L_n^A \otimes \idmat_B \right)_{n \in \N}$ converges ultraweakly to a limit $\tilde{K}_A \otimes \idmat_B$. I.e., we get $K=  -\tilde{K}_A\otimes\idmat_B - \idmat_A\otimes B$. Therefore,
\begin{align*}
    0 = L(X_A \otimes \idmat_B) - L(X_A \otimes \idmat_B) = (L^A(X_A) - \tilde{K}_A^\dagger X_A - X_A \tilde{K}_A) \otimes \idmat_B - X_A \otimes (B + B^\dagger),
\end{align*}
which can only be true for all $X_A$, if $B + B^\dagger$ is proportional to $\idmat_B$. Since $B + B^\dagger$ is self-adjoint, we have $B + B^\dagger = 2r\idmat_B$, for some $r \in \R$. We can then set $iH_B := r\idmat_B - B$ and $K_A := -\tilde{K}_A - r\idmat$, so that $H_B$ is self-adjoint and $K = K_A \otimes \idmat + \idmat \otimes iH_B$. 
\end{proof}

If we had restricted our attention to Hamiltonian generators and unitary groups in finite dimensions, an analog of this Lemma would have already followed from the fact that semicausal unitaries are tensor products, which was proved in \cite{Beckman.2001} (and reproved in \cite{Piani.2006}).\\

As another technical ingredient, the following lemma establishes a closedness property of the set of semicausal maps. 

\begin{lem} \label{WeakStarClosednesslemma}
Let $\left(V_m\right)_{m \in \N}$ and $\left(W_n\right)_{n \in \N}$ be ultraweakly convergent sequences in $\blt(\mathcal{H}_A \otimes \mathcal{H}_B; \mathcal{H}_A \otimes \mathcal{H}_B \otimes \mathcal{H}_E)$, with limits $V$ and $W$. Suppose that for all $m,n \in \N$, the map $\Phi_{m, n} : \blt(\mathcal{H}_A \otimes \mathcal{H}_B) \rightarrow \blt(\mathcal{H}_A \otimes \mathcal{H}_B)$, defined by $\Phi_{m, n}(X) = V_m^\dagger (X \otimes \idmat_E) W_n$ is Heisenberg $B \not\to A$ semicausal. Then the map $\Phi : \blt(\mathcal{H}_A \otimes \mathcal{H}_B) \rightarrow \blt(\mathcal{H}_A \otimes \mathcal{H}_B)$, defined by $\Phi(V) = V^\dagger (X \otimes \idmat_E) W$, is also Heisenberg $B \not\to A$ semicausal.
\end{lem}
\begin{proof}
For $X_A \in \blt(\mathcal{H}_A)$ and $\rho \in \trcl(\mathcal{H}_A \otimes \mathcal{H}_B)$ we have that $\rho\, V_m^\dagger (X_A \otimes \idmat_B \otimes \idmat_E) \in \trcl(\mathcal{H}_A \otimes \mathcal{H}_B \otimes \mathcal{H}_E; \mathcal{H}_A \otimes \mathcal{H}_B)$, since the trace-class operators are an ideal in the bounded operators. Hence, by definition of the ultraweak topology,
\begin{align*}
    \tr{\rho\, V_m^\dagger (X_A \otimes \idmat_B \otimes \idmat_E) W} = \lim_{n \rightarrow \infty} \tr{\rho\, V_m^\dagger (X_A \otimes \idmat_B \otimes \idmat_E) W_n} = \lim_{n \rightarrow \infty} \tr{\rho\, (\Phi^A_{m, n}(X_A) \otimes \idmat_B)}.
\end{align*}
Since $\tr{\rho\, \Phi^A_{m, n}(X_A) \otimes \idmat_B}$ converges as $n\to\infty$ for every $\rho\in\trcl(\mathcal{H}_A \otimes \mathcal{H}_B)$, the sequence $\left(\Phi^A_{m,n}(X_A) \otimes \idmat_B\right)_{n \in \N}$ converges ultraweakly~\footnote{Uniqueness of such a limit is clear. Existence follows by the Banach-Alaoglu Theorem and an application of the uniform boundedness priciple, which implies that the sequence $\Phi^A_{m,n}(X_A)$ is norm-bounded}. We call the limit $\Phi^A_{m}(X_A) \otimes \idmat_B$. It is then easy to see that $\Phi^A_{m}(X_A)$, viewed as a map on $\blt(\mathcal{H}_A)$ is linear and continuous. This tells us that the map $\Phi_{m} : \blt(\mathcal{H}_A \otimes \mathcal{H}_B) \rightarrow \blt(\mathcal{H}_A \otimes \mathcal{H}_B)$, defined by $\Phi_{m}(X) = V_m^\dagger (X \otimes \idmat_E) W$ is semicausal for all $m \in \N$. Furthermore, we have that $\rho^\dagger\, W^\dagger (X_A^\dagger \otimes \idmat_B \otimes \idmat_E) \in \trcl(\mathcal{H}_A \otimes \mathcal{H}_B \otimes \mathcal{H}_E; \mathcal{H}_A \otimes \mathcal{H}_B)$ for all $X_A \in \blt(\mathcal{H}_A)$ and $\rho \in \trcl(\mathcal{H}_A \otimes \mathcal{H}_B)$ and thus
\begin{align*}
    \tr{\rho\, V^\dagger (X_A \otimes \idmat_B \otimes \idmat_E) W} &= \overline{\tr{\rho^\dagger W^\dagger (X_A^\dagger \otimes \idmat_B \otimes \idmat_E) V}} = \lim_{m \rightarrow \infty} \overline{\tr{\rho^\dagger W^\dagger (X_A^\dagger \otimes \idmat_B \otimes \idmat_E) V_m}} = \lim_{m \rightarrow \infty} \tr{\rho\, V_m^\dagger (X_A \otimes \idmat_B \otimes \idmat_E) W} \\&= \lim_{m \rightarrow \infty} \tr{\rho\, (\Phi^A_m(X_A) \otimes \idmat_E)}.
\end{align*}
Repeating the argument above then shows that $\Phi$ is semicausal. 
\end{proof}

As a final preparatory step we observe that, given a semicausal Lindblad generator, we can use its CP part to define a family of semicausal CP-maps.

\begin{lem} \label{SemicausalPolarizationLemmaInfd}
Let $L : \blt(\mathcal{H}_A \otimes \mathcal H_B) \rightarrow \blt(\mathcal{H}_A \otimes \mathcal{H}_B)$ be defined by $L(X) := V^\dagger (X \otimes \idmat_E) V - K^\dagger X - X K$, with $V \in \blt(\mathcal{H}_A \otimes \mathcal{H}_B; \mathcal{H}_A \otimes \mathcal{H}_B \otimes \mathcal{H}_E)$ and $K \in \blt(\mathcal{H}_A \otimes \mathcal{H}_B)$. If $L$ is Heisenberg $B \not\to A$ semicausal, then the map $S_{Y,Z} : \blt(\mathcal{H}_A \otimes \mathcal{H}_B) \rightarrow \blt(\mathcal{H}_A \otimes \mathcal{H}_B)$, defined by
\begin{align*}
    S_{Y,Z}(X) = \left[V (Z \otimes \idmat_B) - (Z \otimes \idmat_B \otimes \idmat_E)V\right]^\dagger\; (X \otimes \idmat_E)\; \left[V (Y \otimes \idmat_B) - (Y \otimes \idmat_B \otimes \idmat_E)V\right],
\end{align*} 
is Heisenberg $B \not\to A$ semicausal for every $Y, Z \in \blt(\mathcal{H}_A)$.
\end{lem}
\begin{proof}
For every $M \in \blt(\mathcal{H}_A \otimes \mathcal{H}_B)$, we define the map $\Psi_M : \blt(\mathcal{H}_A \otimes \mathcal{H}_B) \rightarrow \blt(\mathcal{H}_A \otimes \mathcal{H}_B)$ by
\begin{align*}
    \Psi_M(X) &= L(M^\dagger XM) - M^\dagger L(XM) - L(M^\dagger X)M + M^\dagger L(X)M
    \\ &= \left[ (M \otimes \idmat_E)V - V M \right]^\dagger (X \otimes \idmat_E) \left[ (M \otimes \idmat_E) V - V M \right].
\end{align*}
This map has already been used, for a different purpose, in Lindblad's original work \cite[Eq. 5.1]{Lindblad.1976}.
It follows from the semicausality of $L$ that, if we choose $M = M_A \otimes \idmat_B$, for some $M_A \in \blt(\mathcal{H}_A)$, then $\Psi_M$ is semicausal. 
Furthermore, a calculation shows that
\begin{align*}
    \frac{1}{4}\sum_{k = 0}^3  i^k \Psi_{M+i^k N}(X) = \left[V N - (N \otimes \idmat_E)V\right]^\dagger\; (X\otimes \idmat_E)\; \left[V M - (M \otimes \idmat_E)V\right].
\end{align*}
By choosing $N = Z \otimes \idmat_B$ and $M = Y \otimes \idmat_B$, it follows that $S_{Y,Z}$ is the linear combination of four semicausal maps, and hence is itself semicausal. 
\end{proof}

We now combine this Lemma with an integration over the Haar measure to obtain the key Lemma in our proof.

\begin{lem} \label{KeyLemmaInf}
Let $L : \blt(\mathcal{H}_A \otimes \mathcal H_B) \rightarrow \blt(\mathcal{H}_A \otimes \mathcal{H}_B)$ be defined by $L(X) := V^\dagger (X \otimes \idmat_E) V - K^\dagger X - X K$, with $V \in \blt(\mathcal{H}_A \otimes \mathcal{H}_B; \mathcal{H}_A \otimes \mathcal{H}_B \otimes \mathcal{H}_E)$ and $K \in \blt(\mathcal{H}_A \otimes \mathcal{H}_B)$. If $L$ is Heisenberg $B \not\to A$ semicausal, then there exists $B \in \blt(\mathcal{H}_B; \mathcal{H}_B \otimes \mathcal{H}_E)$ such that the map $S : \blt(\mathcal{H}_A \otimes \mathcal{H}_B) \rightarrow \blt(\mathcal{H}_A \otimes \mathcal{H}_B)$, defined by
\begin{align*}
    S(X) = \left[V - \idmat_A \otimes B\right]^\dagger (X \otimes \idmat_E) \left[V - \idmat_A \otimes B\right],
\end{align*} 
is also Heisenberg $B \not\to A$ semicausal. \\Furthermore, if $\mathcal{H}_A$ is finite-dimensional, then we can choose $B = \ptr{A}{V}/\mathrm{dim}(\mathcal{H}_A)$.
\end{lem}

\begin{proof}
Let $\mathcal{H}_n$ and $\mathcal{H}_m$ be $n$ and $m$ dimensional subspaces of $\mathcal{H}_A$ with respective orthogonal projections $P_n \in \blt(\mathcal{H}_A)$ and $P_m \in \blt(\mathcal{H}_A)$. Since for every $U \in \unitary_P(\mathcal{H}_n)$ and $W \in \unitary_P(\mathcal{H}_m)$, the map $S_{U,W}$, defined in Lemma \ref{SemicausalPolarizationLemmaInfd} is semicausal, also the map $\overline{S} : \blt(\mathcal{H}_A \otimes \mathcal{H}_B) \rightarrow \blt(\mathcal{H}_A \otimes \mathcal{H}_B)$, defined by 
\begin{align*}
    \overline{S}(X) := \int_{\unitary_P(\mathcal{H}_n)} \int_{\unitary_P(\mathcal{H}_m)} (U \otimes \idmat_B) S_{U,W}(X) (W^\dagger \otimes \idmat_B)\; dW dU,
\end{align*}
is semicausal. Writing out the definition of $S_{U,W}$ yields 
\begin{align*}
     \overline{S}(X) &= \left[V (P_n \otimes \idmat_B) - \int_{\unitary_P(\mathcal{H}_n)} (U \otimes \idmat_B \otimes \idmat_E) V (U^\dagger \otimes \idmat_E) dU \right]^\dagger (X \otimes \idmat_E) \left[V (P_m \otimes \idmat_B) - \int_{\unitary_P(\mathcal{H}_m)} (W \otimes \idmat_B \otimes \idmat_E) V (W^\dagger \otimes \idmat_B) dW \right] \\
    &= \left[V(P_n \otimes \idmat_B) - P_n \otimes \frac{1}{n} \ptr{P_n}{V}  \right]^\dagger (X \otimes \idmat_E) \left[V(P_m \otimes \idmat_B) - P_m \otimes \frac{1}{m} \ptr{P_m}{V}  \right],
\end{align*}
where the last line was obtained by using Lemma \ref{HaarMeasureIntegrationLemma}. If $\mathcal{H}_A$ is finite-dimensional, we can choose $\mathcal{H}_n = \mathcal{H}_m = \mathcal{H}_A$, so that $P_n = P_m = \idmat_A$, and obtain the desired result immediately. If $\mathcal{H}_A$ is separable infinite-dimensional and $\{\ket{e_i}\}_{i \in \N}$ is an orthonormal basis and $\mathcal{H}_k := \mathrm{span}\{\ket{e_1}, \ket{e_2}, \dots, \ket{e_k}\}$, then, by Lemma  \ref{HaarMeasureIntegrationLemma}, the sequence $\left(P_k \otimes \frac{1}{k}\ptr{P_k}{V} \right)_{k \in \N}$ has an ultraweakly convergent subsequence with a limit $\idmat_A \otimes B$, where $B \in \blt(\mathcal{H}_B; \mathcal{H}_B \otimes \mathcal{H}_E)$. Furthermore, since $\left(P_k\right)_{k \in \N}$ converges ultraweakly to $\idmat_A$, we have that the sequence $\left( V(P_k \otimes \idmat_B) - P_k \otimes \frac{1}{k} \ptr{P_k}{V} \right)_{k \in \N}$ has a subsequence that converges ultraweakly to $V - \idmat_A \otimes B$. Hence, by passing to subsequences, we can apply Lemma \ref{WeakStarClosednesslemma}, which yields that $S$ is semicausal.
\end{proof}

\begin{rem}
The previous two lemmas are at the heart of our result. They illustrate a (to the best of our knowledge) novel technique that allows to characterize GKLS generators with a certain constraint, if this constraint is well understood for completely positive maps. It seems useful to develop this method more generally, but this is beyond the scope of the present work.
\end{rem}

With these tools at hand, we can now prove our main result.

\begin{proof} (Theorem \ref{Thm:SemicausalInfiniteDimMainResult})
A straightforward calculation shows that $L$, defined via \eqref{eq:GeneralFormInfDimA} and \eqref{eq:GeneralFormInfDimB} is semicausal. To prove the converse, note that by the Stinespring dilation theorem, there exist a separable Hilbert space $\tilde{\mathcal{H}}_E$ and $\tilde{V} \in \blt(\mathcal{H}_A \otimes \mathcal{H}_B; \mathcal{H}_A \otimes \mathcal{H}_B \otimes \tilde{\mathcal{H}}_E)$, such that $\Phi(X) = \tilde{V}^\dagger (X \otimes \idmat_E) \tilde{V}$. It is well known (see, e.g., \cite[Thm. 2.1 and Thm. 2.2]{WolfQuantumChannels}) that if $\mathcal{H}_A$ and $\mathcal{H}_B$ are finite-dimensional with dimensions $d_A$ and $d_B$, then $\tilde{\mathcal{H}}_E$ can be chosen such that $\mathrm{dim}(\tilde{\mathcal{H}}_E) \leq \left(d_Ad_B\right)^2$. By Lemma \ref{KeyLemmaInf}, there exists $\tilde{B} \in \blt(\mathcal{H}_B; \mathcal{H}_B \otimes \tilde{\mathcal{H}}_E)$ such that the map $\Phi_0 \in \mathrm{CP}_\sigma(\mathcal{H}_A \otimes \mathcal{H}_B)$, defined by $\Phi_0(X)=\left[\tilde{V} - \idmat_A \otimes \tilde{B}\right]^\dagger (X \otimes \idmat_E) \left[\tilde{V} - \idmat_A \otimes \tilde{B}\right]$ is semicausal. We define $V_{sc} = \tilde{V} - \idmat \otimes \tilde{B}$ and obtain 
\begin{align*}
    \Phi(X_A \otimes \idmat_B) = \Phi_0(X_A \otimes \idmat_B) + \kappa^\dagger(X_A \otimes \idmat_B) + (X_A \otimes \idmat_B)\kappa,
\end{align*}
where $\kappa =  (\idmat_A \otimes \tilde{B}^\dagger)V_{sc} + \frac{1}{2}\left( \idmat_A \otimes \tilde{B}^\dagger \tilde{B} \right)$. Since $L$ and $\Phi_0$ are semicausal, we can write $L(X_A \otimes \idmat) = L^A(X_A) \otimes \idmat_B$ and $\Phi_0(X_A \otimes \idmat_B) = \Phi_0^A(X_A) \otimes \idmat_B$, for all $X_A \in \blt(\mathcal{H}_A)$. Hence,
\begin{align} \label{Eq:EqToFindFormOfKSc}
     L(X_A \otimes \idmat_B) - \Phi_0(X_A \otimes \idmat_B) = \left(L^A(X_A) - \Phi_0^A(X_A)\right) \otimes \idmat_B = -(K - \kappa)^\dagger(X_A \otimes \idmat_B) - (X_A \otimes \idmat_B)(K - \kappa). 
\end{align}
It follows that the map defined by $X \mapsto -(K - \kappa)^\dagger X - X(K - \kappa)$ is semicausal. Thus, Lemma \ref{Lem:CharactierizationWithoutCPPart} implies that there exist $K_A \in \blt(\mathcal{H}_A)$ and a self-adjoint $H_B \in \blt(\mathcal{H}_B)$, such that $K - \kappa = K_A \otimes \idmat + \idmat \otimes iH_B$. \\
What we have achieved so far is that $\tilde{V} = V_{sc} + \idmat \otimes \tilde{B}$ and $K = (\idmat_A \otimes \tilde{B}^\dagger)V_{sc} + \frac{1}{2} \idmat \otimes \tilde{B}^\dagger\tilde{B} + K_A \otimes \idmat + \idmat \otimes iH_B$. So, if we can decompose $V_{sc} = (\idmat_A \otimes U)(A \otimes \idmat_B)$, then we are basically done. But this decomposition is given (up to details) by the equivalence between semicausal and semilocalizable channels \cite{Eggeling.2002}. Since the conclusion in \cite{Eggeling.2002} was in the finite-dimensional setting, we will repeat the argument here, showing that it goes through also for infinite-dimensional spaces, while paying special attention to the dimensions of the spaces involved. 
Since $\Phi_0 \in \mathrm{CP}_\sigma(\mathcal{H}_A \otimes \mathcal{H}_B)$ and $\Phi_0(X_A \otimes \idmat_B) = \Phi_0^A(X_A) \otimes \idmat_B$, we also have $\Phi_0^A \in \mathrm{CP}_\sigma(\mathcal{H}_A)$. By the Stinespring dilation theorem (for normal CP-maps), there exist a separable Hilbert space $\mathcal{H}_F$ and $W \in \blt(\mathcal{H}_A; \mathcal{H}_A \otimes \mathcal{H}_F)$ such that $\Phi_0^A(X_A) = W^\dagger (X_A \otimes \idmat_F) W$ and such that $\mathrm{span}\left\{ (X_A \otimes \idmat_F)W \ket{\psi} \mid X_A \in \blt(\mathcal{H}_A), \ket{\psi} \in \mathcal{H}_A \right\}$ is dense in $\mathcal{H}_A \otimes \mathcal{H}_F$. The last condition is called the minimality condition. We then get
\begin{align*}
    V_{sc}^\dagger (X_A \otimes \idmat_B \otimes \idmat_{\tilde{E}}) V_{sc} = (W \otimes \idmat_B)^\dagger (X_A \otimes \idmat_{F} \otimes \idmat_B) (W \otimes \idmat_B)
\end{align*}
Clearly, $\mathrm{span}\left\{ (X_A \otimes \idmat_{F} \otimes \idmat_B)(W \otimes \idmat_B) \ket{\psi} \mid X_A \in \blt(\mathcal{H}_A), \ket{\psi} \in \mathcal{H}_A \otimes \mathcal{H}_B \right\}$ is dense in $\mathcal{H}_A \otimes \mathcal{H}_F \otimes \mathcal{H}_B$. Thus, by minimality, there exists an isometry $\tilde{U} \in \blt(\mathcal{H}_F \otimes \mathcal{H}_B; \mathcal{H}_B \otimes \tilde{\mathcal{H}}_E)$, such that $V_{sc} = (\idmat_A \otimes \tilde{U})(W \otimes \idmat_B)$. 
In the finite-dimensional case, the fact that $\tilde{U}$ is an isometry then implies that $\mathrm{dim}(\mathcal{H}_F) \leq \dim(\tilde{\mathcal{H}}_E)$, such that we can think of $\mathcal{H}_F$ as a subspace of $\tilde{\mathcal{H}}_E$. Thus, $\tilde{U}$ can be extended to a unitary operator $\hat{\tilde{U}} \in \unitary(\tilde{\mathcal{H}}_E \otimes \mathcal{H}_B; \mathcal{H}_B \otimes \tilde{\mathcal{H}}_E)$. Then, defining $\mathcal{H}_E = \tilde{\mathcal{H}}_E$, $U = \hat{\tilde{U}}$, $B = \tilde{B}$ and $A = W$ proves the claim in this case. In the infinite-dimensional case, we can take $\mathcal{H}_E = \mathcal{H}_F \oplus \tilde{\mathcal{H}}_E$. We can now view both $\tilde{\mathcal{H}}_E \otimes \mathcal{H}_B$ and $\mathcal{H}_F \otimes \mathcal{H}_B$ as closed subspaces of $\mathcal{H}_E \otimes \mathcal{H}_B$. Then $\left(\tilde{U}(\mathcal{H}_F \otimes \mathcal{H}_B)\right)^\bot$ and $\left(\mathcal{H}_F \otimes \mathcal{H}_B\right)^\bot$ are isomorphic. Hence $\tilde{U}$ can be extended to a unitary operator $\hat{\tilde{U}} \in \unitary(\mathcal{H}_E \otimes \mathcal{H}_B; \mathcal{H}_B \otimes \mathcal{H}_E)$. We finish the proof by defining $U = \hat{\tilde{U}}$, $B = (\idmat_B \otimes \idmat_{\tilde{E} \to E}) \tilde{B}$ and $A = (\idmat_A \otimes \idmat_{F \to E}) W$ where $\idmat_{\tilde{E} \to E}$ and $\idmat_{F \to E}$ denote the isometric embeddings of $\tilde{\mathcal{H}}_E$ and $\mathcal{H}_F$ into $\mathcal{H}_E$, respectively.  
\end{proof}

As a first consequence, we obtain the analogous theorem for semigroups of Schrödinger $B \not\to A$ semicausal CP-maps:

\begin{cor} \label{cor:SchrodingerSemicausalForm}
Let $L : \trcl(\mathcal{H}_A \otimes \mathcal{H}_B) \rightarrow \trcl(\mathcal{H}_A \otimes \mathcal{H}_B)$ be defined by $L(\rho) = \Phi_S(\rho) - K \rho - \rho K^\dagger$, where $\Phi_S \in \mathrm{CP}_S(\mathcal{H}_A \otimes \mathcal{H}_B)$, with $\Phi_S(\rho) = \ptr{E}{V \rho V^\dagger}$ and $K \in \blt(\mathcal{H}_A \otimes \mathcal{H}_B)$. Then $L$ is Schrödinger $B \not\to A$ semicausal, if and only if, $K$, $V$ and $\mathcal{H}_E$ can be chosen as in \eqref{eq:GeneralFormInfDimA} and \eqref{eq:GeneralFormInfDimB}. 
\end{cor}

As a further corollary, we translate the results above to the familiar representation in terms of jump-operators (by going from Stinespring to Kraus).

\begin{cor}\label{Cor:NormalFormSemicausalLindbladGenerator}
A map $L : \trcl(\mathcal{H}_A \otimes \mathcal{H}_B) \rightarrow \trcl(\mathcal{H}_A \otimes \mathcal{H}_B)$ generates a (trace-)norm-continuous semigroup of trace-preserving Schrödinger $B \not\to A$ semicausal CP-maps, if and only if there exists $\{\phi_j\}_j \subset \blt(\mathcal{H}_A \otimes \mathcal{H}_B)$, $\{B_j\}_j \subset \blt(\mathcal{H}_B)$ and $H_A\in\blt(\mathcal{H}_A)$ and $H_B\in\blt(\mathcal{H}_B)$ such that $\{\phi_j\}_j$ is a set of Kraus operators of a Schrödinger $B \not\to A$ semicausal CP-map and $\{B_j\}_j$ is a set of Kraus operators of some CP-map, such that
\begin{align*}
    L(\rho) &= - i\comu{H_A \otimes \idmat_B + \idmat_A \otimes H_B}{\rho} \\ &+ \sum_j \left(\phi_j + \idmat_A \otimes B_j\right) \rho \left(\phi_j + \idmat_A \otimes B_j\right)^\dagger - \frac{1}{2}\acomu{\idmat_A \otimes B_j^\dagger B_j + \phi_j^\dagger \phi_j}{ \rho } - (\idmat_A \otimes B_j^\dagger)\phi_j \rho - \rho \phi_j^\dagger (\idmat_A \otimes B_j) 
\end{align*}
\end{cor}
\begin{proof}
A simple calculation by defining the Kraus operators as $(\idmat_{AB} \otimes \ket{e_i})V$, with $\{\ket{e_j}\}_j$ an orthonormal basis of $\mathcal{H}_E$ and $V$ given by Theorem \ref{Thm:SemicausalInfiniteDimMainResult}.  
\end{proof}

We conclude this section about semicausal semigroups with an example that uses our normal form in full generality. 

\begin{figure}[!ht]
    \centering
    \includegraphics[scale = 0.5]{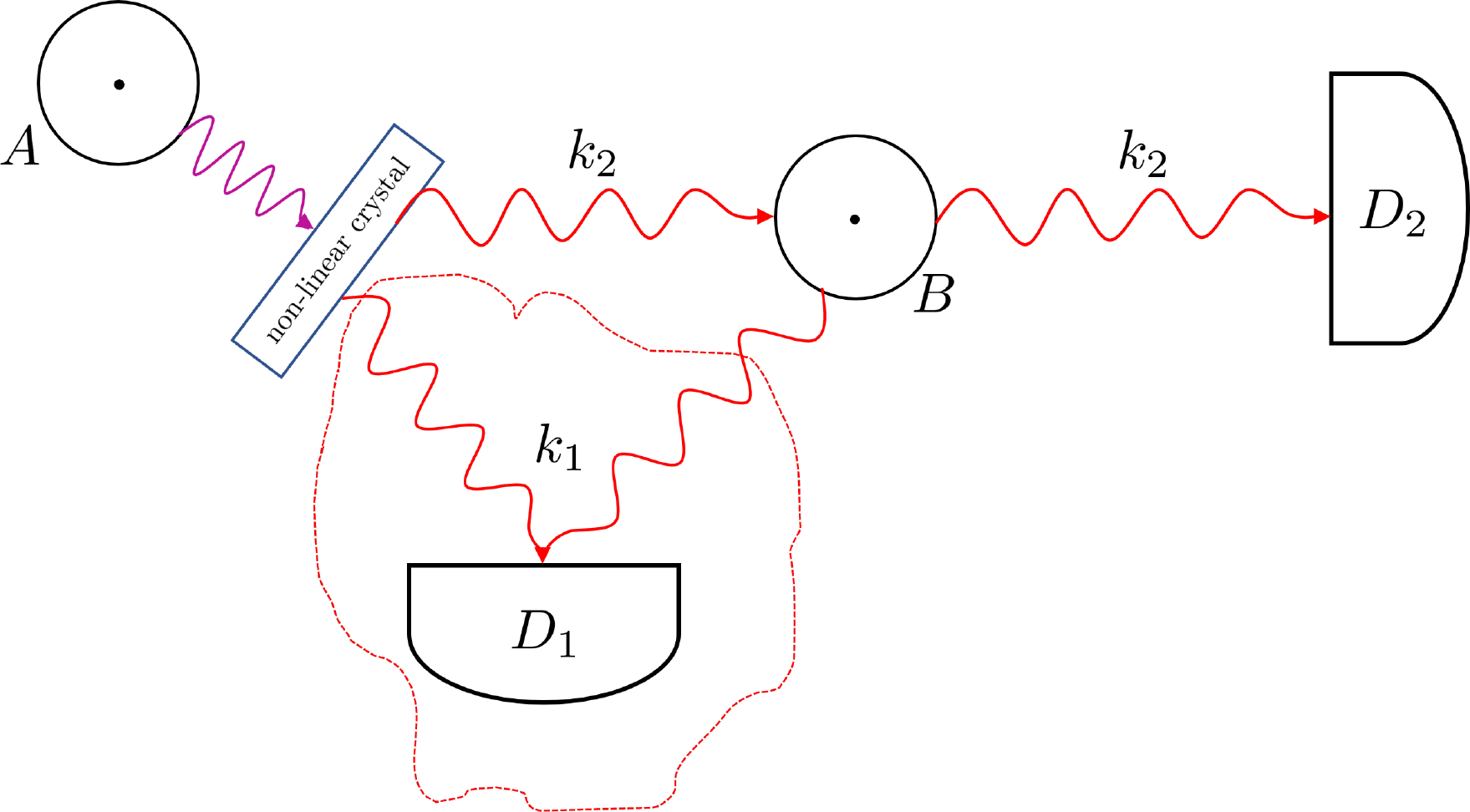}
    \caption{Systems $A$ and $B$ describe $2$-level systems, respectively. The allowed interactions are infinitesimally described as follows: If $A$ is in its excited state, it can emit a photon. Through parametric down-conversion, the photon is converted into two photons (of lower energy). One of those two photons, $k_1$, is sent to a detector $D_1$. The other, $k_2$, is sent to $B$. If $B$ is in its ground state, it absorbs $k_2$. If $B$ is in its excited state, it cannot absorb $k_2$, so $k_2$ passes through $B$ and travels to a detector $D_2$. Additionally, in this case, $B$ can emit a photon, indistinguishable from $k_1$, to $D_1$. }
    \label{fig:example}
\end{figure}

\begin{ex}
We consider the scenario of two $2$-level atoms that can interact according to the processes specified in Fig.~\ref{fig:example}.
We can describe this process either via a dilation (as in Theorem \ref{Thm:SemicausalInfiniteDimMainResult}) or via the Kraus operators (as in Corollary \ref{Cor:NormalFormSemicausalLindbladGenerator}). In the dilation picture, we introduce an auxiliary Hilbert space $\mathcal{H}_E:=\mathcal{H}_1\otimes\mathcal{H}_2$, where $\mathcal{H}_i$ is for the $i^{\textrm{th}}$ photon. Then, the process is described by $V=(\idmat_A\otimes U)(A\otimes \idmat_{B}) + (\idmat_{A}\otimes B)$, with
\begin{align*}
    A\in\blt\left(\mathcal{H}_A; \mathcal{H}_A\otimes\mathcal{H}_E \right),~A &= \ket{0}\bra{1}_A \otimes \ket{11}_E,\\
    B\in\blt\left(\mathcal{H}_B; \mathcal{H}_B\otimes\mathcal{H}_E  \right),~B &= \ket{10}_E\otimes \ket{0}\bra{1}_B,\\
    U\in\unitary(\mathcal{H}_E \otimes \mathcal{H}_B; \mathcal{H}_B \otimes \mathcal{H}_E),~U &= \mathbb{F}_{E;B}( \idmat_{\mathcal{H}_1}\otimes \tilde{U}), 
\end{align*}
where $\tilde{U}\in\unitary\left(\mathcal{H}_2\otimes \mathcal{H}_B \right)$ is determined by
\begin{align*}
    \tilde{U}\ket{00}_{\mathcal{H}_2 B} = \ket{00}_{\mathcal{H}_2 B},~
    \tilde{U}\ket{10}_{\mathcal{H}_2 B} = \ket{01}_{\mathcal{H}_2 B},~
    \tilde{U}\ket{11}_{\mathcal{H}_2 B} = \ket{11}_{\mathcal{H}_2 B}.
\end{align*}
The crucial feature of this example is that the CP-part of the generator ($\ptr{E}{V \cdot V^\dagger}$) cannot be written as a convex combination of the two building blocks ($\Phi_{sc}$ and $\idop_A \otimes \hat{B}$). As mentioned also in the quantization procedure before, this is a pure quantum feature and stems from the fact that it cannot be determined if a photon arriving at the detector $D_1$ came from $B$ or $A$. Hence, the system remains in a superposition state.  \\ 
We can also look at the usual representation via jump operators. This can be achieved by switching from dilations to Kraus operators. We obtain the two jump-operators
\begin{align*}
    L_1 := \underbrace{L_e\otimes L_a}_{=:\phi_1} + \idmat_A\otimes \underbrace{L_e}_{B_1},~
    L_2 := \underbrace{L_e\otimes\ket{1}\bra{1}}_{=:\phi_2},
\end{align*}
where $L_e = \ket{0}\bra{1}$ and $L_a = L_e^\dagger$ describe emission and absorption of a photon, respectively. Thus, the usual Lindblad equation reads:
\begin{align*}
    \frac{d\rho}{dt} = (L_e \otimes L_a + \idmat_A \otimes L_e) \rho (L_e \otimes L_a + \idmat_A \otimes L_e) + (\idmat_A \otimes L_e) \rho (\idmat_A \otimes L_e) - \frac{1}{2} \acomu{\idmat_A \otimes L_e^\dagger L_e + L_e^\dagger L_e \otimes \idmat_B}{\rho}
\end{align*}
It is also possible and instructive to consider the reduced dynamics on system $A$, which can also be described by a Lindblad equation, since $B$ does not communicate to $A$ (this is not true otherwise):
\begin{align*}
    \frac{d \rho_A}{dt} = L_e \rho_A L_e^\dagger - \frac{1}{2} \acomu{L_e^\dagger L_e}{\rho_A},
\end{align*}
where $\rho_A(t) = \ptr{B}{\rho(t)}$. Not surprisingly (given our model), this describes an atom emitting photons. 
\end{ex}

\subsection{Generators of semigroups of quantum superchannels} \label{subsec:QuantumSuperchannels}

We finally turn to semigroups of quantum superchannels (on finite-dimensional spaces), that is, a collection of quantum superchannels $\{\hat{S}_t\}_{t \geq 0} \subseteq \blt(\blt(\blt(\mathcal{H}_A);\blt( \mathcal{H}_B)))$, such that $\hat{S}_0 = \idop$, $\hat{S}_{t+s} = \hat{S}_t \hat{S}_s$ and the map $t \mapsto \hat{S}_t$ is continuous (w.r.t.~any and thus all of the equivalent norms on the finite-dimensional space $\blt(\blt(\blt(\mathcal{H}_A);\blt( \mathcal{H}_B)))$). To formulate a technically slightly stronger result, we call a map $\hat{S} \in \blt(\blt(\blt(\mathcal{H}_A);\blt( \mathcal{H}_B)))$ a preselecting supermap if $\choi_{A;B} \circ \hat{S} \circ \choi^{-1}_{A;B}$ is a Schrödinger $B \not\to A$ semicausal CP-map. Theorem \ref{Thm:RelationQunatumSemicausalSuperchannels} then tells us that a superchannel is a special preselecting supermap. Again, as for semicausal CP-maps, we characterize the generators of semigroups of preselecting supermaps and superchannels in two ways: First, we answer how to determine if a given map $\hat{L} \in \blt(\blt(\blt(\mathcal{H}_A);\blt( \mathcal{H}_B)))$ is such a generator. Second, we provide a normal form for all generators. 

The answer to the first question is really a corollary of Lemma \ref{Lem:VerifySemicausality} together with Theorem \ref{Thm:RelationQunatumSemicausalSuperchannels}. To this end, define $\hat{\mathfrak{L}} := \choi_{AB; AB}\left( \choi_{A;B} \circ \hat{L} \circ \choi^{-1}_{A;B} \right) \in \blt(\mathcal{H}_{A_1} \otimes \mathcal{H}_{B_1} \otimes \mathcal{H}_{A_2} \otimes \mathcal{H}_{B_2})$, where we fix some orthonormal bases $\{\ket{a_i}\}_{i = 1}^{\mathrm{dim}(\mathcal{H}_A)}$ and $\{\ket{b_j}\}_{j = 1}^{\mathrm{dim}(\mathcal{H}_B)}$ of $\mathcal{H}_A$ and $\mathcal{H}_B$, such that $\choi_{A;B}$ is defined w.r.t.~$\{\ket{a_i}\}_{i = 1}^{\mathrm{dim}(\mathcal{H}_A)}$ and $\choi_{AB; AB}$ is defined w.r.t.~the product of the two bases. Furthermore, we introduced the spaces $\mathcal{H}_{A_1} = \mathcal{H}_{A_2} = \mathcal{H}_A$ and $\mathcal{H}_{B_1} = \mathcal{H}_{B_2} = \mathcal{H}_A$ for notational convenience. Finally, we define $P^\bot \in \blt(\mathcal{H}_{A_1} \otimes \mathcal{H}_{B_1} \otimes \mathcal{H}_{A_2} \otimes \mathcal{H}_{B_2})$ to be the orthogonal projection onto the orthogonal complement of $\{\ket{\Omega}\}$, where $\ket{\Omega} = \sum_{i,j} \ket{a_i} \otimes \ket{b_j} \otimes \ket{a_i} \otimes \ket{b_j}$. We then have

\begin{lem} \label{thm:checkforGeneratorSuperchannel}
A linear map $\hat{L} \in \blt(\blt(\blt(\mathcal{H}_A);\blt( \mathcal{H}_B)))$ generates a semigroup of quantum superchannels if and only if 
\begin{itemize}
    \item $\hat{\mathfrak{L}}$ is self-adjoint and $P^\bot \hat{\mathfrak{L}} P^\bot \geq 0$,
    \item $(\mathbb{F}_{A_1;B_1} \otimes \idmat_{A_2}) \ptr{B_2}{\hat{\mathfrak{L}}} (\mathbb{F}_{A_1;B_1} \otimes \idmat_{A_2}) = \idmat_{B_1} \otimes \hat{\mathfrak{L}}^A$, for some (then necessarily self-adjoint) $\hat{\mathfrak{L}}^A \in \blt(\mathcal{H}_{A_1} \otimes \mathcal{H}_{A_2})$ and
    \item $\ptr{A_1}{\hat{\mathfrak{L}}^A} = 0$.
\end{itemize}
$\hat{L}$ is preselecting if and only if the first two conditions hold.  
\end{lem}
\begin{proof}
Theorem \ref{Thm:RelationQunatumSemicausalSuperchannels} tells us that $\{\hat{S}_t\}_{t \geq 0}$ forming a semigroup of superchannels is eqiuvalent to $S_t = \choi_{A;B} \circ \hat{S}_t \circ \choi_{A;B}^{-1}$ forming a semigroup of Schrödinger $B \not\to A$ semicausal CP-maps and that the reduced map $S_t^A$ satisfies $S_t^A(\idmat_A) = \idmat_A$. By Lemma \ref{Lem:VerifySemicausality} the semicausal semigroup property is equivalent to the first two conditions in the statement. This proves the claim about preselecting $\hat{L}$.\\ 
By differentiation, it follows that $S_t^A(\idmat_A) = \idmat_A$ is satisfied if and only if $L^A$, the generator of $\{S_t^A\}_{t \geq 0}$, satisfies $L^A(\idmat_A) = 0$. But since $\ptr{A_1}{\hat{\mathfrak{L}}^A} = L^A(\idmat_A)$, the claim follows. 
\end{proof}

We finally turn to a normal form for generators of semigroups of preselecting supermaps and superchannels:

\begin{thm} \label{Thm:CharacterizationGeneratorsSuperchannels}
A linear map $\hat{L} : \blt(\blt(\mathcal{H}_A); \blt(\mathcal{H}_B)) \rightarrow \blt(\blt(\mathcal{H}_A); \blt(\mathcal{H}_B))$ generates a semigroup of hyper-preselecting supermaps if and only if there exist a Hilbert space $\mathcal{H}_E$, a state $\sigma \in \blt(\mathcal{H}_E)$, a unitary $U \in \unitary(\mathcal{H}_B \otimes \mathcal{H}_E)$, a self-adjoint operator $H_B \in \blt(\mathcal{H}_B)$, and arbitrary operators $A \in \blt(\mathcal{H}_A \otimes \mathcal{H}_E)$, $B \in \blt(\mathcal{H}_B \otimes \mathcal{H}_E)$ and $K_A \in \blt(\mathcal{H}_A)$, such that $\hat{L}$ acts on $T \in \blt(\blt(\mathcal{H}_A); \blt(\mathcal{H}_B))$ as $\hat{L}(T) = \hat{\Phi}(T) - \hat{\kappa}_L(T) - \hat{\kappa}_R(T)$, with
\begin{align} \label{Eq:CPPartTranslation}
    \begin{split}
    \hat{\Phi}(T)(\rho) = &\ptr{E}{U\; (T \otimes \idop_E)(A (\rho \otimes \sigma) A^\dagger) \;U^\dagger} + \ptr{E}{B\; (T \otimes \idop_E)((\rho \otimes \sigma) A^\dagger) \;U^\dagger} \\&+ \ptr{E}{U\; (T \otimes \idop_E)(A (\rho \otimes \sigma)) \;B^\dagger} + \ptr{E}{ B\; (T \otimes \idop_E)((\rho \otimes \sigma)) \;B^\dagger},
    \end{split}
\end{align}
\begin{subequations}
\begin{align}
    \hat{\kappa}_L(T)(\rho) &= \ptr{E}{B^\dagger U\; (T \otimes \idop_E)(A (\rho \otimes \sigma))} + \frac{1}{2} \ptr{E}{B^\dagger B (T\otimes \idop_E)(\rho\otimes \sigma)}\;  + T(K_A\,\rho) + iH_B\, T(\rho), \\
    \hat{\kappa}_R(T)(\rho) &= \ptr{E}{(T \otimes \idop_E)((\rho \otimes \sigma)A^\dagger)\;U^\dagger B} +  \frac{1}{2}   \ptr{E}{(T\otimes \idop_E)(\rho\otimes \sigma) B^\dagger B}\; + T(\rho\,K_A^\dagger) -  T(\rho) \,iH_B.
\end{align}
\end{subequations}
We can choose $\sigma$ to be pure and $\mathcal{H}_E$ with $\mathrm{dim}(\mathcal{H}_E) \leq (d_A d_B)^2$, where $d_A$ and $d_B$ are the dimensions of $\mathcal{H}_A$ and $\mathcal{H}_B$, respectively.  \\
Furthermore, $\hat{L}$ generates a semigroup of superchannels if and only if $\hat{L}$ generates a semigroup of preselecting supermaps and $\ptr{\sigma}{A^\dagger A} = K_A + K_A^\dagger$. In that case, we can split $\hat{L}$ into a dissipative part $\hat{D}$ and a \sq{Hamiltonian} part $\hat{H}$, i.e., a part which generates a (semi-)group of invertible superchannels whose inverses are superchannels as well. We have $\hat{L}(T) = \hat{D}(T) + \hat{H}(T)$, with
\begin{align*}
    \hat{D}(T)(\rho) = \ptr{E}{\hat{D}^\prime(T)(\rho)} \quad \text{and} \quad \hat{H}(T)(\rho) = -i\comu{H_B}{T(\rho)} - iT(\comu{H_A}{\rho}),
\end{align*}
where $H_A$ is the imaginary part of $K_A$, where
\begin{subequations} \label{Eq:SuperchannelGenerators}
\begin{align}
    \hat{D}^\prime(T)(\rho) &= U (T \otimes \idop_E)(A(\rho \otimes \sigma)A^\dagger) U^\dagger& &- \frac{1}{2}  (T \otimes \idop_E)(\acomu{A^\dagger A} {\rho \otimes \sigma})& \\
    &+ B (T \otimes \idop_E)(\rho \otimes \sigma) B^\dagger& &- \frac{1}{2} \acomu{B^\dagger B}{ (T \otimes \idop_E)(\rho \otimes \sigma)}& \\
    &+ \comu{U (T \otimes \idop_E)(A(\rho \otimes \sigma))}{B^\dagger}& &+ \comu{B}{(T \otimes \idop_E)((\rho \otimes \sigma)A^\dagger)U^\dagger},&
\end{align}
\end{subequations}
and where $[\cdot,\cdot]$ and $\{\cdot,\cdot\}$ denote the commutator and anticommutator, respectively.
\end{thm}

\begin{rem}
Similarly to Theorem \ref{Thm:SemicausalInfiniteDimMainResult}, the proof of Theorem \ref{Thm:CharacterizationGeneratorsSuperchannels} is constructive. In Appendix \ref{Appe:ComputationalConstruction2} we discuss in detail how to obtain the operators $A$, $U$, $K_A$, $B$, $H_A$ and $H_B$ starting from the conditions in Theorem \ref{thm:checkforGeneratorSuperchannel}. 
\end{rem}

As in the classical case, the proof strategy is to use the relation between superchannels and semicausal channels and Theorem \ref{Thm:SemicausalInfiniteDimMainResult}. As this translation process is more involved than in the classical case, we need two auxilliary lemmas.  

\begin{lem} \label{Lem:TranslationLemma}
Let $S : \blt(\mathcal{H}_A \otimes \mathcal{H}_B) \rightarrow \blt(\mathcal{H}_A \otimes \mathcal{H}_B)$ be given by
\begin{align}
    S(X) = \ptr{E}{(\idmat_A \otimes L_B)(L_A \otimes \idmat_B) X (R_A^\dagger\otimes \idmat_B)(\idmat_A \otimes R_B^\dagger)},
\end{align}
with Hilbert spaces $\mathcal{H}_C$ and $\mathcal{H}_E$, operators $L_A, R_A \in \blt(\mathcal{H}_A; \mathcal{H}_A \otimes \mathcal{H}_C)$ and $L_B, R_B \in \blt(\mathcal{H}_C \otimes \mathcal{H}_B; \mathcal{H}_B \otimes \mathcal{H}_E)$. Then, for $T \in \blt(\blt(\mathcal{H}_A); \blt(\mathcal{H}_B))$ and $\rho \in \blt(\mathcal{H}_A)$, 
\begin{align}
    \left[\choi_{A;B}^{-1} \circ S \circ \choi_{A;B}\right](T)(\rho) = \ptr{E}{V_L (T\otimes \idop_C)\left(W_L \rho W_R^\dagger \right) V_R^\dagger},
\end{align}
with $V_L = L_B \flip_{B;C}$, $V_R = R_B \flip_{B;C}$; and $W_L = L_A^{T_A}$, $W_R = R_A^{T_A}$. Here, the partial transpose on $\mathcal{H}_A$ is taken w.r.t. the basis used to define the Choi–Jamiołkowski isomorphism.
\end{lem}
\begin{proof}
The proof is a direct calculation. We present it in detail in Appendix \ref{TranslatingSupermapsAppendix}. 
\end{proof}

\begin{lem} \label{Lem:PartialTransposeCalculation}
Let $X \in \blt(\mathcal{H}_A \otimes \mathcal{H}_C; \mathcal{H}_A \otimes \mathcal{H}_B)$, $Y \in \blt(\mathcal{H}_A \otimes \mathcal{H}_B; \mathcal{H}_A \otimes \mathcal{H}_C)$, $\rho \in \trcl(\mathcal{H}_B)$. Then $\ptr{\rho}{XY}^T = \ptr{C}{Y^{T_A}(\idmat_A \otimes \rho)X^{T_A}}$.
\end{lem}
\begin{proof}
The proof is a direct calculation. We present it in detail in Appendix \ref{TranslatingSupermapsAppendix}. 
\end{proof}

We are finally ready to prove Theorem \ref{Thm:CharacterizationGeneratorsSuperchannels}

\begin{proof} (Theorem \ref{Thm:CharacterizationGeneratorsSuperchannels})
The idea is to relate the generators of superchannels to semicausal maps. This relation is given by definition for preselecting superamps and by Theorem \ref{Thm:RelationQunatumSemicausalSuperchannels} for superchannels. For a generator $\hat{L}$ of a semigroup of preselecting supermaps $\{\hat{S_t}\}_{t \geq 0}$, we have
\begin{align*}
    \hat{L} =  \choi_{A;B}^{-1} \circ \frac{d}{dt}\bigg\vert_{t = 0} [\choi_{A;B} \circ \hat{S}_t \circ \choi_{A;B}^{-1}] \circ \choi_{A;B}.
\end{align*}
Thus $\hat{L}$ generates a semigroup of preselecting supermaps if and only if $\hat{L}$ can be written as $\hat{L} = \choi_{A;B}^{-1} \circ L \circ \choi_{A;B}$ for some generator $L$ of a semigroup of Schrödinger $B \not\to A$ semicausal CP-maps. Thus to prove the first part of our theorem, we can take the normal form in Corollary \ref{cor:SchrodingerSemicausalForm} and compute the similarity transformation above. 
We now execute this in detail. To start with, Corollary \ref{cor:SchrodingerSemicausalForm} tells us that $L(\rho) = \Phi_S(\rho) - K \rho - \rho K^\dagger$, where
\begin{subequations} \label{Eq:GeneralFormInProof}
\begin{align} 
    \Phi_S(\rho) &= \ptr{E}{V \rho V^\dagger}, \text{ with } V = (\idmat_A \otimes \tilde{U})(\tilde{A} \otimes \idmat_B) + (\idmat_A \otimes \tilde{B}),  \label{eq:GeneralFormInfDimA}\\ 
    K &=  (\idmat_A \otimes \tilde{B}^\dagger \tilde{U}) (\tilde{A} \otimes \idmat_B) + \frac{1}{2} \idmat_A \otimes \tilde{B}^\dagger \tilde{B} + \tilde{K}_A \otimes \idmat_B + \idmat_A \otimes i\tilde{H}_B.\label{eq:GeneralFormInfDimB}
\end{align}
\end{subequations}
for some unitary $\tilde{U} \in \unitary(\mathcal{H}_E \otimes \mathcal{H}_B; \mathcal{H}_B \otimes \mathcal{H}_E)$, some self-adjoint $\tilde{H}_B \in \blt(\mathcal{H}_B)$ and some operators $\tilde{A} \in \blt(\mathcal{H}_A; \mathcal{H}_A \otimes \mathcal{H}_E)$, $\tilde{B} \in \blt(\mathcal{H}_B; \mathcal{H}_B \otimes \mathcal{H}_E)$ and $\tilde{K}_A \in \blt(\mathcal{H}_A)$. 
In order to apply Lemma \ref{Lem:TranslationLemma}, we fix a unit vector $\ket{\xi} \in \mathcal{H}_E$ and define $\Xi_A := \idmat_A \otimes \ket{\xi} \in \blt(\mathcal{H}_A; \mathcal{H}_A \otimes \mathcal{H}_E)$ and $\Xi_B := \ket{\xi} \otimes \idmat_B \in \blt(\mathcal{H}_B; \mathcal{H}_E \otimes \mathcal{H}_B)$, so that $\idmat_A \otimes \tilde{B} = (\idmat_A \otimes \tilde{B}\Xi_B^\dagger)(\Xi_A \otimes \idmat_B)$. We can then write 
\begin{align*}
    \Phi_S(\rho) &= \ptr{E}{(\idmat_A \otimes \tilde{U})(\tilde{A} \otimes \idmat_B) \rho (\tilde{A}^\dagger \otimes \idmat_B)(\idmat_A \otimes U^\dagger)} + \ptr{E}{(\idmat_A \otimes \tilde{B}\Xi_B^\dagger)(\Xi_A \otimes \idmat_B) \rho (\Xi_A^\dagger \otimes \idmat_B)(\idmat_A \otimes \Xi_B\tilde{B}^\dagger)} \\
    &+ \ptr{E}{(\idmat_A \otimes \tilde{U})(\tilde{A} \otimes \idmat_B)\rho(\Xi_A^\dagger \otimes \idmat_B)(\idmat_A \otimes \Xi_B\tilde{B}^\dagger)} + \ptr{E}{(\idmat_A \otimes \tilde{B}\Xi_B^\dagger)(\Xi_A \otimes \idmat_B)\rho(\tilde{A}^\dagger \otimes \idmat_B)(\idmat_A \otimes U^\dagger)},
\end{align*}
which is an expression suitable for a term by term application of Lemma \ref{Lem:TranslationLemma}. Doing so yields
\begin{align*}
    \hat{\Phi}(T)(\rho) &:= (\choi^{-1}_{A;B} \circ \Phi_S \circ \choi_{A;B})(T)(\rho) \\
    &= \ptr{E}{U\; (T \otimes \idop_E)(A (\rho \otimes \sigma) A^\dagger) \;U^\dagger} + \ptr{E}{B\; (T \otimes \idop_E)((\rho \otimes \sigma) A^\dagger) \;U^\dagger} \\&+ \ptr{E}{U\; (T \otimes \idop_E)(A (\rho \otimes \sigma)) \;B^\dagger} + \ptr{E}{ B\; (T \otimes \idop_E)((\rho \otimes \sigma)) \;B^\dagger},
\end{align*}
where we defined $U := \tilde{U}\flip_{B;E}$, $B := \tilde{B}\Xi_B^\dagger \flip_{B;E}$, $A := \tilde{A}^{T_A} \Xi_A^\dagger$ and $\sigma := \ket{\xi}\bra{\xi}$. This proves Equation \eqref{Eq:CPPartTranslation}. Similarly, upon defining $\kappa_L(\rho) := K \rho$ we can write~\footnote{The partial trace $\ptr{\C}{\cdot}$ over the one-dimensional space $\C$ is just to ensure formal similarity with Lemma \ref{Lem:TranslationLemma}.} 
\begin{align*}
    \kappa_L(\rho) &= \ptr{E}{(\idmat_A \otimes \flip_{E;B} \Xi_B\tilde{B}^\dagger\tilde{U})(\tilde{A} \otimes \idmat_B) \rho (\Xi_A^\dagger \otimes \idmat_B)(\idmat_A \otimes \flip_{B;E})} + \ptr{E}{(\idmat_A \otimes \flip_{E;B} \Xi_B \tilde{B}^\dagger \tilde{B} \Xi_B^\dagger) (\Xi_A \otimes \idmat_B) \rho (\Xi_A^\dagger \otimes \idmat_B)(\idmat_A \otimes \flip_{B;E})}\\
    & + \ptr{\C}{(\idmat_A \otimes \idmat_B) (\tilde{K}_A \otimes \idmat_B) \rho (\idmat_A \otimes \idmat_B)(\idmat_A \otimes \idmat_B)} + \ptr{\C}{ (\idmat_A \otimes iH_B)(\idmat_A \otimes \idmat_B) \rho (\idmat_A \otimes \idmat_B)(\idmat_A \otimes \idmat_B)},
\end{align*}
and apply Lemma \ref{Lem:TranslationLemma} term by term, which yields 
\begin{align*}
    \hat{\kappa}_L(T)(\rho) &:= (\choi^{-1}_{A;B} \circ \kappa_L \circ \choi_{A;B})(T)(\rho)\\
    &= \ptr{E}{B^\dagger U\; (T \otimes \idop_E)(A (\rho \otimes \sigma))} + \frac{1}{2} \ptr{E}{B^\dagger B (T\otimes \idop_E)(\rho\otimes \sigma)}\;  + T(K_A\,\rho) + iH_B\, T(\rho),
\end{align*}
where $U$, $A$ and $B$ are defined as above and $K_A := (\tilde{K}_A)^T$ and $H_B := \tilde{H}_B$. 
An analogous calculation with $\kappa_R(\rho) := \rho K^\dagger$ and $\hat{\kappa}_R(T) := (\choi^{-1}_{A;B} \circ \kappa_R \circ \choi_{A;B})(T)$ finishes the proof of the first part, since the claim about the dimension of $\mathcal{H}_E$ follows form the corresponding statements in Theorem \ref{Thm:SemicausalInfiniteDimMainResult}.

To prove the second part, first remember that we have observed above that  Theorem \ref{Thm:RelationQunatumSemicausalSuperchannels} implies that $L$ is Schrödinger $B \not\to A$ semicausal, with $\ptr{B}{L(\rho)} = L^A(\ptr{B}{\rho})$. Furthermore, if we write $S_t = \choi_{A;B} \circ \hat{S}_t \circ \choi^{-1}_{A;B}$, then Theorem \ref{Thm:RelationQunatumSemicausalSuperchannels} implies that $S_t$ is Schrödinger $B \not\to A$ semicausal for all $t \geq 0$, with $\ptr{B}{S_t(\rho)} = S_t^A(\ptr{B}{\rho})$ and also $S_t^A(\idmat_A) = \idmat_A$ holds. Differentiating that expression at $t = 0$ yields the equivalent condition $L^A(\idmat_A) = 0$. So, our goal is to incorporate the last condition into the form of \eqref{Eq:GeneralFormInProof}.
To do so, we determine $L^A$ by calculating $\ptr{B}{L(\rho)}$, where $L$ is in the form of \eqref{Eq:GeneralFormInProof}. We obtain $\ptr{B}{L(\rho)} = \ptr{E}{\tilde{A}\,\ptr{B}{\rho}\,\tilde{A}^\dagger} - \tilde{K}_A\ptr{B}{\rho} - \ptr{B}{\rho} \tilde{K}_A^\dagger$. Thus, the condition $L^A(\idmat) = 0$ holds if and only if $\ptr{E}{\tilde{A}\tilde{A}^\dagger} = \tilde{K}_A + \tilde{K}_A^\dagger$. Transposing both sides of this equation and using that the definition of $A$ implies that $\tilde{A} = A^{T_A}\Xi_A$, yields $\left(\ptr{E}{A^{T_A}(\idmat_A \otimes \sigma)(A^\dagger)^{T_A}}\right)^T = K_A + K_A^\dagger$. But the left hand side is, by Lemma \ref{Lem:PartialTransposeCalculation}, equal to $\ptr{\sigma}{A^\dagger A}$. This proves the claim that $\hat{L}$ generates a semigroup of superchannels if and only if $\hat{L}$ is hyper-preselecting and $\ptr{\sigma}{A^\dagger A} = K_A + K_A^\dagger$. Finally,  defining $H_A := \frac{1}{2i}(K_A - K_A^\dagger)$ and a few rearrangements lead to \eqref{Eq:SuperchannelGenerators}.\\
\end{proof}

\section{Conclusion}\label{sct:conclusion}

\paragraph{Summary}
The underlying question of this work was: How can we mathematically characterize the processes that describe the aging of quantum devices? We have argued that, under a Markovianity assumption, such processes can be modeled by continuous semigroups of quantum superchannels. Therefore, the goal of this work was to provide a full characterization of such semigroups of superchannels.

We have derived such a general characterization in terms of the generators of these semigroups. Crucially, we have exploited that superchannels correspond to certain semicausal maps, and that therefore it suffices to characterize generators of semigroups of semicausal maps. We have demonstrated both an efficient procedure for checking whether a given generator is indeed a valid semicausal GKLS generator and a complete characterization of such valid semicausal GKLS generators. The latter is constructive in the sense that it can be used to describe parametrizations of these generators. Aside from the theoretical relevance of these results, they will be valuable in studying properties of these generators numerically. Finally, we have translated these results back to the level of superchannels, thus answering our initial question.

We have also posed and answered the classical counterpart of the above question. I.e., we have characterized the generators semigroups of classical superchannels and of semicausal non-negative maps. These results for the classical case might be of independent interest. From the perspective of quantum information theory, they provide a comparison helpful to understand and interpret the characterizations in the quantum case.

\paragraph{Outlook and open questions}

We conclude by presenting some open questions raised by our work. First, in our proof of the characterization of semicausal GKLS generators, we have described a procedure for constructing a semicausal CP-map associated to such a generator. We believe that this method can be applied to a wide range of problems. Determining the exact scope of this method is currently work in progress. 

Second, there is a wealth of results on the spectral properties of quantum channels and, in particular, semigroups of quantum channels. With the explicit form of generators of semigroups of superchannels now known, we can conduct analogous studies for semigroups of quantum superchannels. Understanding such spectral properties, and potentially how they differ from the properties in the scenario of quantum channels, would in particular lead to a better understanding of the asymptotic behavior of semigroups of superchannels, e.g., w.r.t.~entropy production \cite{PhysRevResearch.3.023096, Gour.2019}, the thermodynamics of quantum channels \cite{PhysRevLett.122.200601} or entanglement-breaking properties \cite{chen2020entanglement}.

A further natural question would be a quantum superchannel analogue of the Markovianity problem: When can a quantum superchannel $\hat{S}$ be written as $e^{\hat{L}}$ for some $\hat{L}$ that generates a semigroup of superchannels? Several works have investigated the Markovianity problem for quantum channels \cite{wolf2008assessing, cubitt2012complexity, cubitt2012extracting, onorati2021fitting} and a divisibility variant of this question, both for quantum channels and for stochastic matrices \cite{wolf2008dividing, bausch2016complexity, caro2021necessary}. It would be interesting to see how these results translate to quantum or classical superchannels. Similarly, we can now ask questions of reachability along Markovian paths. 
Yet another question aiming at understanding Markovianity: If we consider master equations arising from a Markovianity assumption on the underlying process formalized not via semigroups of channels, but instead via semigroups of superchannels, what are the associated classes of (time-dependent) generators and corresponding CPTP evolutions?

Two related directions, both of which will lead to a better understanding of Markovian structures in higher order quantum operations, are: Support our mathematical characterization of the generators of semigroups of superchannels by a physical interpretation, similar to the Monte Carlo wave function interpretation of Lindblad generators of quantum channels. And extend our characterization from superchannels to general higher order maps.

This work has focused on generators of general semigroups of superchannels, without further restrictions. For quantum channels and their Lindblad generators, there exists a well developed theory of locality, at the center of which are Lieb-Robinson bounds \cite{nachtergaele2019quasi}. If we put locality restrictions on generators of superchannels, how do these translate to the generated superchannels?

Finally, an important conceptual direction for future work is to identify further applications of our theory of dynamical semigroups of superchannels. 
In the introduction we gave a physical meaning to semigroups of superchannels by relating them to the decay process of quantum devices. 
This, however, is only one possible interpretation. For example, semigroups of superchannels might also describe a manufacturing process, where a quantum device is created layer-by-layer. We hope that other use-cases will be found in the future.


\section*{Acknowledgments}
Both M.C.C.~and M.H.~thank Michael M.~Wolf for insightful discussions about the contents of this paper. We also thank Li Gao, Lisa Hänggli, Robert König, and Farzin Salek for helpful suggestions for improving the presentation.
Also, M.C.C.~and M.H.~thank the anonymous reviewers from TQC 2022 and from the Journal of Mathematical Physics for their constructive criticism.
M.H.~was supported by the Bavarian excellence network ENB via the International PhD Programme of Excellence \emph{Exploring Quantum Matter} (EXQM).
M.C.C.~gratefully acknowledges support from the TopMath Graduate Center of the TUM Graduate School at the Technische Universität München, Germany, from the TopMath Program at the Elite Network of Bavaria, and from the German Academic Scholarship Foundation (Studienstiftung des deutschen Volkes).

\appendix

\section{Proof of Lemmas \ref{Lem:TranslationLemma} and \ref{Lem:PartialTransposeCalculation}} \label{TranslatingSupermapsAppendix}

In this appendix we provide a complete proof of Lemmas \ref{Lem:TranslationLemma} and \ref{Lem:PartialTransposeCalculation}. 

\begin{lem} \emph{(Restatement of Lemma \ref{Lem:TranslationLemma})}
Let $S : \blt(\mathcal{H}_A \otimes \mathcal{H}_B) \rightarrow \blt(\mathcal{H}_A \otimes \mathcal{H}_B)$ be given by
\begin{align*}
    S(X) = \ptr{E}{(\idmat_A \otimes L_B)(L_A \otimes \idmat_B) X (R_A^\dagger\otimes \idmat_B)(\idmat_A \otimes R_B^\dagger)},
\end{align*}
with Hilbert spaces $\mathcal{H}_C$ and $\mathcal{H}_E$, operators $L_A, R_A \in \blt(\mathcal{H}_A; \mathcal{H}_A \otimes \mathcal{H}_C)$ and $L_B, R_B \in \blt(\mathcal{H}_C \otimes \mathcal{H}_B; \mathcal{H}_B \otimes \mathcal{H}_E)$. Then, for $T \in \blt(\blt(\mathcal{H}_A); \blt(\mathcal{H}_B))$ and $\rho \in \blt(\mathcal{H}_A)$, 
\begin{align*}
    \left[\choi_{A;B}^{-1} \circ S \circ \choi_{A;B}\right](T)(\rho) = \ptr{E}{V_L (T\otimes \idop_C)\left(W_L \rho W_R^\dagger \right) V_R^\dagger},
\end{align*}
with $V_L = L_B \flip_{B;C}$, $V_R = R_B \flip_{B;C}$; and $W_L = L_A^{T_A}$, $W_R = R_A^{T_A}$. Here, the partial transpose on $\mathcal{H}_A$ is taken w.r.t. the basis used to define the Choi–Jamiołkowski isomorphism.
\end{lem}
\begin{proof}
Let $\{\ket{e_i}\}_i$ be the orthonormal basis of $\mathcal{H}_A$, w.r.t.~which the Choi–Jamiołkowski isomorphism is defined. Let $\{\ket{c_n}\}_n$ be an orthonormal basis of $\mathcal{H}_C$. Then the formal calculation, which is an algebraic version of drawing the corresponding tensor-network pictures, can be executed as follows:
\begin{align*}
    \left[\choi_{A;B}^{-1} \circ S \circ \choi_{A;B}\right](T)(\rho) &= \ptr{A}{(\rho^T \otimes \idmat_B)\, \ptr{E}{(\idmat_A \otimes L_B)(L_A \otimes \idmat_B) \choi_{A;B}(T) (R_A^\dagger \otimes \idmat_B)(\idmat_A \otimes R_B^\dagger)}} \\
    &= \ptr{E}{L_B\, \ptr{A}{(\rho^T \otimes \idmat_C \otimes \idmat_B)(L_A \otimes \idmat_B) \choi_{A;B}(T) (R_A^\dagger \otimes \idmat_B)} R_B^\dagger} \\
    &= \sum_{i, j}\ptr{E}{L_B\, \left( \ptr{A}{(\rho^T \otimes \idmat_C) L_A \ket{e_i} \bra{e_j} R_A^\dagger} \otimes T(\ket{e_i}\bra{e_j} \right) R_B^\dagger} \\
    &= \sum_{i, j, k, m, n} \braket{e_k\,c_n}{\left((\rho^T \otimes \idmat_C) L_A \ket{e_i} \bra{e_j} R_A^\dagger\right)\, e_k\,c_m}\; \ptr{E}{L_B\,  \left(\ket{c_n}\bra{c_m} \otimes T(\ket{e_i}\bra{e_j})\right) R_B^\dagger}\\
    &=  \sum_{i, j, m, n} \braket{e_i}{\left( L_A^T (\rho \otimes \ket{c_n}\bra{c_m}) \overline{R}_A \right)\,e_j}\; \ptr{E}{L_B\,  \left(\ket{c_n}\bra{c_m} \otimes T(\ket{e_i}\bra{e_j})\right) R_B^\dagger} \\
    &= \sum_{m, n} \ptr{E}{L_B\,  \left(\ket{c_n}\bra{c_m} \otimes T\left(L_A^T (\rho \otimes \ket{c_n}\bra{c_m}) \overline{R}_A \right)\right) R_B^\dagger}\\
    &= \ptr{E}{L_B \flip_{B;C}  (T\otimes \idop_C)\left(\left[\sum_n (\idmat_A \otimes \ket{c_n})L_A^T(\idmat_A \otimes \ket{c_n})\right]\, \rho \,  \left[\sum_m (\idmat_A \otimes \ket{c_m})R_A^T(\idmat_A \otimes \ket{c_m})  \right]^\dagger\right) \flip_{B;C} R_B^\dagger} \\
    &= \ptr{E}{V_L (T\otimes \idop_C)\left(W_L \rho W_R^\dagger \right) V_R^\dagger}.
\end{align*}
\end{proof}

\begin{lem}
Let $X \in \blt(\mathcal{H}_A \otimes \mathcal{H}_C; \mathcal{H}_A \otimes \mathcal{H}_B)$, $Y \in \blt(\mathcal{H}_A \otimes \mathcal{H}_B; \mathcal{H}_A \otimes \mathcal{H}_C)$, $\rho \in \trcl(\mathcal{H}_B)$. Then $\ptr{\rho}{XY}^T = \ptr{C}{Y^{T_A}(\idmat_A \otimes \rho)X^{T_A}}$.
\end{lem}
\begin{proof}
Let $\{\ket{a_i}\}_i$ be the orthonormal basis w.r.t. which the transposition is taken. Using the general identity $\tr{M^T} = \tr{M}$, the definition of the trace w.r.t. a trace-class operator and the cyclicity of the trace, we obtain for every $\sigma \in \trcl(\mathcal{H}_A)$,
\begin{align*}
    \tr{\sigma \ptr{\rho}{XY}^T} &= \tr{\sigma^T \ptr{\rho}{XY}} \\
    &= \tr{(\sigma^T \otimes \rho) XY} \\
    &= \sum_{i,j,k} \tr{(\bra{a_i} \otimes \idmat_B) (\sigma^T \otimes \rho) (\ket{a_j}\bra{a_j} \otimes \idmat_B) X (\ket{a_k}\bra{a_k}\otimes \idmat_C) Y (\ket{a_i}\otimes \idmat_B)} \\
    &= \sum_{i,j,k} \tr{(\bra{a_j} \otimes \idmat_B) (\sigma \otimes \rho) (\ket{a_i}\bra{a_k} \otimes \idmat_B) X^{T_A} (\ket{a_j}\bra{a_i}\otimes \idmat_C) Y^{T_A} (\ket{a_k}\otimes \idmat_B)} \\
    &= \sum_k \tr{\rho (\bra{a_k}\otimes \idmat_B) X^{T_A} \left(\left(\sum_{i,j} \braket{a_j}{\sigma \,a_i} \ket{a_i}\bra{a_j} \right) \otimes \idmat_C \right) Y^{T_A} (\ket{a_k} \otimes \idmat_B) } \\
    &= \tr{(\idmat_A \otimes \rho) X^{T_A} (\sigma \otimes \idmat_C) Y^{T_A}} \\
    &= \tr{\sigma\ptr{C}{Y^{T_A} (\idmat_A \otimes \rho) X^{T_A}}}.
\end{align*}
This proves the claim.
\end{proof}

\section{No information without disturbance} \label{Ap:InformationDisturbanceLemma}

Here we prove a `no information without disturbance'-like lemma that yielded a useful interpretation in the main text. 
\begin{lem}
Let $T \in \mathrm{CP}_\sigma(\mathcal{H}_A \otimes \mathcal{H}_B)$ be such that
\begin{align} \label{Eq:SpecialSemicausality}
    T(X_A \otimes \idmat_B) = X_A \otimes \idmat_B, \tag{B1}
\end{align}
for all $X_A \in \blt(\mathcal{H}_A)$. Then $T(X) = (\idmat_A \otimes W^\dagger) (X \otimes \idmat_E) (\idmat_A \otimes W)$, for all $X \in \blt(\mathcal{H}_A \otimes \mathcal{H}_B)$ and some isometry $W \in \blt(\mathcal{H}_B; \mathcal{H}_B \otimes \mathcal{H}_E)$, where $\mathcal{H}_E$ is some Hilbert space. 
\end{lem}
\begin{proof}
This claim follows from the uniqueness of the minimal Stinespring dilation in the same way as the ``semicausal = semilocalizable'' Theorem. Write Eq.~\eqref{Eq:SpecialSemicausality} in Stinespring form as
\begin{align*}
    V^\dagger (X_A \otimes \idmat_B \otimes \idmat_E) V = X_A \otimes \idmat_B,
\end{align*}
for some $V \in \blt(\mathcal{H}_A \otimes \mathcal{H}_B; \mathcal{H}_A \otimes \mathcal{H}_B \otimes \mathcal{H}_C)$. Then $V$ and $\idmat_{AB}$ are the Stinespring operators of the same CP-map ($X_A \mapsto X_A \otimes \idmat_B$) and the latter clearly belongs to a minimal dilation. Thus there exists an isometry $W \in \blt(\mathcal{H}_B; \mathcal{H}_B \otimes \mathcal{H}_E)$ such that $V = (\idmat_A \otimes W) \idmat_{AB}$. This is the claim. 
\end{proof}

Note that the lemma above is just a formulation of the `obvious' fact that if system $A$ undergoes a closed system evolution ($\idop_A$), then there is no interaction with an external system $B$. 

\section{Constructive Approach to Theorem \ref{Thm:SemicausalInfiniteDimMainResult}} \label{Appe:ComputationalConstruction}

In this appendix, we are going to describe in detail, how one can computationally construct the operators $A$, $U$, $B$, $K_A$ and $H_B$ in Theorem \ref{Thm:SemicausalInfiniteDimMainResult}, if the conditions of Lemma \ref{Lem:VerifySemicausality} are met. 

Since it is important for an actual implementation on a computer, let us be very precise about notation. We introduce indexed copies of $\mathcal{H}_A$ and $\mathcal{H}_B$, i.e. $\mathcal{H}_{A_0} = \mathcal{H}_{A_1} = \mathcal{H}_{A_2} = \mathcal{H}_{A}$ and $\mathcal{H}_{B_0} = \mathcal{H}_{B_1} = \mathcal{H}_{B_2} = \mathcal{H}_{B}$. Furthermore, we fix orthonormal bases $\{\ket{a_i}\}_{i = 1}^{d_A}$ and $\{\ket{b_i}\}_{i = 1}^{d_B}$ of $\mathcal{H}_A$ and $\mathcal{H}_B$, respectively. We use the symbol $\Omega$ with some subscript to denote the maximally entangled state on various systems. For example $\ket{\Omega_{A_1; A_2}} := \sum_i \ket{a_i}\otimes \ket{a_i} \in \mathcal{H}_{A_1}\otimes \mathcal{H}_{A_2}$ and $\ket{\Omega_{A_1B_1; A_2B_2}} =  \sum_{i,j} \ket{a_i} \otimes \ket{b_j} \otimes \ket{a_i}\otimes \ket{b_j} \in \mathcal{H}_{A_1} \otimes \mathcal{H}_{B_1} \otimes \mathcal{H}_{A_2} \otimes \mathcal{H}_{B_2}$. We further reserve reserve $P \in \blt(\mathcal{H}_{A_1} \otimes \mathcal{H}_{B_1} \otimes \mathcal{H}_{A_2} \otimes \mathcal{H}_{B_2})$ for the orthogonal projection onto $\mathrm{span}\{\ket{\Omega_{A_1B_1; A_2B_2}}\}$ (i.e. $P = (d_A d_B)^{-1} \ket{\Omega_{A_1B_1; A_2B_2}}\bra{\Omega_{A_1B_1; A_2B_2}}$) and take $P^\bot = \idmat_{A_1B_1A_2B_2} - P$. 

Now, let $\mathfrak{L} \in \blt(\mathcal{H}_{A_1} \otimes \mathcal{H}_{B_1} \otimes \mathcal{H}_{A_2} \otimes \mathcal{H}_{B_2})$ be given as in Lemma \ref{Lem:VerifySemicausality} then we can compute the operators $A$, $U$, $B$, $K_A$ and $H_B$ via the following fifteen steps:

\begin{enumerate}
    \item Compute $\tau = P^\bot \mathfrak{L} P^\bot$.
    \item Compute $V = (\idmat_{A_0B_0} \otimes \sqrt{\tau})(\ket{\Omega_{A_0B_0;A_1B_1}} \otimes \idmat_{A_2B_2})$.
    \item Define $\mathcal{H}_E := \mathcal{H}_{A_1} \otimes \mathcal{H}_{B_1} \otimes \mathcal{H}_{A_2} \otimes \mathcal{H}_{B_2}$, so that $V \in \blt(\mathcal{H}_{A} \otimes \mathcal{H}_{B}; \mathcal{H}_{A} \otimes \mathcal{H}_{B} \otimes \mathcal{H}_E)$. (identification)
    \item Compute $B = \frac{1}{d_A} \ptr{A}{V}$.
    \item Compute $V_{sc} = V - \idmat_A \otimes B$.
    \item Compute $\tau_{sc} = (\idmat_{A_1B_1} \otimes V_{sc})^\dagger (\ket{\Omega_{A_1B_1;AB}}\bra{\Omega_{A_1B_1;AB}} \otimes \idmat_E) (\idmat_{A_1B_1} \otimes V_{sc}) \in \blt(\mathcal{H}_{A_1} \otimes \mathcal{H}_{B_1} \otimes \mathcal{H}_{A} \otimes \mathcal{H}_{B})$.
    \item Choose any unit vector $\ket{\beta} \in \mathcal{H}_{B}$.
    \item Compute $\tau^A_{sc} = (\idmat_{A_1A_2} \otimes \bra{\beta}) \ptr{B_1}{\tau_{sc}} (\idmat_{A_1A_2} \otimes \ket{\beta})$.
    \item Compute $\mathcal{H}_F = \mathrm{range}(\sqrt{\tau^A_{sc}})$, so that $\sqrt{\tau_{sc}^A} \in \blt(\mathcal{H}_{A_1}\otimes \mathcal{H}_{A_2}; \mathcal{H}_F)$ is surjective.
    \item Compute $A = (\idmat_{A_0} \otimes \sqrt{\tau^A_{sc}})(\ket{\Omega_{A_0;A_1}} \otimes \idmat_{A_2})$.
    \item Compute $U$ as \textit{the} solution of the system of linear equations $\mathcal{M}(U) = V_{sc}$, where the $d_A^2d_B^2d_E \times d_F d_B^2d_E$-matrix $\mathcal{M} : \blt(\mathcal{H}_F \otimes \mathcal{H}_B; \mathcal{H}_B \otimes \mathcal{H}_E) \rightarrow \blt(\mathcal{H}_A \otimes \mathcal{H}_B; \mathcal{H}_A \otimes \mathcal{H}_B \otimes \mathcal{H}_E)$ is defined by $\mathcal{M}(U) = (\idmat_A \otimes U)(A \otimes \idmat_B)$. Clearly, we must first represent $\mathcal{M}$ w.r.t. some basis. 
    \item Compute $K = -\ptr{A_1B_1}{P\mathfrak{L}P^\bot + \frac{1}{2} \tr{P\mathfrak{L}} P}$, where we identify $\mathcal{H}_{A_2}\otimes \mathcal{H}_{B_2} = \mathcal{H}_{A}\otimes \mathcal{H}_{B}$ so that $K \in \blt(\mathcal{H}_A \otimes \mathcal{H}_B)$
    \item Compute $K_{sc} = K - (\idmat_A \otimes B^\dagger)V_{sc} - \frac{1}{2} \idmat_A \otimes B^\dagger B$.
    \item Compute $K_A = \frac{1}{d_B}\ptr{B}{K_{sc}}$.
    \item Compute $H_B = \frac{-i}{d_A} \ptr{A}{K_{sc} - K_A \otimes \idmat_B}$. 
\end{enumerate}
Note that the procedure above computes an isometry $U \in \blt(\mathcal{H}_F \otimes \mathcal{H}_B; \mathcal{H}_B \otimes \mathcal{H}_E)$ which can then be extended to a unitary, if necessary. In that case we also have to embed $\mathcal{H}_F$ into $\mathcal{H}_E$ and redefine $A$ accordingly. More precisely, we need to execute the following additional steps
\begin{enumerate}
    \setcounter{enumi}{15}
    \item Compute $\idmat_{F \rightarrow E} = \idmat_{A_1} \otimes \ket{\beta}_{B_1} \otimes \idmat_{A_2} \otimes \ket{\beta}_{B_2}$
    \item Redefine $A \leftarrow (\idmat_{A_0} \otimes \idmat_{F \rightarrow E})A$
    \item Extend $U$ via the following steps :
    \begin{enumerate}
        \item Compute $\hat{U} = U(\idmat_{F \rightarrow E}^\dagger \otimes \idmat_B)$. 
        \item Compute an orthonormal basis $\{\ket{f^\bot_i}\}_{i = 1}^N$ of $\mathrm{range}(\idmat_{EB} - \hat{U}^\dagger \hat{U})$.
        \item Compute an orthonormal basis $\{\ket{r^\bot_i}\}_{i = 1}^N$ of $\mathrm{range}(\idmat_{BE} - \hat{U} \hat{U}^\dagger)$.
        \item Redefine $U \leftarrow \hat{U} + \sum_{i = 1}^N \ket{r_i^\bot}\bra{f_i^\bot}$.
    \end{enumerate}
\end{enumerate}

Let us comment on why the steps above give the right result. In general, we have $$\mathfrak{L} = P^\bot \mathfrak{L} P^\bot + P\mathfrak{L}P^\bot + P^\bot \mathfrak{L} P + P\mathfrak{L}P = \tau + (P\mathfrak{L}P^\bot + \frac{1}{2} \tr{P\mathfrak{L}} P) + (P^\bot\mathfrak{L}P + \frac{1}{2} \tr{P\mathfrak{L}} P).$$ Thus the maps $\Phi$ and $K$ appearing in the GKLS-form in Theorem \ref{Thm:SemicausalInfiniteDimMainResult} can be extracted from the previous equation by applying the inverse of the Choi–Jamiołkowski isomorphism. One readily obtains $\Phi = \choi_{AB;AB}^{-1} \circ \tau$ and $K = -\ptr{A_1B_1}{P\mathfrak{L}P^\bot + \frac{1}{2} \tr{P\mathfrak{L}} P}$.  

\begin{itemize}
    \item Step 2 computes the Stinespring dilation of a CP-map whose Choi–Jamiołkowski operator is $\tau$. A direct computation shows that $\tau = (\idmat_{A_1B_1} \otimes V)^\dagger (\ket{\Omega_{A_1B_1;A_2B_2}}\bra{\Omega_{A_1B_1;A_2B_2}} \otimes \idmat_E) (\idmat_{A_1B_1} \otimes V)$.
    \item Step 4 computes the operator $B$ in the representation. In the proof of Theorem \ref{Thm:SemicausalInfiniteDimMainResult}, $B$ was obtained from $\tilde{B}$, which in turn was obtained from $V$ and Lemma \ref{KeyLemmaInf}. In the finite-dimensional setting, Lemma \ref{KeyLemmaInf} constructs $B$ exactly as is written down above.
    \item Steps 6,7 and 8 define $\tau_{sc}$ as the Choi–Jamiołkowski operator of a CP-map with Stinespring operator $V_{sc}$. Thus, according to the proof of Theorem \ref{Thm:SemicausalInfiniteDimMainResult}, $\tau$ is the Choi–Jamiołkowski of a Heisenberg $B \not \to A$ semicausal map. And semicausality is expressed on the level of Choi–Jamiołkowski operators by the existence of an operator $\tau_{sc}^A$, such that $\ptr{B_1}{\tau_{sc}} = \tau_{sc}^A \otimes \idmat_{B_2}$ (compare with the proof of Lemma \ref{Lem:VerifySemicausality}). Using this relation makes clear, that step 8 extracts $\tau_{sc}^A$ from $\tau_{sc}$ and that the result is independent of the choice of $\ket{\beta}$. 
    \item Step 10 defines $A$ as the Stinespring dilation of the (reduced) map whose Choi–Jamiołkowski operator is $\tau_{sc}^A$. The dilation constructed in this way is minimal. This is exactly the way in which the operator $W = A$ was constructed in the proof of Theorem \ref{Thm:SemicausalInfiniteDimMainResult}. 
    \item Step 11 obtains $U$ by solving the defining relation (for $\tilde{U}$) in the proof of Theorem \ref{Thm:SemicausalInfiniteDimMainResult}. One might wonder why the solution to this system of equations is unique (even though $\mathcal{M}$ is not a square matrix). Uniqueness follows from the minimality of $A \otimes \idmat_B$, that is, vectors the form $(X_A \otimes \idmat_{FB})(A \otimes \idmat_B)\ket{\psi}$ span $\mathcal{H}_A \otimes \mathcal{H}_B \otimes \mathcal{H}_E$. In detail, if $U$ and $U^\prime$ satisfy $\mathcal{M}(U) = \mathcal{M}(U^\prime)$, then $0 = (\idmat_A \otimes (U-U^\prime))(A \otimes \idmat_B)$ and hence $0 = (\idmat_A \otimes (U-U^\prime))(X_A \otimes \idmat_{FB})(A \otimes \idmat_B)\ket{\psi}$. By linearity, this implies $U-U^\prime = 0$. 
    \item Step 12 computes the operator $K$ in the GKLS-form according to the discussion above. 
    \item Step 13 defines an operator $K_{sc}$, which according the statement of Theorem \ref{Thm:SemicausalInfiniteDimMainResult} and also due to the discussion below Eq. \eqref{Eq:EqToFindFormOfKSc} is of the form $K_A \otimes \idmat_B + \idmat_A \otimes iH_B$.
    \item Steps 14 and 15 extract $K_A$ and $H_B$ from $K_{sc}$. Note that such a decomposition is not unique, since for any $\lambda \in \R$, the transformation $K_A \rightarrow K_A + i\lambda \idmat_A$, $H_B \rightarrow H_B - \lambda \idmat_B$ leaves $K_{sc}$ invariant. This transformation, however, allows us to choose $H_B$ traceless. In that case steps 14 and 15 determine $K_A$ and $H_B$. 
\end{itemize}

\section{Constructive Approach to Theorem \ref{Thm:CharacterizationGeneratorsSuperchannels}} \label{Appe:ComputationalConstruction2}

In this appendix, we are going to describe in detail, how one can computationally the operators $A$, $U$, $B$, $H_A$ and $H_B$ in Theorem \ref{Thm:CharacterizationGeneratorsSuperchannels}, if the conditions of Lemma \ref{thm:checkforGeneratorSuperchannel} are met. We use the notation from Appendix \ref{Appe:ComputationalConstruction}. 

Given the operator $\hat{\mathfrak{L}} \in \blt(\mathcal{H}_{A_1} \otimes \mathcal{H}_{B_1} \otimes \mathcal{H}_{A_2} \otimes \mathcal{H}_{B_2})$ be given as in Lemma \ref{thm:checkforGeneratorSuperchannel} then we can compute the operators $A$, $U$, $B$, $H_A$ and $H_B$ via the following eight steps:

\begin{enumerate}
    \item Apply the steps 1-18 in the protocol in Appendix \ref{Appe:ComputationalConstruction} to $\hat{\mathfrak{L}}$. This yields $\mathcal{H}_E = \mathcal{H}_{A_1} \otimes \mathcal{H}_{B_1} \otimes \mathcal{H}_{A_2} \otimes \mathcal{H}_{B_2}$, $\tilde{A} \in \blt(\mathcal{H}_{A_2}; \mathcal{H}_{A_0} \otimes \mathcal{H}_E)$, $\tilde{U} \in \blt(\mathcal{H}_E \otimes \mathcal{H}_B; \mathcal{H}_B \otimes \mathcal{H}_E)$, $\tilde{K}_A \in \blt(\mathcal{H}_A)$ and $\tilde{H}_B \in \blt(\mathcal{H}_B)$.
    \item Choose any unit vector $\ket{\xi} \in \mathcal{H}_E$.
    \item Compute $\sigma = \ket{\xi}\bra{\xi}$.
    \item Compute $A = (\idmat_{A_{-1}} \otimes \idmat_E \otimes \bra{\Omega_{A_0;A_3}})(\idmat_{A_{-1}} \otimes \flip_{A_0; E}\tilde{A} \otimes \idmat_{A_3})(\ket{\Omega_{A_{-1}; A_2}} \otimes \idmat_{A_3} \otimes \bra{\xi})$.
    \item Compute $B = \tilde{B}(\idmat_B \otimes \bra{\xi})$.
    \item Compute $U = \tilde{U}\flip_{B;E}$.
    \item Set $H_B = \tilde{H}_B$.
    \item Calculate $H_A = \frac{1}{2i}(\tilde{K}_A^T - \tilde{K}_A^{\dagger T})$, where the transposition is w.r.t. the $\{\ket{a_i}\}$ basis defined in Appendix \ref{Appe:ComputationalConstruction}.
\end{enumerate}

Let us comment on why the steps above yield the right result:

\begin{itemize}
    \item Step 1 can be executed, since the assumptions of Lemma \ref{Lem:VerifySemicausality} are the first two assumptions in Lemma \ref{thm:checkforGeneratorSuperchannel}. 
    \item Steps 2, 3 define $\sigma$ as in the proof of Theorem \ref{Thm:CharacterizationGeneratorsSuperchannels}.
    \item Step 4 is a more explicit expression for $\tilde{A}^{T_A}\Xi_A^\dagger$ in the proof of Theorem \ref{Thm:CharacterizationGeneratorsSuperchannels}.
    \item Steps 5, 6 and 7 are exactly the definitions of $B$, $U$ and $\mathcal{H}_B$ in the proof of Theorem \ref{Thm:CharacterizationGeneratorsSuperchannels}.
    \item For step 8, we note that the condition $\ptr{A_1}{\hat{\mathfrak{L}}^A} = 0$ implies $L^A(\idmat) = 0$ so that we can follow the last few sentences in the proof of Theorem \ref{Thm:CharacterizationGeneratorsSuperchannels}.
\end{itemize}

\bibliography{aipsamp}

\end{document}